\documentclass[11pt,a4paper]{article}

\usepackage[T1]{fontenc}
\usepackage[utf8]{inputenc}
\usepackage{lmodern}
\usepackage{microtype}
\usepackage[margin=1in]{geometry}
\usepackage{amsmath,amssymb,amsthm}
\usepackage{mathtools}
\usepackage{enumitem}
\usepackage[round]{natbib}
\usepackage{xcolor}
\usepackage[colorlinks=true,linkcolor=blue!60!black,citecolor=blue!60!black,urlcolor=blue!60!black]{hyperref}

\theoremstyle{plain}
\newtheorem{theorem}{Theorem}
\newtheorem{lemma}{Lemma}
\newtheorem{proposition}{Proposition}
\newtheorem{corollary}{Corollary}
\theoremstyle{definition}
\newtheorem{definition}{Definition}
\newtheorem{example}{Example}
\theoremstyle{remark}
\newtheorem{remark}{Remark}

\newcommand{\FOMC}{\mathrm{FOMC}}
\newcommand{\UFOMC}{\mathrm{UFOMC}}
\newcommand{\Mod}{\mathrm{Mod}}
\newcommand{\UMod}{\mathrm{UMod}}
\DeclareMathOperator{\Aut}{Aut}
\newcommand{\dom}{\mathrm{dom}}
\newcommand{\PPone}{\#\mathrm{P}_1}
\newcommand{\acc}{\mathrm{acc}}
\newcommand{\Bij}{\mathrm{Bij}}
\newcommand{\PERM}{\operatorname{PERM}}
\newcommand{\CycSucc}{\mathrm{CycSucc}}
\newcommand{\NewLabel}[2]{\operatorname{newlabel}_{#1,#2}}
\newcommand{\prof}{\mathrm{prof}}
\newcommand{\Qc}{\mathcal{Q}}
\newcommand{\FOk}[1]{\mathrm{FO}^{#1}}
\newcommand{\FOke}[1]{\mathrm{FO}^{#1}_{=}}
\newcommand{\FOkef}[1]{\mathrm{FO}^{#1}_{=}[f]}
\newcommand{\Ck}[1]{\mathrm{C}^{#1}}
\newcommand{\Cke}[1]{\mathrm{C}^{#1}_{=}}
\newcommand{\ind}[1]{\mathbf{1}_{#1}}
\newcommand{\FOkefg}[1]{\mathrm{FO}^{#1}_{=}[f,g]}
\newcommand{\blank}{\sqcup}
\newcommand{\cB}{\mathsf B}
\newcommand{\cR}{\mathsf R}
\newcommand{\cM}{\mathsf M}
\newcommand{\NF}{\mathsf{NF}}
\newcommand{\Final}{\mathsf{Final}}
\newcommand{\First}{\mathsf{First}}
\newcommand{\Read}{\mathsf{Read}}
\newcommand{\Write}{\mathsf{Write}}
\newcommand{\Move}{\mathsf{Move}}
\newcommand{\Rvis}{\mathsf R}

\title{Unary Functions, Automorphisms, and\\ Unlabeled First-Order Model Counting}
\author{Ondřej Kuželka\\
Faculty of Electrical Engineering, Czech Technical University in Prague\\
Prague, Czech Republic}
\date{}

\begin{document}

\maketitle

\begin{abstract}
Every fixed first-order sentence $\varphi$ determines an enumerative sequence $n\mapsto \FOMC(\varphi,n)$, counting its models on the labeled domain $[n]$. We study the complexity of these sequences when logical specifications may use genuine unary function symbols and hence nested terms $x,f(x),f^2(x),\ldots$.

We first prove that, for every fixed sentence $\varphi\in \Cke{1}[f]$, with one unary function and an arbitrary finite relational vocabulary, $\FOMC(\varphi,n)$ is computable in time polynomial in $n$. By contrast, permitting either a second variable or a second unary function already yields hardness. Without counting quantifiers, there is a fixed sentence in $\FOkef{2}$ whose model-counting function is $\PPone$-complete. With one variable and two unary functions, there is a fixed constant-free universal sentence in $\FOkefg{1}$, using only unary predicates besides $f$ and $g$, whose model-counting function is again $\PPone$-complete.

We also relate labeled and unlabeled enumeration exactly. For every relational sentence $\varphi$, we construct an extension $\varphi_{\mathrm{aut}}$ in which a unary function records an automorphism and
\[
\FOMC(\varphi_{\mathrm{aut}},n)
=
n!\,\UFOMC(\varphi,n),
\]
where $\UFOMC(\varphi,n)$ denotes the number of $n$-element models of $\varphi$ up to isomorphism.
Thus automorphism marking gives a one-query exact reduction from unlabeled to labeled model counting at the same domain size. Eliminating the auxiliary function yields single-query reductions to purely relational model counting. If the source sentence uses at most $k$ variables and its vocabulary has maximum relation arity $r$, the target uses at most $\max\{k,r+1,3\}$ variables in first-order logic with equality, and at most $\max\{k,r+1,2\}$ variables in counting logic.
\end{abstract}

\section{Introduction}\label{sec:intro}

What constitutes a satisfactory solution to an enumeration problem? Suppose that an enumeration problem gives an integer sequence $(a_n)_{n\geq 0}$, where $a_n$ is the number of objects of size $n$. A closed formula is one possibility, but from an algorithmic point of view a recurrence, generating function, or other representation is useful only as long as it permits the desired numbers to be computed efficiently. \citet{wilf1982} proposed judging solutions to enumeration problems partly by their computational complexity. In particular, one of his criteria asks whether $a_n$ can be computed in time polynomial in $n$; \citet{pak2018} refers to algorithms of this kind as Wilfian formulas of type (W1). This viewpoint belongs to a broader complexity-theoretic study of enumerative combinatorics, in which the complexity of exact computation is treated as an intrinsic aspect of an enumeration problem (\citealp{pak2018,pak2024}; see also \citealp{klazar2018}).

First-order model counting provides a uniform way to ask such questions for logically specified classes of finite structures. A fixed first-order sentence $\varphi$ describes a family of structures independently of the domain size, and
\[
\FOMC(\varphi,n)
\]
is the number of its models on the labeled domain
\[
[n]=\{1,\ldots,n\}.
\]
Thus every fixed sentence determines an enumerative sequence. In the model-counting literature, a sentence or fragment is called \emph{domain-liftable} when this count, or more generally its symmetric weighted version, can be computed in time polynomial in $n$. For ordinary model counting, domain liftability therefore says that every fixed specification in the fragment gives an enumerative sequence computable in time polynomial in $n$.

A tractability theorem for a logical fragment therefore gives a single sufficient condition for polynomial-time exact enumeration across all fixed specifications expressible in that fragment. A standard illustration is supplied by regular graphs. For fixed $k$, the sentence
\[
\forall x\,\neg E(x,x)
\;\wedge\;
\forall x\forall y\,(E(x,y)\Rightarrow E(y,x))
\;\wedge\;
\forall x\,\exists^{=k}y\,E(x,y)
\]
has as its models precisely the labeled simple $k$-regular graphs. It belongs to the two-variable logic with counting quantifiers $\Ck{2}$, whose weighted model-counting problem is polynomial-time in the domain size for every fixed sentence \citep{kuzelka2021}. Adding a fixed number $l$ of unary color predicates, requiring every vertex to receive exactly one color, and forbidding monochromatic edges similarly counts pairs consisting of a labeled $k$-regular graph and a proper $l$-coloring. For every fixed $k$, the counting sequence for labeled $k$-regular graphs was already known to be P-recursive \citep{goulden1986,gessel1990}, and hence computable in time polynomial in $n$; see also \citep[Section~1.3]{pak2018}. The regular-graph example illustrates the role of the logical result rather than its scope: the $\Ck{2}$ tractability theorem applies to every fixed $\Ck{2}$ specification, including structurally quite different relational classes and combinations of constraints expressible in the fragment.

This perspective has increasingly brought first-order model counting into contact with enumerative combinatorics. The original tractability results established polynomial-time symmetric weighted model counting for every fixed sentence of the two-variable fragment $\FOk{2}$ \citep{vandenbroeck2011,vandenbroeck2014,beame2015}, and \citet{kuzelka2021} extended this to $\Ck{2}$. Subsequent work showed that polynomial-time counting survives the addition of structural requirements such as tree, forest, connectedness, directed acyclicity, and linear order axioms \citep{vanbremen2021,toth2023,malhotra2025}. Some of these arguments import classical enumerative tools directly: the tree result, for example, uses the matrix-tree theorem, while connectedness and acyclicity are handled through adaptations of familiar combinatorial recurrences.

More recently, \citet{kuang2026} developed model-counting analogues of graph polynomials whose coefficients encode structural information about the models of a sentence. Their weak- and strong-connectedness polynomials provide a common framework for several previously studied structural axioms and yield further tractable counting problems. Classical graph polynomials also arise within this framework: the Tutte polynomial is obtained as a special case of the weak-connectedness polynomial, while directed chromatic polynomials arise from the strong-connectedness construction.

The present paper asks what happens when this framework is extended from relational structures to structures with genuine unary function symbols. There is an important distinction here. Functionality itself is not new to relational model counting: a binary relation $F$ can be required to be the graph of a total function by the $\Ck{2}$ sentence
\[
\forall x\,\exists^{=1}y\,F(x,y),
\]
and functionality constraints already occur in earlier tractability results \citep{kuusisto2018,kuzelka2021}. A unary function symbol, however, changes the term language. From a single variable $x$ one may form
\[
x,\ f(x),\ f^2(x),\ldots,
\]
compare different iterates for equality, and place several points of one forward orbit into the same relation atom, as in
\[
R(x,f(x),f^2(x)).
\]
This information is not available in the same syntactic fragment merely by representing $f$ through a functional binary relation. In particular, equality atoms can compare nonadjacent points of a single forward orbit while using only one variable. Nested function terms therefore give the one-variable language genuine, though bounded, access to orbit structure.

For every fixed sentence
\[
\varphi\in \Cke{1}[f],
\]
over one unary function symbol and an arbitrary finite relational vocabulary, we prove that
\[
n\longmapsto \FOMC(\varphi,n)
\]
is computable in time polynomial in $n$. Thus every fixed specification in $\Cke{1}[f]$ admits a polynomial-time exact enumeration algorithm in the domain size. The fragment is nevertheless quite small: the formula has only one variable and one unary function, and each fixed sentence can inspect only a bounded portion of any forward orbit.

The main complexity message is that tractability for one-variable counting logic with a single unary function does not extend much further. If a second variable is allowed, hardness already occurs without counting quantifiers: we construct a fixed sentence in $\FOkef{2}$ whose model-counting function is $\PPone$-complete. Keeping one variable but allowing a second unary function is equally powerful: there is a fixed constant-free universal sentence in $\FOkefg{1}$, using only unary predicates besides $f$ and $g$, whose model-counting function is again $\PPone$-complete. Here $\PPone$ is the unary-input version of $\#\mathrm{P}$. Giving the domain size as $1^n$ is precisely the complexity-theoretic formalization of asking for computation in time polynomial in the enumerative parameter $n$, rather than polynomial in the binary length of $n$. Thus no analogue of our polynomial-time theorem can hold throughout either of these richer languages unless $\PPone$-complete counting functions admit polynomial-time algorithms.

The two hardness proofs encode the same fixed linear-time nondeterministic computation by different means; see Section~\ref{sec:hardness}. Taken together, the results leave a narrow tractable region:
\[
\Cke{1}[f]\text{-FOMC is polynomial-time computable,}
\]
whereas $\FOkef{2}$ and $\FOkefg{1}$ already contain fixed sentences with $\PPone$-complete model-counting functions.

Our second theme concerns the distinction between labeled and unlabeled enumeration. First-order model counting is inherently labeled: the domain is the fixed set $[n]$. Many classical enumeration problems instead ask for isomorphism classes. The relation between the two counts is obstructed by automorphisms. If $A$ is an $n$-element structure, then its isomorphism class has
\[
\frac{n!}{|\Aut(A)|}
\]
labeled copies, so simply dividing a labeled count by $n!$ is valid only when every model is rigid.

We show that this varying automorphism factor can itself be absorbed into a logical specification. Given a relational sentence $\varphi$, introduce a unary function $f$ and require it to be a permutation preserving every relation. The models of the resulting sentence $\varphi_{\mathrm{aut}}$ are precisely the pairs
\[
(A,\pi)
\]
where $A$ is a labeled model of $\varphi$ and $\pi\in\Aut(A)$. Consequently each isomorphism class contributes
\[
\frac{n!}{|\Aut(A)|}
\,|\Aut(A)|
=
n!
\]
expanded labeled structures, independently of the size of its automorphism group. Hence
\[
\FOMC(\varphi_{\mathrm{aut}},n)
=
n!\,\UFOMC(\varphi,n).
\]
At the numerical level this is the familiar Burnside/orbit-stabilizer cancellation. The point of the construction is that the permutation is made part of the structure, turning the Burnside fixed-point sum into an ordinary labeled model count for one fixed sentence.

The auxiliary function can also be eliminated. Its graph can be represented by a binary relation, and its action on an $r$-ary relation can be enforced one argument at a time. If the source sentence uses at most $k$ variables and the vocabulary has maximum relation arity $r$, put
\[
q_{\mathrm{FO}}=\max\{k,r+1,3\},
\qquad
q_{\mathrm C}=\max\{k,r+1,2\}.
\]
We obtain single-query polynomial-time reductions from unlabeled $\FOke{k}$ model counting to labeled $\FOke{q_{\mathrm{FO}}}$ model counting, and from unlabeled $\Ck{k}$ model counting to labeled $\Ck{q_{\mathrm C}}$ model counting. In particular, when $r<k$, the counting-logic reduction stays within $\Ck{k}$, and the first-order reduction stays within $\FOke{k}$ whenever $k\geq3$. The oracle is queried once, at the same domain size, and the answer is followed only by exact division by $n!$. For unary relational vocabularies the construction stays inside $\Ck{2}$; together with the polynomial-time $\Ck{2}$ model-counting theorem, this gives polynomial-time unlabeled model counting for every fixed $\Ck{1}$ sentence.

This reduction should not be confused with a canonical-form or symmetry-breaking construction. That distinction is important from the complexity-theoretic viewpoint emphasized by \citet{pak2018,pak2024}. \citet{wilf1982} conjectured that the number $u_n$ of unlabeled graphs on $n$ vertices cannot be computed in time polynomial in $n$. \citet{pak2024} asks the different, and still open, question whether the sequence $(u_n)$ even has a combinatorial interpretation in the complexity-theoretic sense, equivalently, whether it belongs to $\#\mathrm{P}$ when $n$ is given in unary. The obstacle is precisely that an unlabeled count is an orbit count: there is no evident polynomially checkable procedure that selects one representative from each orbit. For some special classes, such as plane triangulations, polynomially bounded and effectively computable automorphism groups permit explicit symmetry breaking and hence $\#\mathrm{P}$-membership \citep{pak2024}.

Our automorphism-marking construction addresses a different question. It does not select representatives and therefore does not resolve the accepting-path complexity of unlabeled graphs. Instead, it gives a one-query exact reduction from the unlabeled count to an ordinary labeled model count. The oracle is queried at the same domain size, and the answer is divided by the known factor $n!$. The reduction works by changing the objects being counted: each labeled representative is paired with one of its automorphisms, so the varying stabilizer factor cancels.

The paper is organized as follows. Section~\ref{sec:background} fixes the logical, counting, and reduction conventions. Section~\ref{sec:onevar} gives the polynomial-time model-counting algorithm for $\Cke{1}[f]$. Section~\ref{sec:unlabelled-counting} develops automorphism marking and its purely relational form. Section~\ref{sec:hardness} proves the two hardness results obtained by allowing, respectively, a second variable and a second unary function. Section~\ref{sec:discussion} summarizes the resulting narrow tractability boundary and its enumerative consequences, and Section~\ref{sec:related} places the results in the surrounding literature. The longer proofs are collected in Section~\ref{sec:proofs}.

We work throughout with ordinary, unweighted model counting in order to keep the presentation close to the enumerative setting. The constructions also admit symmetric weighted variants; the required modifications are discussed in Section~\ref{sec:discussion}.

\section{Background}\label{sec:background}

In this section we describe the background material and conventions needed for our technical results.

\subsection{Structures and Finite-Variable Fragments}

We assume that the reader is familiar with first-order logic and we only briefly set up the notation used throughout the paper. All vocabularies are finite. Relation and function symbols have fixed arities, and function symbols are always interpreted as total functions. We do not allow constant symbols. Unless stated otherwise, structures are finite and nonempty, and labeled $n$-element structures have the domain $[n] = \{1, \dots, n\}$. We assume $n \geq 1$ throughout.

We write $\FOk{k}$ for the fragment of first-order logic that uses at most $k$ variable symbols, which may be reused under quantification. The presence of equality is indicated by a subscript, and $[f]$ indicates the presence of one unary function symbol. Thus $\FOkef{2}$ is two-variable first-order logic with equality and one unary function. We write $\FOkefg{1}$ for the corresponding one-variable fragment with equality and two unary function symbols $f,g$. We write $\Ck{k}$ for the usual $k$-variable logic with the counting quantifiers $\exists^{=m}$ and $\exists^{\geq m}$, where $m$ is a constant appearing in the formula. For $k\geq2$, equality does not change the expressive power of this logic, so we use the standard notation $\Ck{k}$ without a separate equality subscript. In the one-variable functional setting the distinction is important: equalities between nested terms, such as $f^i(x)=f^j(x)$, are an essential part of the tractable fragment. We therefore write $\Cke{1}[f]$ for one-variable counting logic with equality and one unary function symbol.

The \emph{1-type} of an element $a$ in a structure $A$ is the set of all function-free atoms in the single variable $x$ that are satisfied by $a$ in $A$; besides the unary atoms $P(x)$, this includes the reflexive atoms of the higher-arity predicates, such as $R(x, x)$. Intuitively, one may imagine the finitely many possible 1-types as colors, so that the 1-type of an element simply records its color. Over a vocabulary that contains only the unary relation symbols $P_1, \dots, P_s$---as in the monadic core and in the target of the reduction in Section~\ref{sec:proofonevar}---a 1-type is simply a subset of $\{P_1, \dots, P_s\}$, and there are $2^s$ of them.

\subsection{Labeled and Unlabeled Model Counting}

Next we define the two model-counting problems that we study in this paper.

\begin{definition}[Labeled model count]\label{def:fomc}
For a sentence $\varphi$, let
\[
\Mod_n(\varphi) = \{ A : \dom(A) = [n] \text{ and } A \models \varphi \}.
\]
The \emph{labeled model count} of $\varphi$ is $\FOMC(\varphi, n) = |{\Mod_n(\varphi)}|$.
\end{definition}

The sentence $\varphi$ is considered fixed and the domain size $n$ is given in unary; the function $n \mapsto \FOMC(\varphi, n)$ is therefore a tally function, i.e.\ a function of a unary input.

\begin{definition}[Unlabeled model count]\label{def:ufomc}
Let $\UMod_n(\varphi)$ be the set of isomorphism classes of $n$-element models of $\varphi$. The \emph{unlabeled model count} of $\varphi$ is $\UFOMC(\varphi, n) = |{\UMod_n(\varphi)}|$.
\end{definition}

For a structure $A$, we denote its automorphism group by $\Aut(A)$. We illustrate the difference between the two counting problems on a small example.

\begin{example}\label{ex:graphs}

Let $E$ be a binary relation symbol and let
\[
\varphi_{\mathrm{graph}} \; := \; \big(\forall x : \neg E(x,x)\big) \wedge \big(\forall x \forall y : (E(x,y) \Rightarrow E(y,x))\big),
\]
whose models are exactly the (simple, undirected) graphs. Each of the $\binom{n}{2}$ unordered pairs of domain elements is independently either an edge or a non-edge, so $\FOMC(\varphi_{\mathrm{graph}}, n) = 2^{\binom{n}{2}}$; for $n = 3$ this gives $8$ labeled graphs. Up to isomorphism there are only $4$ graphs on $3$ vertices (with $0$, $1$, $2$ and $3$ edges), so $\UFOMC(\varphi_{\mathrm{graph}}, 3) = 4$. The values $\UFOMC(\varphi_{\mathrm{graph}}, n) = 1, 2, 4, 11, 34, 156, 1044, \dots$ form the sequence of numbers of unlabeled graphs.\footnote{Sequence A000088 in the On-Line Encyclopedia of Integer Sequences, \url{http://oeis.org/A000088}.}
\end{example}

The unary counting class $\PPone$ \citep{valiant1979} consists of the functions $g : \mathbb{N} \rightarrow \mathbb{N}$ for which there is a nondeterministic polynomial-time Turing machine whose number of accepting paths on the input $1^n$ is exactly $g(n)$.  For every fixed finite vocabulary and fixed first-order sentence $\varphi$, the function $n \mapsto \FOMC(\varphi, n)$ belongs to $\PPone$: a machine guesses the interpretation of every relation and function symbol on $[n]$ and then checks $\varphi$ in polynomial time.

\subsection{Fixed-Sentence Oracle Reductions}\label{sec:reductions}

For a fixed sentence $\varphi$, we write $\FOMC_\varphi(n) = \FOMC(\varphi, n)$ and we define $\UFOMC_\varphi$ analogously. A polynomial-time Turing reduction to $\FOMC_\psi$ may query the single unary oracle
\[
m \longmapsto \FOMC(\psi, m)
\]
on unary inputs $1^m$. The target sentence $\psi$ must be fixed before the input is given; in particular, the reduction may not choose a sentence depending on $n$. Oracle answers and intermediate integers are represented in binary, and arithmetic is charged according to bit complexity.

The two stages of the reduction in Section~\ref{sec:unlabelled-counting} are in fact of a particularly transparent form: they make a single oracle query, at the same domain size as the input, and then perform one exact division by $n!$. Division by an explicitly known integer is ordinary polynomial-time arithmetic; note, however, that it is not a parsimonious operation.

\subsection{Automorphisms and Labelings}

We will repeatedly use the following elementary consequence of the orbit--stabilizer theorem.

\begin{proposition}\label{prop:orbitstab}
Let $B$ be an $n$-element structure. Then $B$ has exactly $n!/|{\Aut(B)}|$ labeled copies on the domain $[n]$. Consequently, the isomorphism class of $B$ contributes exactly
\[
\frac{n!}{|{\Aut(B)}|} \cdot |{\Aut(B)}| = n!
\]
pairs $(A, \pi)$ such that $A$ is a labeled copy of $B$ on $[n]$ and $\pi \in \Aut(A)$.
\end{proposition}

\begin{proof}
Every bijection from the domain of $B$ to $[n]$ produces a labeled copy of $B$, and two bijections produce the same copy exactly when they differ by an automorphism of $B$. This gives the first claim. The second claim follows because all labeled copies $A$ of $B$ satisfy $|{\Aut(A)}| = |{\Aut(B)}|$.
\end{proof}

This is all that we will need from the theory of group actions in this paper.

\section{A Tractable One-Variable Fragment}\label{sec:onevar}

In this section we show that one unary function does not, by itself, make fixed-sentence model counting hard: the entire one-variable fragment with counting quantifiers and one unary function admits a polynomial-time model-counting algorithm, with no restriction on relation arities. As explained in the introduction, this is not a direct consequence of the usual algorithms for relational fragments, because nested terms $f^i(x)$ carry bounded but genuine functional information.

\begin{theorem}[One-variable tractability]\label{thm:onevar}
For every fixed sentence $\varphi \in \Cke{1}[f]$ over an arbitrary finite relational vocabulary and one unary function symbol $f$, the function
\[
n \longmapsto \FOMC(\varphi, n)
\]
is computable in time polynomial in $n$.
\end{theorem}

The proof, given in Section~\ref{sec:proofonevar}, has two stages. We first
treat the case in which all relation symbols are unary. Let $d$ be the largest
nesting depth of $f$ in the fixed sentence. We summarize the unary information
and equality pattern on the segment
\[
x,f(x),\ldots,f^d(x)
\]
by finitely many \emph{forward profiles}. Profiles assigned to successive
vertices must agree on their overlapping orbit positions, and the profiles of
cyclic vertices record every return visible within depth $d$. For each class
of elements relevant to a counting quantifier, we retain its cardinality only
up to a fixed bound determined by the thresholds in the sentence. This finite
local information is used to form exponential generating functions for the
rooted in-trees and directed cycles of the functional digraph. Their
coefficients through degree $n$ can be computed in time polynomial in $n$.
The second stage removes the restriction to unary relation symbols by the
exact reduction described below.

\paragraph{From arbitrary arity to the monadic core.}
Every occurrence of an $r$-ary relation symbol in a one-variable sentence has
the form
\[
R(f^{e_1}(x),\ldots,f^{e_r}(x)),
\]
and therefore inspects $R$ only along finitely many orbit traces determined by
the fixed sentence. These traces can be encoded by finitely many unary
predicates. This yields an exact reduction from arbitrary relational
vocabularies to the monadic case. For every fixed sentence $\varphi$, the
reduction produces a monadic sentence $\theta_\varphi$ and explicitly
computable nonnegative integers $K_\varphi(n)$ and $N_\varphi(n)$ such that
\begin{equation}\label{eq:arity-reduction-overview}
2^{K_\varphi(n)}\FOMC(\varphi,n)
=
2^{N_\varphi(n)}\FOMC(\theta_\varphi,n).
\end{equation}
The two model counts can therefore be converted into one another by exact multiplication and division by powers of two, in polynomial time. The construction and the proof of the identity are given in Section~\ref{sec:arity-reduction}.

\begin{remark}[Domain-size-dependent counting thresholds]\label{rem:variable-thresholds}
The proof also permits the numerical threshold of a counting quantifier to be a fixed nonnegative integer-valued function $t(n)$ computable in time polynomial in $n$. For a domain of size $n$, no witness count exceeds $n$. If $t(n)>n$, the corresponding comparison can therefore be decided immediately; otherwise the required counter has cap at most $n+1$. The formula gives only a fixed number, say $h$, of counter coordinates, so the resulting coefficient algebra has dimension $O((n+2)^h)$. Its dimension is no longer constant, but it is still polynomial in $n$, and all truncated-series operations remain polynomial-time. Thus quantifiers such as $\exists^{=t(n)}x$ and $\exists^{\geq t(n)}x$ may be used. In particular, this applies to linear and Presburger-definable functions of $n$.
\end{remark}

\subsection{Examples}\label{sec:onevar-examples}

We give three examples illustrating the tractability result on familiar combinatorial objects. The first counts endofunctions with a prescribed, domain-size-dependent number of $3$-cycles. The second concerns quasi-kernels in functional digraphs. The third considers bounded-height rooted forests subject to a local two-colour constraint.

\begin{example}[A prescribed number of $3$-cycles]\label{ex:threecycles}
Using the extension from Remark~\ref{rem:variable-thresholds}, consider
\[
\gamma
:=
\exists^{=\,3\lfloor n/3\rfloor}x : \bigl(f^3(x)=x\wedge f(x)\neq x\bigr).
\]
The displayed property holds exactly at the vertices lying on directed $3$-cycles. Hence $\gamma$ counts endofunctions on $[n]$ having exactly $\lfloor n/3\rfloor$ directed $3$-cycles, with no restriction on the remaining components of the functional digraph. Although elementary, the example already uses equality between nonadjacent points of a forward orbit: $f^3(x)$ is compared directly with $x$.

In this case the count also has a simple closed form. Write $n=3c+r$, where $c=\lfloor n/3\rfloor$ and $r\in\{0,1,2\}$. The $3c$ vertices on the directed $3$-cycles can be chosen and partitioned into directed $3$-cycles in
\[
\frac{n!}{r!\,3^c c!}
\]
ways, while each of the remaining $r$ vertices may choose its image arbitrarily. Thus
\[
\FOMC(\gamma,n)
=
\frac{n!}{r!\,3^c c!}\,n^r.
\]
For $n=1,\ldots,12$, the resulting sequence is
\[
1,\ 4,\ 2,\ 32,\ 500,\ 40,\ 1960,\ 71680,\ 2240,\ 224000,\ 14907200,\ 246400.
\]

\end{example}

\begin{example}[Quasi-kernels and forbidden short cycles]\label{ex:quasikernel}
A quasi-kernel of a digraph $D$ is a set $Q\subseteq V(D)$ such that no arc of $D$ has both endpoints in $Q$, and every vertex outside $Q$ can reach a vertex of $Q$ by a directed path of length at most two. Chv\'atal and Lov\'asz proved that every finite digraph has a quasi-kernel~\citep{chvatal1974semikernel}. For the functional digraph of an endofunction $f$, where the directed edge from $x$ is $x\to f(x)$, these two conditions become
\[
\forall x : \bigl(Q(x)\Rightarrow \neg Q(f(x))\bigr)
\]
and
\[
\forall x : \bigl(\neg Q(x)\Rightarrow (Q(f(x))\vee Q(f^2(x)))\bigr).
\]

Fix $k\geq 1$, and let
\begin{align*}
\kappa_k := {}&
\bigl(\exists^{=k}x : Q(x)\bigr)\\
&\wedge\ \forall x : \bigl(Q(x)\Rightarrow \neg Q(f(x))\bigr)\\
&\wedge\ \forall x : \bigl(\neg Q(x)\Rightarrow (Q(f(x))\vee Q(f^2(x)))\bigr).
\end{align*}
Thus $\FOMC(\kappa_k,n)$ counts pairs $(f,Q)$, where $f:[n]\to[n]$ is an endofunction and $Q$ is a distinguished quasi-kernel of size $k$.

If instead we fix in advance a particular $k$-element subset $K\subseteq[n]$ and ask how many endofunctions have $K$ as a quasi-kernel, symmetry gives
\[
N_{n,k}
=
\frac{\FOMC(\kappa_k,n)}{\binom nk}.
\]
Indeed, every $k$-element subset of the labeled domain occurs equally often as the interpretation of $Q$. Notice that this quotient counts functions for which one prescribed set $K$ is a quasi-kernel; it does not count functions that merely possess some quasi-kernel of size $k$, since a function may have several such quasi-kernels.

The quasi-kernel condition itself is already covered by the known tractability of relational $C^2$. One may replace $f$ by a binary relation constrained to be the graph of a total function and express the two-step reachability condition using the two available variables.

We can make the example more interesting by additionally forbidding directed cycles of lengths $2$ and $3$:
\[
\widehat\kappa_k
:=
\kappa_k
\wedge
\forall x : f^2(x)\neq x
\wedge
\forall x : f^3(x)\neq x.
\]
The quasi-kernel conditions already rule out fixed points, so the two additional conjuncts exclude precisely cycles of lengths $2$ and $3$. Excluding $2$-cycles still has an immediate relational $C^2$ formulation. The condition $f^3(x)\neq x$, however, compares a vertex with its third iterate and is no longer handled by the direct two-variable relational encoding above.

Hence, for every fixed $k$, Theorem~\ref{thm:onevar} gives a polynomial-time algorithm for
\[
\frac{\FOMC(\widehat\kappa_k,n)}{\binom nk},
\]
the number of endofunctions on $[n]$ with no directed cycles of lengths $2$ or $3$ for which a prescribed $k$-element subset is a quasi-kernel.

For a prescribed set $K$, the first values of this normalized count for $k=2$ and $k=3$ are
\[
\begin{array}{c|r|r}
n & N_{n,2} & N_{n,3}\\
\hline
4  & 2       & 0\\
5  & 36      & 12\\
6  & 456     & 372\\
7  & 5240    & 7416\\
8  & 59280   & 126960\\
9  & 680148  & 2055060\\
10 & 8014496 & 32737320
\end{array}
\]

\end{example}

\begin{example}[Two-coloured bounded-height rooted forests]\label{ex:colored-forests}
Fix $k\geq2$, and let $B$ be a unary predicate representing one of two colours, with $\neg B$ representing the other. Consider
\begin{align*}
\varphi_k := {}&
\forall x : f^{k+1}(x)=f^k(x)\\
&\wedge\ \forall x : \Bigl(
 f^2(x)\neq f(x)
 \Rightarrow
 \bigl(B(x)\Leftrightarrow (B(f(x))\vee B(f^2(x)))\bigr)
\Bigr).
\end{align*}
The first conjunct says that every vertex reaches a fixed point after at most $k$ applications of $f$. Thus, after the self-loop at each fixed point is omitted, the functional digraph is a rooted forest of height at most $k$, with all edges directed towards the roots.

The second conjunct imposes a colouring condition on every directed path
\[
v_1\longrightarrow v_2\longrightarrow v_3
\]
of three distinct vertices: $v_1$ has colour $B$ if and only if at least one of $v_2$ and $v_3$ has colour $B$. Equivalently, $v_1$ has the other colour exactly when both of the next two vertices towards the root have the other colour.

Consequently, $\FOMC(\varphi_k,n)$ counts two-coloured rooted forests on the labeled vertex set $[n]$, of height at most $k$, satisfying this local colouring rule on every root-directed path of length two. Theorem~\ref{thm:onevar} gives a polynomial-time algorithm for this sequence for every fixed $k$.

For the first three height bounds, the initial values are
\[
\begin{array}{c|r|r|r}
n & \FOMC(\varphi_2,n) & \FOMC(\varphi_3,n) & \FOMC(\varphi_4,n)\\
\hline
1  & 2          & 2          & 2\\
2  & 12         & 12         & 12\\
3  & 104        & 104        & 104\\
4  & 1088       & 1184       & 1184\\
5  & 13552      & 16192      & 16672\\
6  & 195448     & 257848     & 276568\\
7  & 3189944    & 4680944    & 5258864\\
8  & 57947744   & 95278688   & 112542368\\
9  & 1156976864 & 2145968000 & 2674030976\\
10 & 25141280744 & 52921824104 & 69788748584
\end{array}
\]

\end{example}

The fragment $\Cke{1}[f]$ is quite restrictive, and this restriction is essential to our tractability result. As we show in Section~\ref{sec:hardness}, allowing a second variable yields hardness already in $\FOkef{2}$, while allowing a second unary function yields hardness already in $\FOkefg{1}$. The fragment cannot enforce bijectivity, not even with auxiliary unary predicates serving as certificates. Nevertheless, our algorithm for $\Cke{1}[f]$ can be restricted to bijective functions, so adding $\PERM(f)$ as an external axiom preserves polynomial-time model counting. We defer the precise statements and proofs to Appendix~\ref{sec:permutation-axiom}.

\section{\texorpdfstring{Unlabeled Counting}{Unlabeled Counting}}\label{sec:unlabelled-counting}

We carry out the reduction from unlabeled to labeled model counting in two steps. We first use a unary function to mark an automorphism of the structure. We then eliminate the function symbol in favor of a purely relational encoding.

\subsection{Automorphism Marking}\label{sec:marking}

We first explain the automorphism-marking construction in the simplest interesting case, that of undirected graphs, and then give the general construction.

\subsubsection{The Graph Case}\label{sec:graphcase}

Let $\varphi_{\mathrm{graph}}$ be the sentence axiomatizing simple undirected graphs from Example~\ref{ex:graphs}. We extend the vocabulary by a unary function symbol $f$ and add the two axioms
\begin{align*}
&\forall x \forall y : \big( f(x) = f(y) \Rightarrow x = y \big), \\
&\forall x \forall y : \big( E(x,y) \Leftrightarrow E(f(x), f(y)) \big).
\end{align*}
Over a finite domain the first axiom forces $f$ to be a permutation, and the second one then says exactly that $f$ is an automorphism of the graph. The labeled models of the extended sentence are therefore precisely the pairs $(G, \pi)$ where $G$ is a graph on $[n]$ and $\pi$ is an automorphism of $G$.

Now consider an unlabeled graph $B$ whose automorphism group has size $a$. By Proposition~\ref{prop:orbitstab}, $B$ has $n!/a$ labeled copies, and each copy admits exactly $a$ choices of the interpretation of $f$. Its isomorphism class therefore contributes exactly $n!$ models of the extended sentence. The point of the construction is that it corrects the varying automorphism factors \emph{internally}, before the model count is taken. As we explain in Remark~\ref{rem:burnside} below, this is Burnside's lemma in disguise.

\begin{example}\label{ex:marking3}

Let $\varphi'$ denote the extension of $\varphi_{\mathrm{graph}}$ by the two axioms above. Let us verify the identity $\FOMC(\varphi', 3) = 3! \cdot \UFOMC(\varphi_{\mathrm{graph}}, 3)$ by hand. Recall from Example~\ref{ex:graphs} that $\UFOMC(\varphi_{\mathrm{graph}}, 3) = 4$. The $8$ labeled graphs on $[3]$ consist of the empty graph ($6$ automorphisms), three graphs with a single edge ($2$ automorphisms each), three paths with two edges ($2$ automorphisms each), and the triangle ($6$ automorphisms). Summing the numbers of automorphisms over all $8$ labeled graphs gives
\[
6 + 3 \cdot 2 + 3 \cdot 2 + 6 = 24 = 3! \cdot 4,
\]
in agreement with Theorem~\ref{thm:marking} below.
\end{example}

\subsubsection{The General Construction}

Let $\sigma$ be a relational vocabulary and let $f$ be a fresh unary function symbol. We define
\[
\Bij(f) := \forall x \forall y : \big( f(x) = f(y) \Rightarrow x = y \big).
\]
Since functions are total and domains are finite, $\Bij(f)$ forces $f$ to be a permutation. For an $r$-ary relation symbol $R \in \sigma$, we let
\[
\mathrm{Aut}_R(f) := \forall x_1 \cdots \forall x_r : \big( R(x_1, \dots, x_r) \Leftrightarrow R(f(x_1), \dots, f(x_r)) \big),
\]
and for a sentence $\varphi$ over $\sigma$ we define
\[
\varphi_{\mathrm{aut}} := \varphi \wedge \Bij(f) \wedge \bigwedge_{R \in \sigma} \mathrm{Aut}_R(f).
\]

\begin{lemma}\label{lem:markingbij}
For every $n$, the labeled models of $\varphi_{\mathrm{aut}}$ on $[n]$ are in bijection with the pairs $(A, \pi)$ such that $A \in \Mod_n(\varphi)$ and $\pi \in \Aut(A)$.
\end{lemma}

\begin{proof}
We interpret $f$ as a map $\pi : [n] \rightarrow [n]$. The conjunct $\Bij(f)$ makes $\pi$ a permutation. The remaining added conjuncts say exactly that $\pi$ preserves every relation of the reduct $A$. Thus $\pi$ is an automorphism of $A$, and conversely every automorphism of $A$ gives a unique expansion of $A$ to a model of $\varphi_{\mathrm{aut}}$.
\end{proof}

\begin{theorem}[Automorphism marking]\label{thm:marking}
For every relational sentence $\varphi$ and every $n$,
\[
\FOMC(\varphi_{\mathrm{aut}}, n) = n! \cdot \UFOMC(\varphi, n).
\]
Consequently, $\UFOMC_\varphi(n)$ can be computed by one query to the oracle $\FOMC_{\varphi_{\mathrm{aut}}}$ at the same domain size, followed by one exact division by $n!$.
\end{theorem}

\begin{proof}
By Lemma~\ref{lem:markingbij},
\[
\FOMC(\varphi_{\mathrm{aut}}, n) = \sum_{A \in \Mod_n(\varphi)} |{\Aut(A)}|.
\]
We group the summands by isomorphism class: by Proposition~\ref{prop:orbitstab}, every class contributes exactly $n!$ to the sum, and summing over the $\UFOMC(\varphi, n)$ classes gives the identity.
\end{proof}

\begin{remark}[Burnside's lemma]\label{rem:burnside}
The proof above counts the pairs $(A,\pi)$, where $A\in\Mod_n(\varphi)$ and $\pi\in\Aut(A)$, by first fixing $A$. Counting the same pairs by first fixing $\pi$ gives Burnside's lemma. Indeed, let $S_n$ act on $\Mod_n(\varphi)$ by relabeling the domain, and put
\[
\mathrm{Fix}(\pi)
=
\{A\in\Mod_n(\varphi):\pi\in\Aut(A)\}.
\]
Lemma~\ref{lem:markingbij} then gives
\[
\FOMC(\varphi_{\mathrm{aut}}, n)
=
\sum_{\pi\in S_n}|{\mathrm{Fix}(\pi)}|.
\]
By Burnside's lemma,
\[
\frac{1}{n!}\sum_{\pi\in S_n}|{\mathrm{Fix}(\pi)}|
=
\bigl|\Mod_n(\varphi)/S_n\bigr|
=
\UFOMC(\varphi,n),
\]
since the orbits of the relabeling action are precisely the isomorphism classes. Thus the function symbol $f$ turns the fixed-point sum in Burnside's lemma into an ordinary labeled model count: it records the permutation $\pi$, while the added axioms ensure that $\pi$ fixes $A$. For the classical enumerative setting, see \citet{polya1987}.
\end{remark}

\begin{corollary}[No variable blow-up]\label{cor:noblowup}
Suppose that $\sigma$ has maximum arity at most $k$, where $k\geq2$. If $\varphi\in\FOke{k}$, then $\varphi_{\mathrm{aut}}\in\FOkef{k}$; if $\varphi\in\Ck{k}$, then $\varphi_{\mathrm{aut}}\in\Ck{k}[f]$. Thus the intermediate functional construction preserves variable width for both first-order logic with equality and counting logic.
\end{corollary}

Indeed, the added automorphism axioms are first-order formulas using at most $k$ variables. Equality atoms in the source sentence are retained unchanged. Thus conjoining the axioms with an $\FOke{k}$ source sentence gives a sentence of $\FOkef{k}$; conjoining them with a $\Ck{k}$ source sentence leaves its counting quantifiers unchanged and gives a sentence of $\Ck{k}[f]$. The assumption $k\geq2$ remains necessary: $\Bij(f)$ needs two variables.

\subsection{Eliminating the Function Symbol}\label{sec:eliminating}

The function-symbol formulation of automorphism marking is the most succinct one, but much of the model-counting literature works over purely relational vocabularies. We therefore encode the graph of the permutation by a fresh binary relation symbol $F$.

The following $\FOke{3}$ sentence says that $F$ is total, functional and injective:
\begin{align*}
&\forall x \exists y : F(x,y), \\
&\forall x \forall y \forall z : \big( (F(x,y) \wedge F(x,z)) \Rightarrow y = z \big), \\
&\forall x \forall y \forall z : \big( (F(x,z) \wedge F(y,z)) \Rightarrow x = y \big).
\end{align*}
On finite domains, $F$ is therefore the graph of a permutation.

In counting logic, the same condition can be expressed with two variables by
\[
\bigl(\forall x\,\exists^{=1}y\,F(x,y)\bigr)
\;\wedge\;
\bigl(\forall y\,\exists^{=1}x\,F(x,y)\bigr).
\]

A direct preservation axiom for an $r$-ary relation $R$ would use $2r$ variables, one tuple before and one tuple after applying the permutation:
\[
\forall x_1 \cdots \forall x_r \forall y_1 \cdots \forall y_r : \Big( \Big( \bigwedge_{i=1}^{r} F(x_i, y_i) \Big) \Rightarrow \big( R(x_1, \dots, x_r) \Leftrightarrow R(y_1, \dots, y_r) \big) \Big).
\]
The blow-up can be avoided by moving one argument at a time. Let $R$ have arity $r$. We introduce fresh $r$-ary relation symbols $R_1, \dots, R_r$ and write $R_0 = R$. If $F$ is interpreted as the graph of a permutation $\pi$, the intended meaning of the auxiliary relations is that, for all domain elements $a_1, \dots, a_r$,
\[
R_i(a_1, \dots, a_r) \iff R(\pi(a_1), \dots, \pi(a_i), a_{i+1}, \dots, a_r).
\]
For each $i = 1, \dots, r$, we enforce this by
\[
\forall x_1 \cdots \forall x_r \forall y : \Big( F(x_i, y) \Rightarrow \big( R_i(x_1, \dots, x_r) \Leftrightarrow R_{i-1}(x_1, \dots, x_{i-1}, y, x_{i+1}, \dots, x_r) \big) \Big),
\]
and finally we add
\[
\forall x_1 \cdots \forall x_r : \big( R(x_1, \dots, x_r) \Leftrightarrow R_r(x_1, \dots, x_r) \big).
\]
For a binary $R$, the relations $R_1$ and $R_2$ represent $R(\pi(a_1), a_2)$ and $R(\pi(a_1), \pi(a_2))$, respectively, and the last axiom is then exactly the condition that $\pi$ preserves $R$.

\begin{lemma}\label{lem:uniquestages}
If $F$ is the graph of a permutation, then there exist auxiliary relations $R_1, \dots, R_r$ satisfying the staged transition axioms, and they are uniquely determined by $R$ and $F$.
\end{lemma}

\begin{proof}
Assume that $R_{i-1}$ has been determined. For each tuple $(a_1, \dots, a_r)$, there is a unique $b$ with $F(a_i, b)$. Define $R_i$ on that tuple by
\[
R_i(a_1, \dots, a_r) \Leftrightarrow R_{i-1}(a_1, \dots, a_{i-1}, b, a_{i+1}, \dots, a_r).
\]
This definition satisfies the $i$-th transition axiom, so induction from $R_0=R$ proves existence. The same displayed equivalence is forced by the axiom, and the uniqueness of $b$ therefore proves uniqueness at every stage.
\end{proof}

\pagebreak[3]
\begin{theorem}[Relational automorphism marking]\label{thm:relmarking}
Let $\sigma$ be a relational vocabulary of maximum arity $r$, let $k\geq1$, and let $\varphi$ be a sentence over $\sigma$. One can effectively construct a purely relational sentence $\varphi_{\mathrm{aut,rel}}$ such that
\[
\varphi_{\mathrm{aut,rel}}\in
\FOke{\max\{k,r+1,3\}}
\quad\text{if }\varphi\in\FOke{k},
\]
and
\[
\varphi_{\mathrm{aut,rel}}\in
\Ck{\max\{k,r+1,2\}}
\quad\text{if }\varphi\in\Ck{k}.
\]
In either case,
\[
\FOMC(\varphi_{\mathrm{aut,rel}}, n) = n! \cdot \UFOMC(\varphi, n).
\]
\end{theorem}

\begin{proof}
We conjoin $\varphi$ with the permutation-graph axioms for $F$, the staged transition axioms for every relation symbol of $\sigma$, and the final preservation axioms. By Lemma~\ref{lem:uniquestages}, the auxiliary relations introduce no multiplicity, so the models of the resulting sentence are in bijection with those of $\varphi_{\mathrm{aut}}$, and the statement follows from Theorem~\ref{thm:marking}. The staged transition axioms for a relation of arity at most $r$ use at most $r+1$ variables. The first-order permutation-graph axioms use three variables, whereas their counting-logic formulation uses two. Combining these requirements with the $k$ variables of the source sentence gives the two stated maxima.
\end{proof}

For unary vocabularies, counting quantifiers give a sharper relational statement. If $\varphi \in \Ck{1}$, we let
\[
\theta_\varphi := \varphi \wedge \big(\forall x \exists^{=1} y : F(x,y)\big) \wedge \big(\forall y \exists^{=1} x : F(x,y)\big) \wedge \bigwedge_{P \in \sigma} \forall x \forall y : \big( F(x,y) \Rightarrow (P(x) \Leftrightarrow P(y)) \big).
\]
Then $\theta_\varphi \in \Ck{2}$ and the same labeled/unlabeled identity holds, because over a unary vocabulary a permutation is an automorphism if and only if it preserves 1-types.

\begin{corollary}\label{cor:c1}
For every fixed $\Ck{1}$ sentence $\varphi$ over a unary relational vocabulary, the function $n \mapsto \UFOMC(\varphi, n)$ is computable in time polynomial in $n$.
\end{corollary}

\begin{proof}
Weighted first-order model counting, and hence ordinary model counting, is computable in time polynomial in the domain size for every fixed $\Ck{2}$ sentence \citep{kuzelka2021}. We compute $\FOMC(\theta_\varphi, n)$ and divide by $n!$.
\end{proof}

\begin{example}\label{ex:compositions}

In how many ways can $n$ indistinguishable elements be split into three named nonempty classes? Let $\sigma = \{P_1, P_2, P_3\}$ and let
\[
\varphi_{\mathrm{part}} := \Big( \forall x : \big( P_1(x) \vee P_2(x) \vee P_3(x) \big) \Big) \wedge \bigwedge_{1 \leq i < j \leq 3} \forall x : \neg\big(P_i(x) \wedge P_j(x)\big) \wedge \bigwedge_{i=1}^{3} \exists x : P_i(x).
\]
An unlabeled model of $\varphi_{\mathrm{part}}$ is determined by the triple of class sizes $(n_1, n_2, n_3)$ with $n_1 + n_2 + n_3 = n$ and all $n_i \geq 1$. Hence $\UFOMC(\varphi_{\mathrm{part}}, n) = \binom{n-1}{2}$, the number of compositions of $n$ into three parts; for instance, for $n = 5$ there are $\binom{4}{2} = 6$ such splittings. Corollary~\ref{cor:c1} computes these numbers mechanically as $\FOMC(\theta_{\varphi_{\mathrm{part}}}, n) / n!$. Without the three nonemptiness conjuncts the answer would be $\binom{n+2}{2}$, and in general, for $s$ unary predicates and the trivial sentence $\top$, the unlabeled count is the number of multisets of $n$ 1-types, $\UFOMC(\top, n) = \binom{n + 2^s - 1}{2^s - 1}$.
\end{example}

\section{Hardness with a Second Variable or a Second Unary Function}\label{sec:hardness}

There are two natural ways to relax the restrictions of $\Cke{1}[f]$: allow a second variable, or keep one variable and allow a second unary function. Hardness already appears in the corresponding fragments without counting quantifiers: $\FOkef{2}$ in the first case and $\FOkefg{1}$ in the second.

\subsection{A Second Variable}\label{sec:hardness-two-var}

\begin{theorem}[Two-variable hardness]\label{thm:hardness}
There exists a fixed sentence $\Phi \in \FOkef{2}$ such that
\[
n \longmapsto \FOMC(\Phi, n)
\]
is $\PPone$-complete under polynomial-time Turing reductions on unary inputs.
\end{theorem}

We now sketch the construction; the details are given in Section~\ref{sec:proofhardness}. The tableau construction adapts the $\FOk{3}$ hardness construction of \citet{beame2015}; the role of the unary function is to make successor shifts available as terms, allowing the local tableau constraints to be written with two variables. We use the normalized machine $U$ from Section~\ref{sec:common-machine}. On input $1^n$, every branch makes exactly $an-1$ transitions and every tape head remains among the first $bn$ cells.

The sentence $\Phi$ first sets up a coordinate system. A subformula $\CycSucc$ forces a binary relation $<$ to be a strict linear order and the function $f$ to be its cyclic successor; on the domain $[n]$ there are exactly $n!$ models of $\CycSucc$, one for each linear order, and this choice of an order is the only freedom in a model of $\Phi$ that does not correspond to a choice made by the machine. The $an$ time points and the $bn$ positions of each tape are represented by a fixed block index, compiled into predicate names, together with one domain element. Each of the $k$ tapes has its own family of binary tableau predicates. Unary predicates on the time coordinate record the global machine state and the transition rule selected at every nonfinal step. Applying $f$ moves within a block, while a wrap from the maximum to the minimum increments the compiled block index.

For every tape, each cell label records its symbol and whether it lies to the left of, at, or to the right of that tape's head. Fixed two-variable clauses require one head on every tape, require the transition predicate to agree simultaneously with the state and all scanned symbols, and determine the new label of every cell from its old label and the labels of its two neighbors. Separate clauses cover block wraps and the two tape endpoints. The first time point represents the initial multi-tape configuration and the last state is accepting. For a fixed linear order, an accepting branch uniquely determines the state, rule-selection, and tape predicates, and conversely. The rule-selection predicates therefore introduce no additional multiplicity. Hence
\[
\FOMC(\Phi, n) = n! \cdot \acc_U(n),
\]
and exact division by $n!$ establishes the hardness result; membership in $\PPone$ was observed in Section~\ref{sec:background}.

\begin{remark}[A weighted equality-free consequence]\label{rem:weighted}
The equality-elimination construction of \citet[Appendix]{beame2015} replaces equality by a fresh relation whose positive and negative occurrences receive suitable symmetric weights, followed by coefficient extraction. Applied to $\Phi$, it yields $\PPone$-hardness for symmetric weighted model counting of an equality-free $\FOk{2}[f]$ sentence under the convention that the symmetric weights are part of the oracle input: coefficient extraction makes polynomially many oracle queries with different values of the interpolation weight. Thus the argument neither yields an unweighted equality-free FOMC hardness result nor, without an additional reduction, hardness in pure data complexity for one fixed assignment of weights.
\end{remark}

\subsection{A Second Unary Function}\label{sec:hardness-two-fun}

\begin{theorem}[Two-function hardness]\label{thm:twofun-hardness}
There exists a fixed constant-free universal sentence
\[
\Xi\in\FOkefg{1}
\]
using, besides $f,g$, only finitely many unary predicates, such that
\[
N\longmapsto\FOMC(\Xi,N)
\]
is $\PPone$-complete under polynomial-time Turing reductions on unary inputs.
\end{theorem}

The proof is given in Section~\ref{sec:prooftwofun} and again starts from the normalized machine $U$ of Section~\ref{sec:common-machine}. It is convenient first to construct a constant-free universal one-variable sentence $\Theta$ using a fixed finite family of unary functions; the last step of the proof encodes this family by only two unary functions. We replace $U$ by a slowed machine $V$ having exactly $2an$ transitions on input $1^n$, without changing the number of accepting branches.

The main structural ingredient is a unary permutation $S$. We would like one $S$-cycle to provide all time points of the computation. In a constant-free universal sentence, however, we cannot require the entire model to consist of a single such cycle: disjoint unions of models are again models. We therefore allow the domain to split into several groups of $S$-cycles. A unary function $\pi$ partitions the domain into these groups, which we call $\pi$-components. Each $\pi$-component contains exactly one distinguished odd $S$-cycle and may contain additional $S$-cycles, all of which are even. Only the distinguished odd cycle is used to represent a computation. In the full proof, the odd/even distinction is enforced by a unary predicate, and a conjugacy between $S$ and $S^2$ isolates the distinguished odd cycle without using reachability.

Suppose that the distinguished cycle has size
\[
m=4an+1.
\]
Cutting one edge of this cycle gives the linear sequence
\[
c,S(c),\ldots,S^{4an}(c)=b.
\]
Thus the cycle provides $4an+1$ time points and $4an$ successive steps. Fixed unary terms identify the points $S^{an}(c)$ and $S^{2an}(c)$. The first $an$ steps deterministically construct the unary input $1^n$, the next $an$ steps return the input head to the first cell, and the final $2an$ steps simulate $V$. The last point $b$ is the final time point and carries no further transition.

With only one first-order variable, there is no separate domain coordinate for tape position. Instead, for each tape we use two unary functions that, at a given time, point to the latest earlier visits to the cells immediately to the left and right of the current head position. Because the head moves only between neighboring cells, these previous-visit functions can be updated recursively from one time point to the next. When the head enters a neighboring cell, they identify its most recent earlier visit and hence the last symbol written there. This is enough to verify the symbols read by the simulated machine.\footnote{The underlying observation that a tape symbol can be recovered from the most recent earlier visit to the same cell also appears in \citet{durandmore2002}; our one-variable construction replaces their explicit comparison of times and tape positions by unary functions that maintain the relevant neighboring-cell history locally.} Once the distinguished cycle and an accepting branch are fixed, all of this auxiliary information is uniquely determined.

Consequently, if $R_m$ denotes the number of labeled models on $m=4an+1$ elements whose whole domain consists of a single $\pi$-component with no additional $S$-cycles, so that the distinguished cycle is the entire domain, then
\[
R_m=m!\,\acc_U(n).
\]
Indeed, there are $(m-1)!$ directed cycles on the labeled domain and $m$ choices of the distinguished point; all remaining structural and computational information is then forced by the accepting branch.

The ordinary model count of $\Theta$ also includes the additional structure that the sentence deliberately permits, so two corrections are needed to recover $R_m$. First, a model may contain several $\pi$-components, each with its own distinguished odd $S$-cycle. All unary functions preserve the $\pi$-components, and $\Theta$ is universal, so different $\pi$-components behave independently. A model is therefore an unordered labeled collection of one-$\pi$-component models. The standard labeled-SET recurrence recovers the number of models with exactly one $\pi$-component from the unrestricted model counts.

Second, the remaining $\pi$-component may still contain additional even $S$-cycles besides its distinguished computation cycle. Their number on any prescribed set of labels can be evaluated explicitly from permutation cycle types. This gives a second elementary inversion, which removes the additional even cycles and recovers $R_m$.

At this point the source sentence $\Theta$ still uses a fixed finite number of unary functions. We finally encode all of them by two unary functions $f$ and $g$. The target domain is divided into a fixed number of equally sized layers. One function moves cyclically between the layers, while the restriction of the other function to each layer records one of the source functions. The multiplicity of this encoding is explicit and can therefore be divided out exactly. Combining this encoding with the two inversions above gives a polynomial-time Turing reduction from $n\mapsto\acc_U(n)$ to the model-counting function of the fixed sentence $\Xi$, proving Theorem~\ref{thm:twofun-hardness}.

\section{Discussion}\label{sec:discussion}

The positive and negative results give a sharp contrast around the one-variable, one-function setting. With one variable, a fixed sentence sees only bounded pieces of forward orbits, and the decomposition of a functional digraph into cycles with rooted in-trees lets these local constraints be counted in polynomial time. With two variables, the same function can act on both arguments of binary atoms simultaneously, which is enough to support a computation tableau. A second unary function also suffices for hardness without introducing a second variable. Thus, for fixed sentences and unary domain input,

\[
\begin{gathered}
\Cke{1}[f]\text{-FOMC is in polynomial time},\\
\FOkef{2}\text{-FOMC is } \PPone\text{-complete},\\
\FOkefg{1}\text{-FOMC is }\PPone\text{-complete}.
\end{gathered}
\]

The two hardness proofs encode the normalized machine in different ways. In the two-variable proof, the variables $x$ and $y$ index time and tape position, while fixed predicate indices account for the constant-factor enlargement of the tableau. In the two-function proof, domain elements represent time only; previous-visit functions recover the tape information that is needed when the head moves to a neighboring cell. The first proof keeps the domain size equal to $n$, whereas the second uses a constant-factor larger domain. In addition, the constant-free source sentence of the two-function proof may have several components, which are removed from the count by exact inversion.
{\small
\[
\begin{array}{c|c|c}
&\FOkef{2}&\FOkefg{1}\\
\hline
\text{time/tape representation}&\text{two domain coordinates}&\text{time elements and previous visits}\\
\text{domain size for input }1^n&n&\Theta(n)\\
\text{constant-factor expansion}&\text{predicate indices}&\text{additional domain elements}\\
\text{tape contents}&\text{binary tableau predicates}&\text{unary previous-visit functions}\\
\text{additional components}&\text{excluded by the order axioms}&\text{removed by exact inversion}
\end{array}
\]
}

The automorphism-marking theorem gives a second interpretation of the same symbol: a unary function can select one automorphism without increasing the variable width of the source sentence. The identity
\[
\FOMC(\varphi_{\mathrm{aut}}, n) = n! \cdot \UFOMC(\varphi, n)
\]
shows that the correction for automorphisms, which naive division by $n!$ gets wrong, can be expressed inside the logical specification itself. When function symbols are disallowed, the same correction is still available. The relational encoding uses $\max\{k,r+1,3\}$ variables in first-order logic and $\max\{k,r+1,2\}$ in counting logic, where $k$ is the source width and $r$ the maximum relation arity; thus an additional variable is needed only when these maxima exceed $k$.

It is also instructive to compare the role played by the cyclic successor function in Theorem~\ref{thm:hardness} with known positive results about linear orders. Weighted model counting of $\Ck{2}$ sentences remains polynomial-time in the presence of a linear order axiom \citep{toth2023}, and this remains true when immediate and, more generally, $k$-th predecessor relations, for fixed $k$, are added as further semantic axioms \citep{zou2025}. Access to bounded powers of the successor is therefore not the source of the hardness in Theorem~\ref{thm:hardness}. What the function symbol adds is \emph{coupling}: a single two-variable clause can constrain an atom at the pair $(x, y)$ together with the shifted pair $(f(x), f(y))$, as the local-update clauses of Section~\ref{sec:proofhardness} do, and expressing the same constraint with predecessor relations in place of terms would require more than two variables.

The reduction of Theorem~\ref{thm:marking} should not be confused with a canonical-form construction. It obtains the unlabeled count by an oracle query followed by division by $n!$, not by selecting one representative of every isomorphism class. This leaves open the accepting-path complexity of many natural unlabeled counting functions. The number of unlabeled graphs remains the guiding example. Let $u_n$ denote the number of nonisomorphic graphs on $n$ vertices, i.e., $u_n = \UFOMC(\varphi_{\mathrm{graph}}, n)$ in the notation of Example~\ref{ex:graphs}. \citet{wilf1982} conjectured that $\{u_n\}$ cannot be computed in time polynomial in $n$ (see also Conjecture~1.1 of \citealt{pak2018}). \citet{pak2024} asks a complementary question: whether $\{u_n\}$ has a combinatorial interpretation at all, that is, whether $\{u_n\}$ is in $\#\mathrm{P}$ with $n$ given in unary---in the terminology of this paper, whether $\{u_n\} \in \PPone$ (Open Problem~4.1 there). The difficulty identified there is the one we have been discussing: $u_n$ counts orbits rather than combinatorial objects, and there is no obvious way to choose one representative from every orbit. By contrast, for classes whose structures have automorphism groups of polynomial size, computable in polynomial time---such as unlabeled plane triangulations---membership in $\#\mathrm{P}$ is known \citep[Proposition~4.2]{pak2024}.

The main question that our results leave open is which finite-variable fragments admit polynomial-time \emph{unlabeled} model counting, rather than merely a reduction to a labeled oracle. For labeled counting, Theorems~\ref{thm:onevar}, \ref{thm:hardness}, and \ref{thm:twofun-hardness} show that counting quantifiers and arbitrary relational arity remain tractable with one variable and one unary function, whereas either of the two natural relaxations is enough for hardness: it already occurs in $\FOkef{2}$ when a second variable is available, and in $\FOkefg{1}$ when a second unary function is available. For larger fragments, the complexity may also depend on the allowed nesting of function terms and on their occurrence in relation atoms.

We state and prove our results for ordinary model counting. The same ideas extend to symmetric weighted model counting with additional bookkeeping, but we do not pursue the weighted setting here.

\section{Related Work}\label{sec:related}

The work presented in this paper builds on the line of research on lifted inference and symmetric weighted first-order model counting. The original lifted-inference results and their logical reformulations establish polynomial-time data complexity of weighted model counting for fixed $\FOk{2}$ sentences \citep{vandenbroeck2011,vandenbroeck2014,beame2015}.

The work most closely related to our polynomial-time $\Cke{1}[f]$ model-counting theorem is that of \citet[Theorems~3.8 and~4.7]{kuusisto2018}, who extend this tractability result in two independent relational directions. First, they prove polynomial-time symmetric weighted model counting for sentences of the form
\[
\phi\ \wedge\ \forall x\exists^{=1}y\,\psi(x,y),
\qquad \phi,\psi\in\FOke{2}.
\]
This includes the case where a binary relation is constrained to be the graph of a unary function. Second, they prove polynomial-time symmetric weighted model counting for the uniform one-dimensional fragment $U_1$, which admits relations of arbitrary arity subject to one-dimensionality and uniformity conditions. Both results are stated for relational vocabularies. In their functionality setting, higher-arity symbols can be handled because a two-variable formula is invariant under changing facts whose span has more than two elements \citep[Appendix~A.4]{kuusisto2018}. That argument does not apply to our syntax: an atom such as $R(x,f(x),f^2(x))$ can inspect a tuple of span three. Nor does replacing $f$ by its graph generally place our formulas in $U_1$: the natural relationalization of such an atom simultaneously uses the variable sets $\{x,y\}$, $\{y,z\}$, and $\{x,y,z\}$ and violates uniformity. Short-cycle equalities such as $f^3(x)=x$ likewise survive neither a naive two-variable translation nor the $U_1$ route. Thus their theorem and ours address incomparable syntactic extensions; their theorem is formulated for symmetric weights, while we state the present result for ordinary model counting and indicate the corresponding weighted bookkeeping in Section~\ref{sec:discussion}, without formulating a separate weighted theorem. Our construction, unlike theirs, handles genuine nested function terms.

Subsequent work extended the positive boundary to counting quantifiers \citep{kuzelka2021}, to tree, forest, connectivity and related semantic axioms \citep{vanbremen2021,malhotra2025}, to linear order axioms \citep{toth2023}, and to graph-polynomial methods \citep{kuang2026}. Theorem~\ref{thm:onevar} combines these directions in a different positive class: exact unweighted counting for the full one-variable counting logic $\Cke{1}[f]$, with genuinely nested unary-function terms and arbitrary relation arities.

Outside model counting, earlier logical work with one unary function also exploits the cycle-with-in-trees structure of functional graphs. \citet{durand1998} use this structure together with finitely many local types in their study of spectra with unary function symbols. \citet{durandolive2006} study first-order query evaluation over quasi-unary signatures, consisting of one unary function and unary predicates, and obtain quantifier-elimination and efficient query-evaluation results. These works concern spectra or evaluation on a given structure, rather than exact enumeration of all labeled structures of size $n$. Our profiles play a related local role, but are used to enumerate all compatible functional digraphs by exponential generating functions.

Generating-series methods for logical classes with components also appear in
\citet{woods1997}. Woods studies labeled and unlabeled asymptotic probabilities
of monadic second-order sentences and explicitly includes structures with one
unary function (or partial function) and finitely many unary predicates. The
aim there is the existence of limiting probabilities under asymptotic
hypotheses, rather than exact coefficient computation in polynomial time; the
component decomposition and generating-series viewpoint are nevertheless a
close antecedent of our monadic construction.

The closest related hardness result is due to \citet{beame2015}, who exhibit a fixed relational $\FOk{3}$ sentence with a $\PPone$-complete model-counting function. Their proof adapts the standard Trakhtenbrot tableau encoding as presented by \citet[p.~167]{libkin2004}, and handles constant-factor linear time and tape space by compiling epoch and region indices into the predicate vocabulary. Our $\FOkef{2}$ hardness proof in turn adapts this construction: the unary function supplies successor shifts directly as terms, which lets the local tableau updates be expressed with two variables. The simulation of \citet{beame2015} also produces an $n!$ labeling factor, and they point out that quotienting by isomorphism removes it. Our contribution on the labeled/unlabeled side is to isolate the general reduction-theoretic mechanism behind this observation: automorphism-marked extensions convert \emph{every} unlabeled relational counting problem into a fixed labeled model-counting problem. The exact variable bound for the relational form is given in Theorem~\ref{thm:relmarking}: for source width $k$ and maximum relation arity $r$, the targets use $\max\{k,r+1,3\}$ variables in first-order logic and $\max\{k,r+1,2\}$ in counting logic. In particular, unlabeled $\FOke{2}$ and $\Ck{2}$ model counting over binary vocabularies reduce to labeled model counting in $\FOke{3}$ and $\Ck{3}$, respectively.

The history construction in our two-function hardness proof has a separate antecedent in \citet{durandmore2002}. In an arithmetic encoding of a Turing computation, they define a relation giving the latest earlier time at which a head visited its current tape cell and recover the symbol read from the symbol written on that visit. Our construction uses the same semantic observation, but the one-variable language cannot compare two time points or compute tape positions directly. The unary functions $\lambda_\tau$ and $\rho_\tau$ therefore maintain the relevant history locally for the cells neighboring the current head position.

The two-function hardness result is related to work on spectra with unary functions. \citet{grandjean1990} studies one-variable constructions with unary functions and constants and notes the obstruction caused by disjoint unions when constants are not available. We remain in the constant-free setting and remove the additional components from the model count by the recurrence in Section~\ref{sec:twofun-inversions}. \citet{durand1998} study constant-factor reductions between vocabularies of unary functions, including reductions to two unary functions. Lemma~\ref{lem:compression} adds the information needed for model counting by determining the exact labeled multiplicity $\frac{(rq)!}{q!}$.

On the finite-model-theory side, satisfiability of $\FOk{2}$ is decidable \citep{mortimer1975} and NEXPTIME-complete \citep{gradel1997fo2}. Counting quantifiers preserve decidability \citep{gradel1997c2}, and the matching complexity bounds were developed by \citet{pacholski1997} and \citet{pratt2005}. Function symbols cross this boundary quickly: even two-variable logic without equality and with one unary function symbol has an undecidable satisfiability problem, provided that additional relation symbols are available \citep[p.~66]{gradel1997fo2}, where the result is attributed to \citet{borger1997}. The proviso is necessary: over vocabularies that consist of unary predicates and one unary function only, satisfiability is decidable, even for full first-order logic with equality, by Rabin's tree theorem \citep{rabin1969,borger1997}.  Our hardness result is a data-complexity analogue of this phenomenon in model counting, with the sentence fixed and only the domain size varying.

Finally, the labeled/unlabeled distinction belongs to classical group-action enumeration \citep{polya1987}. The complexity of turning quotient counts into checkable combinatorial objects is emphasized by \citet{pak2018,pak2024}. The automorphism-marking identity is elementary as a group-action formula, but its logical internalization provides a uniform bridge between this enumerative viewpoint and fixed-sentence model counting.

\section{Proofs}\label{sec:proofs}

In this section we give the full proofs of Theorems~\ref{thm:onevar}, \ref{thm:hardness}, and \ref{thm:twofun-hardness}. Section~\ref{sec:common-machine} fixes the source machine used in both hardness proofs. The two logical encodings are then proved separately. The statements and proofs concerning bijectivity are collected in Appendix~\ref{sec:permutation-axiom}.

\subsection{Proof of the One-Variable Tractability Theorem}\label{sec:proofonevar}

The proof of Theorem~\ref{thm:onevar} proceeds in two steps. First we prove the
following lemma, which establishes the same conclusion under the additional
assumption that all relation symbols are unary.

\begin{lemma}[Monadic core]\label{lem:monadic-core}
For every fixed sentence $\psi\in\Cke{1}[f]$ over unary relation symbols and one unary function symbol $f$, the function $n\mapsto\FOMC(\psi,n)$ is computable in time polynomial in $n$.
\end{lemma}

Before giving the proof of this lemma, we introduce the generating-function and
algebraic preliminaries in
Section~\ref{sec:onevar-technical-preliminaries}. The proof of the lemma is then
given in
Sections~\ref{sec:onevar-forward-profiles}--\ref{sec:onevar-directed-cycles}.
Finally, Section~\ref{sec:arity-reduction} removes the restriction on relation
arities and completes the proof of Theorem~\ref{thm:onevar}.

\subsubsection{Technical Preliminaries}\label{sec:onevar-technical-preliminaries}

All generating functions below are formal power series; no analytic convergence
is involved. The variable $x$ records the size of a structure, while the
coefficients lie either in $\mathbb Q$ or in a finite-dimensional algebra used
for additional bookkeeping.

By a \emph{commutative $\mathbb Q$-algebra} we mean a $\mathbb Q$-vector space
equipped with an associative, commutative, bilinear multiplication and a
multiplicative identity $1$. Thus its elements can be added and scaled by
rational numbers, but they can also be multiplied. Basic examples of such
$\mathbb Q$-algebras are $\mathbb Q$ itself and polynomial rings such as
$\mathbb Q[z]$. The algebra
used in the proof is the finite-dimensional monoid algebra defined below.

\begin{samepage}
Let $\mathcal A$ be a commutative $\mathbb Q$-algebra. We write
$\mathcal A[[x]]$ for the algebra of formal power series over $\mathcal A$.
Let $F(x)\in\mathcal A[[x]]$, say
\[
F(x)=\sum_{m\geq0}c_mx^m,
\qquad c_m\in\mathcal A.
\]
We write $[x^m]F(x)=c_m$. If $F(0)=0$, its formal exponential is
\[
\exp(F(x))
=
1+\sum_{j\geq1}\frac{F(x)^j}{j!}.
\]
This defines a formal power series without any appeal to convergence: since
$F(x)$ has zero constant term, $F(x)^j$ has degree at least $j$, and hence only
finitely many summands contribute to any fixed coefficient.
\end{samepage}

We use the standard labeled product and set constructions for exponential
generating functions; see, for example, \citet[Chapter~II]{flajolet2009}. Thus,
if
\[
A(x)=\sum_{m\geq0}a_m\frac{x^m}{m!}
\qquad\text{and}\qquad
B(x)=\sum_{m\geq0}b_m\frac{x^m}{m!}
\]
count labeled classes, their product distributes the label set between an
$A$-object and a $B$-object. If $A(0)=0$, then $\exp(A(x))$ counts finite
unordered sets of $A$-objects. For example, the series $x$ represents one
labeled atom, and
\[
\exp(x)=1+\sum_{m\geq1}\frac{x^m}{m!}
\]
represents an unordered set of such atoms: there is exactly one such structure
on every finite label set.

To carry finite logical summaries through these constructions, we use monoid
algebras. Let $(M,\oplus,\mathbf 0)$ be a finite commutative monoid. One may
think of an element $\mathbf u\in M$ as a summary state and of
$\mathbf u\oplus\mathbf v$ as the state obtained by combining two disjoint
objects. The monoid algebra $\mathbb Q[M]$ is the $\mathbb Q$-vector space with
basis
\[
\{Y^{\mathbf u}:\mathbf u\in M\}.
\]
Here the superscript $\mathbf u$ is an index: $Y^{\mathbf u}$ denotes a formal
basis symbol, not an ordinary power of $Y$. Multiplication is determined on
basis elements by
\[
Y^{\mathbf u}Y^{\mathbf v}=Y^{\mathbf u\oplus\mathbf v},
\]
and is extended bilinearly. The multiplicative identity is
$Y^{\mathbf 0}$, which we also write as $1$, and $\mathbb Q[M]$ has dimension
$|M|$. If
\[
C(x)=\sum_{m\geq0}\sum_{\mathbf u\in M}
c_{m,\mathbf u}x^mY^{\mathbf u},
\]
then $[x^mY^{\mathbf u}]C(x)$ denotes the coefficient
$c_{m,\mathbf u}$.

The monoids used later record capped counts. For the simplest example, let
\[
M_2=\{0,1,2\},
\qquad
u\oplus v=\min(2,u+v).
\]
The algebra $\mathbb Q[M_2]$ has basis $Y^0,Y^1,Y^2$, with
\[
Y^1Y^1=Y^2,
\qquad
Y^1Y^2=Y^2.
\]
The three basis symbols can be read as recording a count of $0$, $1$, or at
least $2$. Multiplication combines the recorded counts and caps the result at
$2$. A direct product of such monoids records several capped counts
simultaneously.

\subsubsection{The Monadic Core: Forward Profiles}\label{sec:onevar-forward-profiles}

Throughout this subsection we fix a sentence $\varphi \in \Cke{1}[f]$ over unary predicates $P_1, \dots, P_s$ and one unary function symbol $f$. Let $d$ be the maximum nesting depth of $f$ in a term occurring in $\varphi$; if $f$ does not occur in $\varphi$, we take $d = 0$.

Profiles generalize the 1-types of Section~\ref{sec:background} from single elements to bounded forward segments. An \emph{atomic $d$-formula} is an atom of the form $P_r(f^i(x))$, where $1 \leq r \leq s$ and $0 \leq i \leq d$, or of the form $f^i(x) = f^j(x)$, where $0 \leq i < j \leq d$. A \emph{$d$-profile} is a set $q$ of atomic $d$-formulas and negations of atomic $d$-formulas such that
\begin{itemize}
\item for every atomic $d$-formula $\theta$, exactly one of $\theta$ and $\neg\theta$ belongs to $q$, and
\item $q$ is satisfiable, i.e., there are a finite structure $A$ and an element $a \in A$ at which every literal of $q$ is true.
\end{itemize}
In other words, a $d$-profile is a complete and satisfiable quantifier-free description of the forward segment $x, f(x), \dots, f^d(x)$: the analogue of a 1-type at term depth $d$. We denote by $\Qc$ the set of all $d$-profiles, and for a structure $A$ and an element $a \in A$ we write $\prof^A_d(a)$ for the \emph{profile of $a$}, i.e.\ the unique $q \in \Qc$ whose literals are all true at $a$.

For a profile $q$, we recover the two coordinates with which the algorithm will work: for $0 \leq i \leq d$, we write $\chi_q(i) = \{ P_r : P_r(f^i(x)) \in q \}$ for the 1-type that $q$ assigns to position $i$ of the segment, and we write $i \equiv_q j$ when $i = j$ or the equality between $f^i(x)$ and $f^j(x)$ belongs to $q$. The relation $\equiv_q$ is reflexive and symmetric by construction. The next lemma characterizes the $d$-profiles among all complete sets of literals and shows, in particular, that $\Qc$ is finite and computable from $\varphi$.

\begin{lemma}\label{lem:realizable}
Let $q$ be a set of literals containing exactly one of $\theta$ and $\neg\theta$ for every atomic $d$-formula $\theta$. Then $q$ is a $d$-profile if and only if $\equiv_q$ is an equivalence relation on $\{0, \dots, d\}$ and the following two conditions hold:
\begin{itemize}
\item[(a)] if $i \equiv_q j$, then $\chi_q(i) = \chi_q(j)$;
\item[(b)] if $i \equiv_q j$ and $i, j < d$, then $i + 1 \equiv_q j + 1$.
\end{itemize}
\end{lemma}

\begin{proof}
The conditions are clearly necessary: equality of elements is transitive, equal elements satisfy the same unary predicates, and applying $f$ to equal elements yields equal elements. For sufficiency, we take the quotient of the directed path $0 \rightarrow 1 \rightarrow \cdots \rightarrow d$ by $\equiv_q$. Condition (b) makes the successor map well-defined on every class containing an index below $d$. If the class of $d$ has no earlier representative, we define its successor arbitrarily, say by a self-loop; otherwise its successor is already determined by an earlier representative. We interpret each unary predicate according to $\chi_q$, which is well-defined by condition (a). The class of $0$ then realizes $q$. Finiteness and computability of $\Qc$ follow, since there are finitely many complete sets of literals and the conditions of the lemma are decidable.
\end{proof}

\begin{example}\label{ex:profiles}

Let the vocabulary consist of one unary predicate $P$ and the function symbol $f$, and let $d = 2$, so that a profile describes the segment $a, f(a), f^2(a)$. Consider the structure $A$ with domain $\{1, 2, 3, 4\}$ in which $f(1) = 2$, $f(2) = 3$, $f(3) = 2$, $f(4) = 4$, and in which $P$ is true exactly at $2$ and $4$. Writing $q_a = \prof^A_2(a)$, the four elements realize the following profiles:
\begin{center}
\begin{tabular}{c|c|c|c}
$a$ & segment $a, f(a), f^2(a)$ & $\chi_{q_a}(0),\ \chi_{q_a}(1),\ \chi_{q_a}(2)$ & $\equiv_{q_a}$ \\ \hline
$1$ & $1, 2, 3$ & $\emptyset,\ \{P\},\ \emptyset$ & $0, 1, 2$ all distinct \\
$2$ & $2, 3, 2$ & $\{P\},\ \emptyset,\ \{P\}$ & $0 \equiv 2$ \\
$3$ & $3, 2, 3$ & $\emptyset,\ \{P\},\ \emptyset$ & $0 \equiv 2$ \\
$4$ & $4, 4, 4$ & $\{P\},\ \{P\},\ \{P\}$ & $0 \equiv 1 \equiv 2$
\end{tabular}
\end{center}
The example also shows why both parts of a profile are needed. The elements $1$ and $3$ carry the same 1-types along their segments but have different equality patterns, so $q_1 \neq q_3$. Written out as a set of literals, the profile of the element $2$ reads
\[
q_2 = \big\{ P(x),\ \neg P(f(x)),\ P(f^2(x)),\ \neg\big(x = f(x)\big),\ x = f^2(x),\ \neg\big(f(x) = f^2(x)\big) \big\}.
\]
The characterization of Lemma~\ref{lem:realizable} also rules out exactly the combinations that cannot occur in any structure: there is no profile with $0 \equiv_q 2$ and $\chi_q(0) \neq \chi_q(2)$, by condition (a), and none with $0 \equiv_q 1$ but $1 \not\equiv_q 2$, by condition (b)---once $f(a) = a$, the whole segment collapses. Note also that the equality pattern of $q_2$ and $q_3$, which identifies positions $0$ and $2$ while separating positions $0$ and $1$, is exactly the equality pattern of a vertex on a directed $2$-cycle.
\end{example}

We next determine when profiles assigned to adjacent vertices can be realized simultaneously and isolate the additional information needed on directed cycles.

\subsubsection{Successor Compatibility and Visible Cycle Length}\label{sec:onevar-successor-compatibility}

{
Recall that the functional digraph of $f$ has the domain of the structure as
its vertex set, with an edge $a\to f(a)$ from every vertex $a$. Each weakly
connected component consists of a unique directed cycle together with directed
trees feeding into it. We call the vertices on the cycles \emph{cyclic} and all
remaining vertices \emph{transient}. Removing the outgoing edge of every cyclic
vertex leaves a directed forest, rooted at the cyclic vertices and oriented
toward them.

We will describe consistent profile assignments using the same two conditions at every vertex. The first compares the overlapping parts of the profiles at the two ends of an edge. Suppose that $a$ has profile $p$ and $f(a)$ has profile $q$. Position $i+1$ in $p$ and position $i$ in $q$ then describe the same element:
\[
f^{i+1}(a)=f^i(f(a)).
\]

\begin{definition}[Successor compatibility]\label{def:successor-compatibility}
For $p,q\in\Qc$, we say that $q$ is \emph{successor-compatible} with $p$, and
write $p\rightsquigarrow q$, if
\begin{align*}
\chi_p(i+1)&=\chi_q(i) &&(0\leq i<d),\\
i+1\equiv_p j+1&\iff i\equiv_q j &&(0\leq i,j<d).
\end{align*}
\end{definition}

Thus, whenever $f(a)=b$,
\[
\prof_d^A(a)\rightsquigarrow\prof_d^A(b).
\]
Successor compatibility determines the unary types and equality relations among positions $1,\ldots,d$ of $p$ from positions $0,\ldots,d-1$ of $q$. It does not determine whether position $0$ of $p$ coincides with a later position.

That remaining information is captured by one parameter. For $q\in\Qc$, define its \emph{visible cycle length} by
\[
\lambda(q)=
\min\{r\in\{1,\ldots,d\}:0\equiv_q r\},
\]
where the minimum of the empty set is $\infty$. If a vertex $a$ has profile $q$
and $\lambda(q)=r<\infty$, then $a=f^r(a)$, so $a$ lies on a cycle whose length
divides $r$. The minimality of $r$ forces that cycle to have length exactly
$r$. If $\lambda(q)=\infty$, then $a$ is either transient or lies on a cycle of
length greater than $d$.

For the functional digraph $D$ of $f$, set
\[
\lambda_D(a)=
\begin{cases}
\ell,&\text{if $a$ lies on a cycle of length $\ell\leq d$},\\
\infty,&\text{otherwise}.
\end{cases}
\]
Thus every structure $A$ with functional digraph $D$ satisfies
\[
\lambda\bigl(\prof_d^A(a)\bigr)=\lambda_D(a)
\]
for every $a\in V(D)$. When $d=0$, both visible-cycle-length definitions remain
valid: every visible cycle length is $\infty$, and successor compatibility has
no conditions to check.

{
\begin{lemma}[Profile realization]\label{lem:profile-assignment}
Let $D$ be the functional digraph of an endofunction $f$, and let
\[
\rho:V(D)\longrightarrow\Qc
\]
assign a profile to every vertex. Let $A_\rho$ be the structure with functional digraph $D$ obtained by interpreting each unary predicate at a vertex $a$ according to $\chi_{\rho(a)}(0)$. Suppose that
\[
\rho(a)\rightsquigarrow\rho(f(a))
\qquad\text{and}\qquad
\lambda(\rho(a))=\lambda_D(a)
\]
for every $a\in V(D)$. Then
\[
\prof_d^{A_\rho}(a)=\rho(a)
\]
for every $a\in V(D)$.
\end{lemma}

\begin{proof}
Fix $a\in V(D)$. Repeated successor compatibility gives
\[
\chi_{\rho(a)}(i)=\chi_{\rho(f^i(a))}(0)
\qquad(0\leq i\leq d),
\]
and
\[
i\equiv_{\rho(a)}j
\quad\Longleftrightarrow\quad
0\equiv_{\rho(f^i(a))}j-i
\qquad(0\leq i<j\leq d).
\]
The first identity shows that $\rho(a)$ records the correct unary types. Since profiles are realizable by definition, for every $q\in\Qc$ and $1\leq r\leq d$,
\[
0\equiv_q r
\quad\Longleftrightarrow\quad
\lambda(q)<\infty\ \text{ and }\ \lambda(q)\mid r.
\]
For every vertex $b$ and the same range of $r$, the definition of $\lambda_D$ gives
\[
b=f^r(b)
\quad\Longleftrightarrow\quad
\lambda_D(b)<\infty\ \text{ and }\ \lambda_D(b)\mid r.
\]
Taking $b=f^i(a)$ and $r=j-i$, and using $\lambda(\rho(b))=\lambda_D(b)$, shows that
\[
i\equiv_{\rho(a)}j
\quad\Longleftrightarrow\quad
f^i(a)=f^j(a).
\]
Thus $\rho(a)$ also records the correct equality pattern, so $\prof_d^{A_\rho}(a)=\rho(a)$.
\end{proof}
}
}

\subsubsection{Capped Cardinalities}\label{sec:onevar-capped-cardinalities}

A profile records everything about one element that can be seen by a
quantifier-free subformula of the fixed sentence. Counting quantifiers require
one additional kind of information: for certain subsets $S$ of the profile set
$\Qc$, the number of elements whose profile lies in $S$. Since the sentence
contains only finitely many fixed numerical thresholds, these numbers need not
be retained exactly. It is enough to cap each of them at an appropriate
threshold.

For example, let $\theta(x)$ be quantifier-free and let $S_\theta$ be the set
of profiles at which $\theta$ holds. To evaluate
$\exists^{=2}x\,\theta(x)$, it is enough to know
\[
\min\bigl(3,|\{a\in A:\prof_d^A(a)\in S_\theta\}|\bigr).
\]
The four possible recorded values distinguish $0$, $1$, $2$, and ``at least
$3$'' witnesses. For $\exists^{\geq2}x\,\theta(x)$, the cap can instead be
$2$. The general statement is as follows.

\begin{lemma}[Capped-cardinality evaluation]\label{lem:capped-evaluation}
There exist $h\geq1$, subsets $S_1,\ldots,S_h\subseteq\Qc$, positive integers
$b_1,\ldots,b_h$, and a function
\[
\operatorname{Acc}_\varphi:
\prod_{j=1}^h\{0,\ldots,b_j\}\longrightarrow\{0,1\},
\]
all effectively determined by the fixed sentence $\varphi$, such that every
structure $A$ satisfies
\[
A\models\varphi
\quad\Longleftrightarrow\quad
\operatorname{Acc}_\varphi\bigl(\mathbf u(A)\bigr)=1,
\]
where $\mathbf u(A)=(u_1(A),\ldots,u_h(A))$ and
\[
u_j(A)
=
\min\bigl(b_j,\,|\{a\in A:\prof_d^A(a)\in S_j\}|\bigr)
\qquad(1\leq j\leq h).
\]
\end{lemma}

Thus the endpoint $b_j$ is a saturation value: it represents every count that
is at least $b_j$.

\begin{proof}
We process quantified subformulas from the inside out. The one-variable
restriction is important here: a quantified subformula occurring properly
inside another quantified formula is closed, because its quantifier binds the
only variable symbol.

Consider a quantified subformula $Qx\,\eta$ after all quantified subformulas
properly contained in $\eta$ have been processed. Fixing the truth values of
these inner closed subformulas turns them into Boolean constants. What remains
inside $Qx$ is quantifier-free in $x$, so its truth at an element depends only
on the profile of that element. It therefore selects a subset
$S\subseteq\Qc$. Only finitely many such subsets arise because $\varphi$ is
fixed.

For these fixed truth values, the truth of $Qx\,\eta$ depends only on the number
of elements whose profiles lie in $S$. For $\exists^{=m}x\,\eta$, retaining
this number up to cap $m+1$ is sufficient: the capped value $m+1$ represents
all cardinalities larger than $m$. For $\exists^{\geq m}x\,\eta$ with
$m\geq1$, cap $m$ is sufficient. Ordinary existential quantification uses cap $1$,
universal quantification is treated as $\neg\exists x\,\neg\eta$, and the case
$m=0$ is simplified directly. Introducing these counters and continuing
outward yields the sets $S_j$, the caps $b_j$, and the Boolean acceptance map
$\operatorname{Acc}_\varphi$. If no counter is needed, we add the harmless
dummy counter $S_1=\varnothing$ with cap $b_1=1$.
\end{proof}

\subsubsection{Rooted In-Trees}\label{sec:onevar-rooted-in-trees}

We use the standard decomposition of a functional digraph into directed cycles
and rooted in-trees. Each connected component contains a unique directed cycle,
and every vertex outside that cycle belongs to an in-tree whose edges are
directed towards the cycle. In this section, we construct generating functions
for these labeled in-trees, with each vertex decorated by a forward profile and
the profiles along every edge required to satisfy the compatibility condition
(Definition~\ref{def:successor-compatibility}). The coefficients also record
the contribution of the vertices to the capped cardinalities of
Lemma~\ref{lem:capped-evaluation}. From these series we obtain the generating
function for a cycle vertex together with the in-tree rooted at that vertex. In
Section~\ref{sec:onevar-directed-cycles}, we then assemble these decorated cycle
vertices into directed cycles and the resulting connected components into
functional digraphs.

Before defining the tree series, we specify how they record the capped
cardinalities. The degree in $x$ records the number of vertices, and hence the
size of the corresponding structure. To record the capped state in the
coefficients, put
\[
M=\prod_{j=1}^h\{0,\ldots,b_j\}
\]
and define
\[
(\mathbf u\oplus\mathbf v)_j
=
\min(b_j,u_j+v_j).
\]
Then $(M,\oplus,\mathbf0)$ is a finite commutative monoid. We use the monoid
algebra $\mathcal R=\mathbb Q[M]$, whose basis element $Y^{\mathbf u}$ records
the capped state $\mathbf u$. As in
Section~\ref{sec:onevar-technical-preliminaries}, $Y^{\mathbf u}$ is a basis
symbol indexed by $\mathbf u$, not an ordinary power of $Y$.

For each profile $q\in\Qc$, define $\mathbf e(q)\in M$ by
\[
e(q)_j=
\begin{cases}
1,&q\in S_j,\\
0,&q\notin S_j,
\end{cases}
\]
Thus $\mathbf e(q)$ records the capped-cardinality coordinates to which a
vertex of profile $q$ contributes. We associate with every such vertex the
weight
\[
\omega_q=Y^{\mathbf e(q)}.
\]
Multiplying vertex weights combines their contributions using capped addition.
For a structure $A$, the resulting weight is $Y^{\mathbf u(A)}$, because
$u_j(A)$ is the number of vertices whose profiles lie in $S_j$, capped at
$b_j$. Thus
\begin{equation}\label{eq:profile-weight-product}
\prod_{a\in A}\omega_{\prof_d^A(a)}
=
Y^{\mathbf u(A)}.
\end{equation}

\begin{example}\label{ex:profile-weights}
Continue with the structure $A$ and the $2$-profiles $q_1,\ldots,q_4$ from
Example~\ref{ex:profiles}. Consider two capped counts: the number of elements
satisfying $P$, capped at $2$, and the number of fixed points of $f$, capped at
$1$. Thus
\[
S_1=\{q\in\Qc:P(x)\in q\},
\qquad b_1=2,
\]
and
\[
S_2=\{q\in\Qc:x=f(x)\in q\},
\qquad b_2=1.
\]
In this case,
\[
M=\{0,1,2\}\times\{0,1\}.
\]
The contribution vectors of the four profiles are
\[
\mathbf e(q_1)=\mathbf e(q_3)=(0,0),
\qquad
\mathbf e(q_2)=(1,0),
\qquad
\mathbf e(q_4)=(1,1).
\]
Consequently,
\[
\omega_{q_1}\omega_{q_2}\omega_{q_3}\omega_{q_4}
=
Y^{(0,0)}Y^{(1,0)}Y^{(0,0)}Y^{(1,1)}
=
Y^{(2,1)}.
\]
In $A$, exactly two elements satisfy $P$, and exactly one element is a fixed
point of $f$. Since both counts reach their caps, the vector $(2,1)$ records
only that at least two elements satisfy $P$ and that $f$ has at least one fixed
point.
\end{example}

We now define the tree series by cutting the functional digraph at a transient
vertex $a$. The vertex $a$, together with all vertices feeding into it, forms a
subtree whose edges point towards $a$. We omit the edge $a\to f(a)$, which
attaches this subtree to the rest of the functional digraph. Let
\[
\Qc_\infty=\{q\in\Qc:\lambda(q)=\infty\}.
\]
Every transient vertex has a profile in $\Qc_\infty$. The converse need not
hold: a vertex on a cycle longer than $d$ also has no visible return. Its role
as a cycle vertex will be imposed only when the cycle is assembled.

For $q\in\Qc_\infty$, let $T_q(x)\in\mathcal R[[x]]$ count rooted labeled
directed trees whose edges point towards the root, whose root has profile $q$,
and whose profile assignment is successor-compatible along every edge. More
precisely,
\[
m![x^mY^{\mathbf u}]T_q(x)
\]
is the number of such trees on the label set $[m]$ for which the product of all
vertex weights is $Y^{\mathbf u}$.

The root contributes $x\omega_q$. Here $x$ records the root as one vertex of
the tree, while $\omega_q$ records its contribution to the capped counts. A
child subtree has root profile $p$ with $p\rightsquigarrow q$, because the edge
from that child points to the vertex of profile $q$. The children form an
unordered labeled set, so the exponential formula gives
\begin{equation}\label{eq:tree-egf}
T_q(x)
=
x\omega_q
\exp\!\left(
\sum_{\substack{p\in\Qc_\infty\\p\rightsquigarrow q}}
T_p(x)
\right),
\qquad q\in\Qc_\infty.
\end{equation}
Every $T_q$ has zero constant term, and, for $m\geq1$, the coefficient of $x^m$
in $T_q$ is $\omega_q$ times the coefficient of $x^{m-1}$ in the exponential
on the right-hand side of \eqref{eq:tree-egf}. This coefficient depends only on
the coefficients of the series $T_p$ through degree $m-1$. Starting with degree
$1$, we can therefore determine all coefficients successively. Thus
\eqref{eq:tree-egf} is a system of equations for the family
$(T_q)_{q\in\Qc_\infty}$ in $\mathcal R[[x]]$, and it has a unique solution with
zero constant terms. As a basic special case, if there is one profile $q$,
$q\rightsquigarrow q$, and $\omega_q=1$, then
$T(x)=x\exp(T(x))$, the usual equation for rooted labeled trees.

A cycle vertex may have any profile $q\in\Qc$, but every noncyclic vertex
feeding into it is transient. We therefore define
\begin{equation}\label{eq:cycle-decoration}
U_q(x)
=
x\omega_q
\exp\!\left(
\sum_{\substack{p\in\Qc_\infty\\p\rightsquigarrow q}}
T_p(x)
\right),
\qquad q\in\Qc.
\end{equation}
The series $U_q$ counts one cycle vertex of profile $q$ together with the
unordered collection of transient in-trees attached to it. When
$q\in\Qc_\infty$, the formulas for $T_q$ and $U_q$ coincide; their roles differ
only in the later assembly, where the distinguished vertex counted by $U_q$
is placed on a cycle.

\begin{lemma}[Tree-series computation]\label{lem:tree-series-computable}
For every $n$, the coefficients of all series in \eqref{eq:tree-egf} and
\eqref{eq:cycle-decoration} through degree $n$ can be computed using a number
of arithmetic operations polynomial in $n$.
\end{lemma}

\begin{proof}
Work throughout modulo $x^{n+1}$. Start with $T_q^{(0)}=0$ for every
$q\in\Qc_\infty$, and obtain $T^{(r+1)}$ by substituting $T^{(r)}$ on the
right-hand side of \eqref{eq:tree-egf}. By induction on $r$, the series
$T_q^{(r)}$ counts exactly the trees of height less than $r$ whose root has
profile $q$. A rooted tree on at most $n$ vertices has height at most $n-1$.
After $n$ substitutions, every coefficient through degree $n$ has therefore
stabilized.

The sets $\Qc$ and $M$ depend only on the fixed sentence. Thus there are only
constantly many series, and every coefficient in $\mathcal R$ is a vector of
the fixed dimension $|M|$. Addition, multiplication, and formal exponentiation
of series truncated after degree $n$ use polynomially many arithmetic
operations. Once the series $T_p$ have been computed through degree $n$, we
obtain each $U_q$ from \eqref{eq:cycle-decoration} by performing the indicated
formal-power-series operations modulo $x^{n+1}$. We defer the bit-complexity
analysis to the proof of Lemma~\ref{lem:monadic-core} at the end of
Section~\ref{sec:onevar-directed-cycles}.
\end{proof}

\subsubsection{Directed Cycles and Components}\label{sec:onevar-directed-cycles}

The series $U_q$ supplies a cycle vertex together with all transient structure
feeding into it. It remains to arrange these decorated vertices into directed
cycles. For $\ell\geq1$, put
\begin{equation}\label{eq:active-cycle-profiles}
\Qc_\ell=
\begin{cases}
\{q\in\Qc:\lambda(q)=\ell\},&\ell\leq d,\\
\Qc_\infty,&\ell>d.
\end{cases}
\end{equation}
These are exactly the profiles that may occur on an actual directed
$\ell$-cycle. If $\ell\leq d$, the return to the starting vertex is visible and
its first occurrence is at position $\ell$. If $\ell>d$, no return is visible
within the profile window, so the profile lies in $\Qc_\infty$.

Let $B_\ell(x)$ be the matrix with rows and columns indexed by $\Qc_\ell$ and
entries
\begin{equation}\label{eq:cycle-matrix}
(B_\ell(x))_{pq}
=
\begin{cases}
U_p(x),&p\rightsquigarrow q,\\
0,&\text{otherwise}.
\end{cases}
\end{equation}
Recall that $U_p(x)$, defined in \eqref{eq:cycle-decoration}, counts a cycle
vertex of profile $p$ together with all transient in-trees feeding into it.
Thus a transition from profile $p$ to the profile $q$ of its successor is
allowed precisely when $p\rightsquigarrow q$, and the transition contributes
the decorated vertex $U_p$ at its source. A term in
$\operatorname{tr}(B_\ell(x)^\ell)$ chooses a closed successor-compatible
sequence
\[
p_0\rightsquigarrow p_1\rightsquigarrow\cdots
\rightsquigarrow p_{\ell-1}\rightsquigarrow p_0
\]
and contributes $U_{p_0}\cdots U_{p_{\ell-1}}$.

The trace distinguishes a starting position on the directed cycle. Every
labeled decorated cycle has exactly $\ell$ choices of such a position, even
when its profile sequence is periodic. Removing this choice gives the series
for one connected component whose cycle has length $\ell$:
\begin{equation}\label{eq:component-egf}
C_\ell(x)
=
\frac{1}{\ell}\operatorname{tr}\bigl(B_\ell(x)^\ell\bigr).
\end{equation}

\begin{example}[The first three cycle lengths]
For $\ell=1,2,3$, formula \eqref{eq:component-egf} gives
\begin{align*}
C_1(x)
&=
\sum_{\substack{p\in\Qc_1\\p\rightsquigarrow p}}U_p(x),
\\
C_2(x)
&=
\frac12
\sum_{\substack{p,q\in\Qc_2\\
p\rightsquigarrow q,\ q\rightsquigarrow p}}
U_p(x)U_q(x),
\\
C_3(x)
&=
\frac13
\sum_{\substack{p,q,r\in\Qc_3\\
p\rightsquigarrow q,\ q\rightsquigarrow r,\ r\rightsquigarrow p}}
U_p(x)U_q(x)U_r(x).
\end{align*}
The profiles in these sums need not be distinct. The factors $1/2$ and $1/3$
remove the choice of origin on the directed cycle. Its orientation is fixed by
the function $f$, so there is no additional factor for reversal.
\end{example}

A functional digraph is an unordered labeled set of its connected components.
A second application of the exponential formula therefore gives
\begin{equation}\label{eq:full-egf}
F_\varphi(x)
=
\exp\!\left(\sum_{\ell\geq1}C_\ell(x)\right).
\end{equation}
This is a well-defined formal power series: only the terms with
$\ell\leq n$ can contribute to the coefficient of $x^n$.

\begin{lemma}[EGF formula]\label{lem:egf-formula}
For every $n\geq0$,
\begin{equation}\label{eq:egf-coefficient-formula}
\FOMC(\varphi,n)
=
n!
\sum_{\substack{\mathbf u\in M\\\operatorname{Acc}_\varphi(\mathbf u)=1}}
[x^nY^{\mathbf u}]F_\varphi(x).
\end{equation}
\end{lemma}

\begin{proof}
Let $A$ be a monadic structure with functional digraph $D$, and assign to every
vertex its actual profile. The unique decomposition of $D$ into directed cycles
with rooted in-trees gives an object counted by \eqref{eq:full-egf}. Every
transient profile lies in $\Qc_\infty$, every edge is successor-compatible, and
the profiles on an $\ell$-cycle lie in $\Qc_\ell$.
Equation~\eqref{eq:tree-egf} counts the transient trees,
\eqref{eq:cycle-decoration} attaches them to cycle vertices,
\eqref{eq:component-egf} removes the distinguished origin from each directed
cycle, and \eqref{eq:full-egf} forms the set of components. By
\eqref{eq:profile-weight-product}, the resulting weight is
$Y^{\mathbf u(A)}$.

Conversely, an object counted by \eqref{eq:full-egf} determines a functional
digraph $D$ together with a profile assignment $\rho$. The definitions of
$T_q$, $U_q$, and $B_\ell$ ensure
\[
\rho(a)\rightsquigarrow\rho(f(a))
\]
for every vertex. If $a$ is transient, then $\rho(a)\in\Qc_\infty$, and hence
$\lambda(\rho(a))=\infty=\lambda_D(a)$. If $a$ lies on an $\ell$-cycle, then
$\rho(a)\in\Qc_\ell$, so again
$\lambda(\rho(a))=\lambda_D(a)$. Lemma~\ref{lem:profile-assignment} now shows
that the assigned profiles are exactly the profiles realized in the resulting
structure. Their zeroth positions determine all unary predicates, so the pair
$(D,\rho)$ determines a unique monadic structure.

It follows that
\[
n![x^nY^{\mathbf u}]F_\varphi(x)
\]
is the number of monadic structures on the labeled domain $[n]$ whose capped
state is $\mathbf u$. Lemma~\ref{lem:capped-evaluation} then shows that summing
over the states accepted by $\operatorname{Acc}_\varphi$ retains exactly the
models of $\varphi$. For $n=0$, the empty set of components contributes
$Y^{\mathbf0}$, and the same formula applies.
\end{proof}

\begin{proof}[Proof of Lemma~\ref{lem:monadic-core}]
Fix $n$ and perform all series calculations modulo $x^{n+1}$.
Lemma~\ref{lem:tree-series-computable} computes the series $T_q$ and the
decorations $U_q$ with polynomially many arithmetic operations. For each
$1\leq\ell\leq n$, the matrix $B_\ell$ has dimension at most $|\Qc|$, a
constant determined by the fixed sentence. Its required power, trace, and the
series $C_\ell$ can therefore be computed with polynomially many operations on
truncated series. The finite sum in \eqref{eq:full-egf}, its formal
exponential, and the coefficient extraction in
\eqref{eq:egf-coefficient-formula} require only polynomially many further
operations. When $d=0$, we have
$\Qc_\ell=\Qc_\infty=\Qc$ for every $\ell$.

This arithmetic-operation bound does not by itself control the lengths of the
integers and rational numbers that occur. We address bit complexity here by
using the natural integral normalization for exponential generating functions.
For $A(x)\in\mathcal R[[x]]$, write
\[
\widehat A_m:=m![x^m]A(x)
\]
for its factorial-scaled coefficient. If $A(x),B(x)\in\mathcal R[[x]]$, then
\[
\widehat{(AB)}_m
=
\sum_{i=0}^m\binom mi\widehat A_i\widehat B_{m-i}.
\]
The binomial coefficient records the choice of which $i$ of the $m$ labels are
used by the first factor. All products in this identity are taken in the monoid
algebra $\mathcal R$.

Suppose that $E(x)=\exp(G(x))$ with $G(0)=0$. The condition $G(0)=0$ ensures
that the formal exponential is well defined and that $E(0)=1$. All operations
in what follows are operations on formal power series; no convergence is
involved. Differentiating the defining series term by term gives
\[
E'
=
\frac{d}{dx}\sum_{j\geq0}\frac{G^j}{j!}
=
\sum_{j\geq1}\frac{G^{j-1}G'}{(j-1)!}
=
G'E.
\]
Write the two series in terms of their factorial-scaled coefficients:
\[
E(x)=\sum_{m\geq0}\widehat E_m\frac{x^m}{m!},
\qquad
G(x)=\sum_{m\geq1}\widehat G_m\frac{x^m}{m!}.
\]
Their derivatives are therefore
\[
E'(x)=\sum_{m\geq0}\widehat E_{m+1}\frac{x^m}{m!},
\qquad
G'(x)=\sum_{i\geq0}\widehat G_{i+1}\frac{x^i}{i!}.
\]
Using the preceding product identity, the coefficient of $x^m/m!$ in $G'E$ is
\[
\sum_{i=0}^m
\binom mi
\widehat G_{i+1}\widehat E_{m-i}.
\]
Here the binomial coefficient arises from rewriting
\[
\frac{1}{i!(m-i)!}
=
\frac{1}{m!}\binom mi.
\]
Comparing the coefficients of $x^m/m!$ on the two sides of $E'=G'E$ now gives
\[
\widehat E_{m+1}
=
\sum_{i=0}^m
\binom mi
\widehat E_{m-i}\widehat G_{i+1}.
\]
The right-hand side involves only $\widehat E_0,\ldots,\widehat E_m$ and the
known coefficients $\widehat G_1,\ldots,\widehat G_{m+1}$. Starting from
$\widehat E_0=1$, we can therefore compute the coefficients of $E$
successively. The recurrence uses only addition, multiplication, and integer
binomial coefficients, so it introduces no rational denominators. The tree
series and the decorations have coefficients in the integral lattice
$\mathbb Z[M]$ after factorial scaling.
In \eqref{eq:component-egf}, division by $\ell$ is exact at the same scaled
level, because the trace counts each labeled decorated cycle once for each of
its $\ell$ choices of origin. The outer exponential therefore also has
factorial-scaled coefficients in $\mathbb Z[M]$.

For $n\geq1$, all such coefficients count profile-decorated objects, possibly
with a distinguished vertex, on at most $n$ vertices. A uniform bound is
$|\Qc|^n n^{n+1}$: choose a profile at each vertex, choose the image of each
vertex under the function, and, for a pointed object, choose one distinguished
vertex. Since $|\Qc|$ and $|M|$ are fixed, the coefficient vectors have entries
with $O(n\log n)$ bits. The case $n=0$ is immediate. Hence the complete
calculation runs in time polynomial in $n$ in the bit model.
\end{proof}

\subsubsection{Removing the Arity Restriction}\label{sec:arity-reduction}

{
We now fix an arbitrary sentence $\varphi\in\Cke{1}[f]$. If nullary relation symbols are admitted, we enumerate their finitely many truth assignments and replace the corresponding atoms by Boolean constants. It therefore suffices to consider relation symbols of positive arity.

Every term is of the form $f^e(x)$. Consider an occurrence
\[
R(f^{e_1}(x),\ldots,f^{e_r}(x))
\]
of a relation symbol $R$ of arity $r\geq2$. Put
\[
m=\min_{1\leq i\leq r}e_i,
\qquad
\beta=(e_1-m,\ldots,e_r-m).
\]
We call a vector $\beta\in\mathbb N^r$ normalized if $\min_i\beta_i=0$,
equivalently, if at least one coordinate is zero. Thus the vector above is
normalized. Let $\Gamma_R$ be the set of vectors obtained in this way from
occurrences of $R$ in $\varphi$. Let $k_R:=|\Gamma_R|$. Fix an ordering of
$\Gamma_R$. For each $\beta\in\Gamma_R$, choose a coordinate $j(\beta)$ with
$\beta_{j(\beta)}=0$.

For a fixed interpretation of $f$, every $\beta\in\Gamma_R$ determines a unary parametrization
\[
\iota_\beta^f:[n]\longrightarrow[n]^r,
\qquad
\iota_\beta^f(a)
=
\bigl(f^{\beta_1}(a),\ldots,f^{\beta_r}(a)\bigr).
\]
The image of $\iota_\beta^f$ is
\[
T_\beta^f
:=\operatorname{im}\iota_\beta^f
=\{\iota_\beta^f(a):a\in[n]\}
=\left\{
\bigl(f^{\beta_1}(a),\ldots,f^{\beta_r}(a)\bigr)
:a\in[n]
\right\}.
\]
The map $\iota_\beta^f$ is injective because its $j(\beta)$-th coordinate is
\[
f^{\beta_{j(\beta)}}(a)=f^0(a)=a.
\]
Let
\[
S_R(f)=\bigcup_{\beta\in\Gamma_R}T_\beta^f.
\]
Since $e_i=m+\beta_i$ for every $i$, we have
\[
\bigl(f^{e_1}(x),\ldots,f^{e_r}(x)\bigr)
=
\iota_\beta^f\bigl(f^m(x)\bigr).
\]
Hence, for every value of $x$, the argument tuple supplied to $R$ belongs to
$T_\beta^f\subseteq S_R(f)$.

Different orbit traces may overlap: distinct pairs $(\beta,a)$ can represent
the same tuple. If these overlaps were ignored, the same $R$-tuple could be
represented more than once, leading to spurious overcounting. We therefore
choose a unique canonical pair for each tuple in $S_R(f)$.

\begin{example}[Overlapping orbit traces]\label{ex:overlapping-orbit-traces}
Let $R$ be binary, let $[n]=\{1,2\}$, let $f$ swap $1$ and $2$, and suppose
that $\Gamma_R$ contains
\[
\beta=(0,1)
\qquad\text{and}\qquad
\gamma=(1,0).
\]
Then
\[
\iota_\beta^f(1)=(1,2)=\iota_\gamma^f(2).
\]
In fact,
\[
T_\beta^f=T_\gamma^f=\{(1,2),(2,1)\}.
\]
Thus the pairs $(\beta,1)$ and $(\gamma,2)$ determine the same tuple.
\end{example}

We choose the first trace in the fixed ordering that contains a given tuple.
Fix $\beta,\gamma\in\Gamma_R$ and $a\in[n]$. If
$\iota_\beta^f(a)=\iota_\gamma^f(b)$, then, because
$\gamma_{j(\gamma)}=0$, the coordinate $j(\gamma)$ forces
\[
b=f^{\beta_{j(\gamma)}}(a).
\]
Substituting this value shows that the two tuples are equal precisely
when
\[
f^{\beta_i}(a)
=
f^{\beta_{j(\gamma)}+\gamma_i}(a)
\qquad(1\leq i\leq r).
\]
Accordingly, let
\[
E^R_{\beta,\gamma}(x)
:=
\bigwedge_{i=1}^{r}
f^{\beta_i}(x)
=
f^{\beta_{j(\gamma)}+\gamma_i}(x).
\]
Then
\[
E^R_{\beta,\gamma}(a)
\quad\Longleftrightarrow\quad
\iota_\beta^f(a)\in T_\gamma^f,
\]
and, when these conditions hold,
\[
\iota_\beta^f(a)
=
\iota_\gamma^f\bigl(f^{\beta_{j(\gamma)}}(a)\bigr).
\]
Here membership in $T_\gamma^f$ means that the tuple can be written as
$\iota_\gamma^f(b)$ for some $b\in[n]$. If an occurrence normalized by
$\gamma$ has exponents $q+\gamma_1,\ldots,q+\gamma_r$, it evaluates only the
tuples in
\[
\iota_\gamma^f\bigl(\operatorname{im}(f^q)\bigr),
\]
which may be a proper subset of $T_\gamma^f$. Thus membership in
$T_\gamma^f$ does not mean that a particular occurrence normalized by
$\gamma$ evaluates the tuple.

Next, we define formulas $\operatorname{Can}^R_{\beta,\gamma}(x)$ that choose a
unique vector $\gamma$ for the tuple $\iota_\beta^f(x)$. The intended meaning is
that $\operatorname{Can}^R_{\beta,\gamma}(a)$ is true if and only if $\gamma$ is
the first vector in the fixed ordering of $\Gamma_R$ for which there exists
$b\in[n]$ such that
\[
\iota_\beta^f(a)=\iota_\gamma^f(b).
\]
In other words, $\gamma$ is the first vector whose map $\iota_\gamma^f$
produces the same tuple as $\iota_\beta^f(a)$. Since
$E^R_{\beta,\gamma}(a)$ expresses precisely that such a $b$ exists, define
\[
\operatorname{Can}^R_{\beta,\gamma}(x)
:=
E^R_{\beta,\gamma}(x)
\wedge
\bigwedge_{\substack{\delta\in\Gamma_R\\\delta<\gamma}}
\neg E^R_{\beta,\delta}(x).
\]
Since $E^R_{\beta,\beta}(x)$ is valid, at least one trace contains
$\iota_\beta^f(a)$. Exactly one $\gamma$ satisfies
$\operatorname{Can}^R_{\beta,\gamma}(a)$, namely the first trace in the fixed
ordering that contains this tuple. Membership in a trace depends only on the
tuple, not on the pair $(\beta,a)$ used to name it, so this choice is canonical.
More explicitly, let
\[
\mathcal C_R(f)
:=
\bigl\{(\gamma,b)\in\Gamma_R\times[n]:
\operatorname{Can}^R_{\gamma,\gamma}(b)\bigr\}.
\]
The map
\[
(\gamma,b)\longmapsto\iota_\gamma^f(b)
\]
is a bijection from $\mathcal C_R(f)$ to $S_R(f)$. Indeed, let $t\in S_R(f)$,
and let $\gamma$ be the first vector in the fixed ordering of $\Gamma_R$ such
that $t\in T_\gamma^f$. Since $T_\gamma^f=\operatorname{im}\iota_\gamma^f$,
there is a unique $b\in[n]$ such that
\[
t=\iota_\gamma^f(b).
\]
The choice of $\gamma$ implies that
$\operatorname{Can}^R_{\gamma,\gamma}(b)$ holds, so
$(\gamma,b)\in\mathcal C_R(f)$.

Conversely, suppose that
\[
\iota_\gamma^f(b)=\iota_{\gamma'}^f(b')
\]
for two pairs in $\mathcal C_R(f)$. Their $\operatorname{Can}$-conditions imply
that both $\gamma$ and $\gamma'$ are the first vector whose corresponding set
contains this tuple. Hence $\gamma=\gamma'$, and the injectivity of
$\iota_\gamma^f$ then gives $b=b'$.

Introduce a fresh unary predicate $P_{R,\gamma}$ for every
$\gamma\in\Gamma_R$. For each $(\gamma,b)\in\mathcal C_R(f)$, the value
$P_{R,\gamma}(b)$ represents the value of $R$ on the tuple
$\iota_\gamma^f(b)$. Values $P_{R,\gamma}(b)$ for pairs outside
$\mathcal C_R(f)$ are not used by the translation.

Consider an atom
\[
R(f^{e_1}(x),\ldots,f^{e_r}(x))
\]
where $e_i=m+\beta_i$ for every $i$. Fix $c\in[n]$ and evaluate the atom at
$x=c$. Put
\[
a=f^m(c).
\]
There is a unique $\gamma\in\Gamma_R$ such that
\[
\operatorname{Can}^R_{\beta,\gamma}(a)
\]
holds. For this $\gamma$, put
\[
b=f^{\beta_{j(\gamma)}}(a)
  =f^{m+\beta_{j(\gamma)}}(c).
\]
By the definition of the canonical pairs, $(\gamma,b)\in\mathcal C_R(f)$, and
\[
(f^{e_1}(c),\ldots,f^{e_r}(c))
=\iota_\beta^f(a)
=\iota_\gamma^f(b).
\]
Thus the truth value of the original atom at $x=c$ is represented by
$P_{R,\gamma}(b)$.

We therefore replace the atom by
\[
\bigvee_{\gamma\in\Gamma_R}
\left(
\operatorname{Can}^R_{\beta,\gamma}(f^m(x))
\wedge
P_{R,\gamma}
\bigl(f^{m+\beta_{j(\gamma)}}(x)\bigr)
\right).
\]
For every value $c$ of $x$, exactly one disjunct is active, and that disjunct
reads the unary predicate value representing $R$ on the tuple evaluated by the
original atom. Hence the replacement preserves the truth value of the atom.
Apply this replacement to every higher-arity relation atom, leaving unary
relation atoms and equality atoms unchanged, and let $\theta_\varphi$ be the
resulting sentence. Then $\theta_\varphi\in\Cke{1}[f]$ is fixed and effectively
constructible from $\varphi$. If $d$ is the largest exponent in $\varphi$,
every term introduced by the translation has depth at most $2d$.

Some tuples in $S_R(f)$ may be absent from every actual evaluation, since an
occurrence with shift $m$ evaluates only tuples in
$\iota_\beta^f(\operatorname{im}(f^m))$.\footnote{The treatment of relation
tuples not inspected by the formula is analogous to that of Kuusisto and Lutz
\citep[Appendix~A.4]{kuusisto2018}. In their two-variable setting, these are
tuples whose span exceeds two. Here, the inspected tuples lie in finitely many
orbit traces, which may overlap and therefore require canonical
representatives.} The formula $\varphi$ does not inspect the $R$-values of
these tuples, and $\theta_\varphi$ does not inspect the corresponding unary
predicate values. Thus every such tuple contributes one free choice to the
count for $\varphi$ and one free choice to the count for $\theta_\varphi$. The
remaining unused bits are different in number but easy to count: the original
interpretation of $R$ has one bit for each tuple in $[n]^r$, whereas the
translated vocabulary has one bit for every pair in $\Gamma_R\times[n]$.

Recall that, for every higher-arity relation symbol $R$, $k_R=|\Gamma_R|$ is
the number of unary predicates $P_{R,\gamma}$ introduced for $R$.
Put
\[
N_\varphi(n)
=
\sum_{\substack{R\in\operatorname{voc}(\varphi)\\
\operatorname{arity}(R)\geq2}}
n^{\operatorname{arity}(R)},
\qquad
K_\varphi(n)
=
n\sum_{\substack{R\in\operatorname{voc}(\varphi)\\
\operatorname{arity}(R)\geq2}}
k_R.
\]
Thus $N_\varphi(n)$ is the total number of table entries in all original
higher-arity relations, while $K_\varphi(n)$ is the total number of positions in
the unary predicates introduced for those relations. The canonicalization above is needed only to ensure that every tuple in $S_R(f)$ is represented by exactly one unary-predicate bit. After the truth values on $S_R(f)$ are fixed, exactly $|S_R(f)|$ bits are constrained on both the original and translated sides, so this common term cancels from the multiplicity ratio. The next lemma compares
the two numbers of free choices.

\begin{lemma}[Exact arity reduction]\label{lem:exact-arity-reduction}
For every $n\geq1$,
\[
2^{K_\varphi(n)}\FOMC(\varphi,n)
=
2^{N_\varphi(n)}\FOMC(\theta_\varphi,n).
\]
\end{lemma}

\begin{proof}
Fix the interpretation of $f$ and of the original unary predicates. For every higher-arity symbol $R$, fix a truth assignment to the tuples in $S_R(f)$. An interpretation of $R$ extending this assignment has
\[
2^{n^{\operatorname{arity}(R)}-|S_R(f)|}
\]
choices, since all tuples outside $S_R(f)$ are free.

In the translated vocabulary, the $k_R$ predicates
$(P_{R,\gamma})_{\gamma\in\Gamma_R}$ have $k_Rn$ values altogether. The
assignment fixes one predicate value for each pair in $\mathcal C_R(f)$. These
pairs are in bijection with $S_R(f)$, so the other predicate values are free and
give
\[
2^{k_Rn-|S_R(f)|}
\]
interpretations of $(P_{R,\gamma})_{\gamma\in\Gamma_R}$. Under the corresponding assignments, every original atom and its translation have the same truth value. A compositional induction gives the same conclusion for all subformulas; at a counting quantifier, it preserves the truth of the quantified subformula at every element and hence the size of its witness set. Tuples in $S_R(f)$ that are not inspected cause no difficulty, since their values are free on both sides.

The quotient of the two multiplicities for $R$ is
\[
2^{n^{\operatorname{arity}(R)}-k_Rn},
\]
independently of $f$. Multiplying over all higher-arity symbols and summing over the fixed interpretations and tuple assignments proves the identity.
\end{proof}

\begin{example}[Ternary data along orbit traces]\label{ex:ternary-orientation}
Let $H$ be ternary and consider
\[
\begin{aligned}
\varphi_{\mathrm{or3}}:=\forall x\bigl(
&f^3(x)=x\wedge f(x)\neq x\\
&{}\wedge\bigl(
H(x,f(x),f^2(x))
\leftrightarrow
H(f(x),f^2(x),x)
\bigr)\\
&{}\wedge\bigl(
H(x,f(x),f^2(x))
\leftrightarrow
\neg H(x,f^2(x),f(x))
\bigr)
\bigr).
\end{aligned}
\]
The function $f$ partitions the domain into directed $3$-cycles. On a cycle $a\to b\to c\to a$, the first equivalence makes the values of $H$ on
\[
(a,b,c),\ (b,c,a),\ (c,a,b)
\]
equal, and the second makes the values on the other three orderings their common negation. Thus each cycle chooses one of its two cyclic orientations. If $n=3c$, there are $n!/(3^c c!)$ possible functions, and the six constrained tuples on each cycle contribute one free bit; all other $H$-tuples are free. Consequently,
\[
\FOMC(\varphi_{\mathrm{or3}},n)
=
\frac{n!}{3^c c!}\,2^{n^3-5c}.
\]
The count is $0$ when $3\nmid n$. Although the notation is ternary, the sentence uses $H$ only through three orbit traces. For example,
\[
H(f(a),f^2(a),a)
=
H(b,c,a)
=
H(b,f(b),f^2(b)),
\]
so different pairs $(\beta,a)$ can determine the same tuple. The construction
above makes them use the same unary predicate bit, reducing the apparent
ternary information to unary orbit data.
\end{example}

\begin{proof}[Proof of Theorem~\ref{thm:onevar}]
For $n=0$ the value is determined directly by the fixed sentence, so assume $n\geq1$. Apply Lemma~\ref{lem:monadic-core} to the fixed sentence $\theta_\varphi$. Lemma~\ref{lem:exact-arity-reduction} gives
\[
\FOMC(\varphi,n)
=
\frac{2^{N_\varphi(n)}\FOMC(\theta_\varphi,n)}
{2^{K_\varphi(n)}}.
\]
The division is exact by the lemma. Both exponents are fixed polynomials in $n$, and the integers in the displayed computation have polynomial bit length, so the required shifts and exact division take polynomial time. The earlier enumeration of nullary truth assignments contributes only a constant number of computations. This proves polynomial-time computability for arbitrary finite relational vocabularies.
\end{proof}
}

\subsection{A Common Source Machine for the Hardness Proofs}\label{sec:common-machine}

Both hardness proofs use the same fixed unary-input counting machine. We state the normalization once and use it in both proofs.

For a fixed synchronous $k$-tape nondeterministic Turing machine, let $Q$ denote its finite state set and let $\Delta$ denote its finite set of transition rules. A rule $\delta\in\Delta$ specifies a source state $q^-_\delta$, a target state $q^+_\delta$, and, for every tape $j$, a symbol $r_{\delta,j}$ to be read, a symbol $w_{\delta,j}$ to be written, and a head displacement
\[
d_{\delta,j}\in\{-1,0,1\}.
\]
The rule $\delta$ is enabled when the machine is in state $q^-_\delta$ and the head on tape $j$ scans $r_{\delta,j}$ for every $j$. Applying $\delta$ changes the state to $q^+_\delta$ and simultaneously, on every tape $j$, writes $w_{\delta,j}$ and moves the head by $d_{\delta,j}$.

Transition rules are identified by these data: two rules having the same source and target states and the same read, write, and movement data on every tape are the same element of $\Delta$.\footnote{A formalism with separately named identical choices can be converted branch-bijectively to this one, at constant-factor slowdown, by inserting a private intermediate state for each choice; the fixed machine used here already has the stated form.} Formally, a computation branch of length $T$ is a sequence
\[
C_0\xrightarrow{\delta_0}C_1
\xrightarrow{\delta_1}\cdots
\xrightarrow{\delta_{T-1}}C_T,
\]
where $\delta_t$ is enabled in $C_t$ and $C_{t+1}$ is obtained by applying $\delta_t$. Two branches are distinct when their sequences of transition rules differ.

\begin{lemma}[Source-machine normal form]\label{lem:source-machine}
There are a fixed $k$-tape nondeterministic Turing machine $U$, constants $a \geq 2$ and $b \geq 3$, a finite set $\Delta$ of transition rules, and distinguished states $q_0$ and $q_{\mathrm{acc}}$ with the following properties. The names $q_0$ and $q_{\mathrm{acc}}$ refer to states of the finite control; they are not constant symbols of the logical vocabulary. For each $n \geq 1$, let $\acc_U(n)$ denote the number of accepting computation branches of $U$ on input $1^n$. The function
\[
n \longmapsto \acc_U(n)
\]
is $\PPone$-hard. On input $1^n$, every branch makes exactly $an-1$ transitions; after those transitions it is accepting exactly when its state is $q_{\mathrm{acc}}$. Each tape head starts at its left endpoint, remains among the first $bn$ cells, and moves by at most one cell per transition. We use the standard marked-left-end convention: the leftmost cell of every tape carries a marked version of its symbol, the mark is preserved whenever that cell is written, and no transition moves left while scanning a marked symbol.
\end{lemma}

\begin{proof}
Let $U_0$ be the machine supplied by Lemma~3.8 of \citet{beame2015}. Beame et al. use a one-tape-per-step presentation of multi-tape machines \citep[Section~3.3]{beame2015}; we use the equivalent synchronous presentation defined above. For each original transition acting on tape $j$, and for each possible tuple of symbols scanned on the other tapes, we include a synchronous rule. On tape $j$, this rule performs the read, write, and head movement prescribed by the original transition. On every other tape, it reads the corresponding symbol, writes that same symbol, and leaves the head stationary. Since the machine has finitely many tapes with finite alphabets, this produces finitely many rules. In any fixed configuration, exactly one tuple of symbols on the other tapes matches. Thus the enabled synchronous rules are in bijection with the enabled original transitions, and computation branches and their lengths are unchanged. Since the lemma gives a linear running-time bound for this fixed machine, we may fix an integer $c$ such that every branch of $U_0$ on input $1^n$ halts after at most $cn$ transitions.\footnote{Lemma~3.8 does not state a numerical value for $c$, but its construction makes such a value effectively obtainable. Its proof consists of fixed elementary routines that convert the unary input to binary and decode the integers $i,j$, followed by a simulation of at most $(ij^i+i)^2\leq n$ transitions. Fix concrete multi-tape implementations of the first two routines and obtain bounds $c_{\mathrm{bin}}n+d_{\mathrm{bin}}$ and $c_{\mathrm{dec}}n+d_{\mathrm{dec}}$ by direct transition counting. If $h$ is the fixed phase-change overhead, then, for $n\geq1$, any integer $c\geq c_{\mathrm{bin}}+c_{\mathrm{dec}}+1+d_{\mathrm{bin}}+d_{\mathrm{dec}}+h$ is a valid bound. Computing such a $c$ is therefore finite bookkeeping for these explicit routines; it does not require inferring a running-time bound for an arbitrary Turing machine.}

We construct a machine $U$ whose branches all have the same length. Throughout the construction, one transition of $U$ means one application of its multi-tape transition relation, which simultaneously updates the state, the scanned symbols, and the head positions on all tapes.

Initially, the input head scans the first symbol of $1^n$. The machine $U$ first copies the input to an additional tape that will serve as the input tape for the simulation of $U_0$. In each of the first $n$ transitions, $U$ copies one input symbol and moves both relevant heads one cell to the right; during the first of these transitions it marks the first copied cell. After the $n$th transition, both heads scan the blank immediately following the input. In the next transition, $U$ writes a right-end marker on the copied tape and moves both heads one cell to the left. It then makes $n-1$ transitions moving both heads to the left, after which they again scan their respective first input cells. One further transition, with both heads stationary, places $U$ in the initial state of the simulation. During the simulation, the marked first cell and the right-end marker are read as $1$ and blank, respectively. This deterministic preprocessing takes exactly
\[
n+1+(n-1)+1=2n+1
\]
transitions and introduces no nondeterministic choices.

Choose an integer $p \geq c+3$ and put $a:=p+2$. During the simulation phase, the original input tape serves as a clock. While its clock head scans the first input symbol, $U$ makes exactly $p-3$ transitions and moves the clock head one cell to the right on the last of them. While the clock head scans each of the remaining $n-1$ input symbols, $U$ makes exactly $p$ transitions and again moves the clock head to the right on the last of them. Because $p$ is fixed, these delays are implemented by finitely many states in the finite control of $U$; no counter depending on $n$ is required. After the clock head reaches the blank following the input, $U$ makes one final transition. The simulation phase therefore contains exactly
\[
(p-3)+p(n-1)+1=pn-2
\]
transitions.

During each of the first
\[
(p-3)+p(n-1)=pn-3
\]
transitions of this phase, $U$ performs the prescribed clock update and, provided that the simulated branch has not halted, simultaneously simulates one transition of $U_0$. This simultaneous update is possible because one transition of the multi-tape machine $U$ can update both the clock tape and the tapes used for the simulation. For every enabled transition rule $\delta$ of $U_0$, the machine $U$ has exactly one rule that performs the prescribed clock update together with the state, tape, and head updates specified by $\delta$.

Since $p \geq c+3$, for every $n \geq 1$ we have
\[
pn-3-cn=(p-c)n-3\geq 3n-3\geq 0.
\]
Thus the $pn-3$ transitions available for the simulation suffice for every branch of $U_0$. When the simulated branch halts, $U$ records whether it accepted or rejected. All subsequent transitions preserve the simulated configuration and the recorded outcome while advancing the clock deterministically. On the final transition, $U$ enters $q_{\mathrm{acc}}$ if and only if the recorded outcome is acceptance.

Consequently, every branch of $U$ has exactly
\[
(2n+1)+(pn-2)=(p+2)n-1=an-1
\]
transitions. The preprocessing, clock, and padding introduce no nondeterministic choices. Every branch of $U_0$ therefore corresponds to exactly one branch of $U$ with the same acceptance status, and hence
\[
\acc_U(n)=\acc_{U_0}(n).
\]

We use the standard semi-infinite-tape convention in which a head at the first cell cannot move to its left, so a transition prescribing such a move is not legal there. Finally, for every tape symbol $c$, add a marked copy $\bar c$ that is used only in the leftmost cell, and duplicate the finitely many transition rules for marked and unmarked scanned symbols. When a scanned symbol is marked, the duplicate preserves the mark; no marked duplicate is supplied for a rule that would move that head to the left. Thus the recoding merely makes the already-existing left boundary syntactically visible and does not change the length or number of computation branches.

All tape heads start at their marked left endpoints and move by at most one cell per transition. Since $U$ makes $an-1$ transitions, choosing a fixed integer $b \geq \max\{a+2,3\}$ ensures that every head remains among the first $bn$ cells. Let $\Delta$ be the finite set of transition rules of $U$.
\end{proof}

\subsection{Proof of the Two-Variable Hardness Theorem}\label{sec:proofhardness}

We adapt the $\FOk{3}$ tableau construction of \citet{beame2015}, which in turn follows the standard Trakhtenbrot encoding as presented by \citet[p.~167]{libkin2004}. In particular, their epoch and region indices compile constant-factor linear time and tape space into finitely many predicate names. Here the unary function $f$ is forced to be a cyclic successor, so successor shifts can be written directly as terms; this is what permits the local tableau constraints to be expressed with two variables.

We encode computations of the machine from Lemma~\ref{lem:source-machine} using this adaptation. The marked-left-end convention is reflected in the initialization of the tableau.

\subsubsection{An Ordered Cyclic Successor}

The vocabulary contains a binary relation symbol $<$, unary predicates $\mathrm{Min}$ and $\mathrm{Max}$, and the unary function symbol $f$. Let $\CycSucc$ be the conjunction of the following axioms:
\begin{gather}
\exists x : \mathrm{Min}(x), \qquad \exists x : \mathrm{Max}(x)
\tag{\textsc{Existence}}\label{ax:cycsucc-existence}\\
\forall x : \neg(x < x)
\tag{\textsc{Irreflexivity}}\label{ax:cycsucc-irreflexivity}\\
\forall x \forall y : \big( x < y \Rightarrow \neg(y < x) \big)
\tag{\textsc{Asymmetry}}\label{ax:cycsucc-asymmetry}\\
\forall x \forall y : \big( x \neq y \Rightarrow (x < y \vee y < x) \big)
\tag{\textsc{Comparability}}\label{ax:cycsucc-comparability}\\
\forall x \forall y : \big( \mathrm{Min}(x) \Rightarrow (x = y \vee x < y) \big)
\tag{\textsc{Minimum}}\label{ax:cycsucc-minimum}\\
\forall x \forall y : \big( \mathrm{Max}(x) \Rightarrow (x = y \vee y < x) \big)
\tag{\textsc{Maximum}}\label{ax:cycsucc-maximum}\\
\forall x \forall y : \big( (\mathrm{Max}(x) \wedge \mathrm{Min}(y)) \Rightarrow f(x) = y \big)
\tag{\textsc{Wrap}}\label{ax:cycsucc-wrap}\\
\forall x : \big( \neg\mathrm{Max}(x) \Rightarrow x < f(x) \big)
\tag{\textsc{Forward}}\label{ax:cycsucc-forward}\\
\forall x \forall y : \neg\big( \neg\mathrm{Max}(x) \wedge x < y \wedge y < f(x) \big)
\tag{\textsc{No-between}}\label{ax:cycsucc-no-between}\\
\forall x \forall y : \big( (\neg\mathrm{Max}(y) \wedge x < y) \Rightarrow x < f(y) \big)
\tag{\textsc{Propagation}}\label{ax:cycsucc-propagation}.
\end{gather}
Every conjunct uses only the variables $x$ and $y$. Note that transitivity of $<$ is not among the axioms---with two variables we could not even express it---but, as the following lemma shows, it comes for free on finite domains.

\begin{lemma}\label{lem:cycsucc}
Every finite nonempty model of $\CycSucc$ interprets $<$ as a strict linear order whose minimum and maximum are named by $\mathrm{Min}$ and $\mathrm{Max}$, and interprets $f$ as the cyclic successor of this order.
\end{lemma}

\begin{proof}
The two assertions in Axiom~\eqref{ax:cycsucc-existence} provide elements $m$ and $M$ satisfying $\mathrm{Min}$ and $\mathrm{Max}$, respectively. These elements are unique: if distinct elements $m$ and $m'$ both satisfied $\mathrm{Min}$, Axiom~\eqref{ax:cycsucc-minimum} would give both $m<m'$ and $m'<m$; the same argument using Axiom~\eqref{ax:cycsucc-maximum} applies to two distinct elements satisfying $\mathrm{Max}$. Either conclusion contradicts Axiom~\eqref{ax:cycsucc-asymmetry}. Starting with $a_0 = m$, define $a_{i+1} = f(a_i)$ until $M$ is reached. If $a_i \neq M$, then Axiom~\eqref{ax:cycsucc-forward} gives $a_i < a_{i+1}$. Axiom~\eqref{ax:cycsucc-propagation} implies inductively that $a_i < a_j$ whenever $i < j$ and $a_j$ is defined before the first visit to $M$. Hence no element repeats before $M$ is reached, and finiteness forces the sequence to reach $M$. Axiom~\eqref{ax:cycsucc-wrap} gives $f(M) = m$.

It remains to show that the sequence contains every element. Suppose that $z$ is omitted. Axioms~\eqref{ax:cycsucc-minimum} and \eqref{ax:cycsucc-maximum} say that $a_0=m$ is the least element and the final element $M$ is the greatest, so there is a least $j>0$ with $z<a_j$. By the minimality of $j$ and Axiom~\eqref{ax:cycsucc-comparability}, $a_{j-1}<z$. But $a_j=f(a_{j-1})$, contradicting Axiom~\eqref{ax:cycsucc-no-between}. Thus the sequence enumerates the whole domain. Axiom~\eqref{ax:cycsucc-propagation} then gives $a_i<a_j$ for all $i<j$; together with Axioms~\eqref{ax:cycsucc-irreflexivity} and \eqref{ax:cycsucc-asymmetry}, this shows that $<$ is exactly the strict linear order induced by the enumeration and that $f$ is its cyclic successor.
\end{proof}

On a labeled domain $[n]$, there are therefore exactly $n!$ models of $\CycSucc$: we choose a linear order, after which $\mathrm{Min}$, $\mathrm{Max}$ and $f$ are all forced.

\subsubsection{Compiled Tableau Coordinates}

{
Recall the fixed constants $a,b,k$ from Lemma~\ref{lem:source-machine}. Let
\[
I_T=\{0,\ldots,a-1\},\qquad
I_P=\{0,\ldots,b-1\},\qquad
K=\{1,\ldots,k\}.
\]
Relative to a model of $\CycSucc$, the time coordinates and the positions on each tape are, respectively,
\[
\mathcal T_n=I_T\times[n]
\qquad\text{and}\qquad
\mathcal P_n=I_P\times[n],
\]
in lexicographic order, where the second coordinate uses the order defined by $<$. If $z$ is not the maximum element, the successor of $(i,z)$ within either coordinate set is $(i,f(z))$; at the maximum it is $(i+1,f(z))$, provided $i$ is not the last block. The compiled index therefore changes exactly when $f$ wraps from the maximum to the minimum.

For tape $j$, let $\Gamma_j$ be its finite alphabet and put
\[
\Lambda_j=\Gamma_j\times\{L,H,R\}.
\]
A label $(c,X)\in\Lambda_j$ records the tape symbol $c$ and whether the cell lies strictly left of the head, is scanned by the head, or lies strictly right of the head, according as $X=L,H$, or $R$. We also write $X_c$ for $(c,X)$. These pairs index predicate symbols; they are not ordered pairs in the first-order domain.
For $e\in I_T$, $j\in K$, $r\in I_P$, and $\gamma\in\Lambda_j$, introduce a binary predicate
\[
T_{e,j,r,\gamma}(x,y),
\]
whose intended meaning is that tape $j$ has label $\gamma$ at time $(e,x)$ and position $(r,y)$. Taken together, these predicates describe, for each tape $j$, a rectangular grid whose rows are indexed by the time coordinates $\mathcal T_n$ and whose columns are indexed by the tape positions $\mathcal P_n$. For a fixed time $t\in\mathcal T_n$, we call the sequence of labels obtained by reading the corresponding grid row from left to right the \emph{tape row of tape $j$ at time $t$}. It records the contents of tape $j$ and the position of its head in the configuration at time $t$.

We also introduce unary predicates $S_{e,q}(x)$ for the global state and $D_{e,\delta}(x)$ for the transition rule selected at a nonfinal time.

Thus, $\mathrm{Max}(x)$ and $\mathrm{Max}(y)$ mark only the end of the current $n$-element block. The indices $e$ and $r$, encoded in the predicate symbols, determine the block; the final time and tape coordinates are therefore $(a-1,\mathrm{Max})$ and $(b-1,\mathrm{Max})$, respectively.

To treat the tape endpoints uniformly, use two auxiliary boundary symbols $\triangleleft_j,\triangleright_j\notin\Lambda_j$, representing a missing left or right neighbor.\footnote{These symbols are metanotational: they are neither domain elements nor symbols of the vocabulary of $\Phi$, so they require no first-order axioms. In a boundary clause, the atom for the missing neighbor is omitted, while $\mathrm{Min}$, $\mathrm{Max}$, and the compiled block index identify the endpoint. Once the finite schemes are expanded, neither boundary symbol occurs in $\Phi$.} For fixed $j\in K$ and $\delta\in\Delta$, define
\[
\NewLabel{j}{\delta}\bigl(\alpha,(c,X),\gamma\bigr)=(c',X'),
\]
where $\alpha$ and $\gamma$ are the current labels of the left and right neighbors, with a boundary symbol used when one is missing. On a neighborhood that can occur in a row of the form $L^*HR^*$, set $c'=w_{\delta,j}$ if $X=H$, and set $c'=c$ otherwise. If $d_{\delta,j}=0$, set $X'=X$. If $d_{\delta,j}=1$, set $X'=H$ when $\alpha=(c_\ell,H)$ for some $c_\ell\in\Gamma_j$, set $X'=L$ when $X=H$, and set $X'=X$ otherwise. If $d_{\delta,j}=-1$, set $X'=H$ when $\gamma=(c_r,H)$ for some $c_r\in\Gamma_j$, set $X'=R$ when $X=H$, and set $X'=X$ otherwise. Applying this map independently at every cell yields the successor row. Define it arbitrarily on the remaining triples, which the one-head clauses below exclude.

We now give the coordinate compilation explicitly. For a finite nonempty set $\Theta$ of formulas, write
\[
\mathsf{One}(\Theta):=
\left(\bigvee_{\theta\in\Theta}\theta\right)
\wedge
\bigwedge_{\substack{\theta,\theta'\in\Theta\\\theta\neq\theta'}}
\neg(\theta\wedge\theta').
\]
This is only an abbreviation for a fixed finite formula.

\begin{lemma}[Coordinate compilation]\label{lem:coordinate-compilation}
For the fixed machine $U$, the conditions that every time has a unique state, every nonfinal time has a unique selected transition rule, every tape cell has a unique label, every tape row contains exactly one head marker, and successive configurations obey the local update determined by the selected rule can be represented by a fixed $\FOkef{2}$ sentence. This sentence covers every time and tape coordinate, including block boundaries and tape endpoints, for every $n\geq1$.
\end{lemma}

\begin{proof}
We give the defining formulas explicitly. The unique-state, unique-rule, and unique-tape-label clauses are
\begin{align}
&\bigwedge_{e\in I_T}\ \forall x:\,
\mathsf{One}\bigl(\{S_{e,q}(x):q\in Q\}\bigr),
\label{eq:unique-state}\\
&\bigwedge_{0\leq e<a-1}\ \forall x:\,
\mathsf{One}\bigl(\{D_{e,\delta}(x):\delta\in\Delta\}\bigr),
\notag\\[-2mm]
&\qquad{}\wedge\
\forall x:\bigl(\neg\mathrm{Max}(x)\Rightarrow
\mathsf{One}\bigl(\{D_{a-1,\delta}(x):\delta\in\Delta\}\bigr)\bigr)
\notag\\
&\qquad{}\wedge\
\forall x:\bigl(\mathrm{Max}(x)\Rightarrow
\bigwedge_{\delta\in\Delta}\neg D_{a-1,\delta}(x)\bigr),
\label{eq:unique-rule}\\
&\bigwedge_{\substack{e\in I_T,\ j\in K\\r\in I_P}}\ 
\forall x\forall y:\,
\mathsf{One}\bigl(\{T_{e,j,r,\gamma}(x,y):\gamma\in\Lambda_j\}\bigr).
\label{eq:unique-tape-label}
\end{align}
Every conjunction is finite because all displayed index sets belong to the fixed machine.

We next enforce that every tape row contains exactly one head marker. Reading only the $L/H/R$ markers from left to right, a tape row must have the form
\[
L^*HR^*.
\]
Equivalently, the adjacent marker pairs
\[
LR,\qquad HL,\qquad HH,\qquad RL,\qquad RH
\]
are forbidden. For each forbidden pair $XY$ in this list, each $c,d\in\Gamma_j$, and every $e\in I_T$ and $j\in K$, include
\begin{align}
&\bigwedge_{r\in I_P}\ \forall x\forall y:\,
\neg\bigl(\neg\mathrm{Max}(y)\wedge
T_{e,j,r,X_c}(x,y)\wedge
T_{e,j,r,Y_d}(x,f(y))\bigr),
\notag\\
&{}\wedge
\bigwedge_{0\leq r<b-1}\ \forall x\forall y:\,
\neg\bigl(\mathrm{Max}(y)\wedge
T_{e,j,r,X_c}(x,y)\wedge
T_{e,j,r+1,Y_d}(x,f(y))\bigr),
\label{eq:adjacent-labels}
\end{align}
where $X_c$ and $Y_d$ denote the corresponding labels in $\Lambda_j$. The two lines forbid these pairs, respectively, within a position block and across a block boundary. The endpoint clauses are
\begin{equation}
\bigwedge_{\substack{e\in I_T,\ j\in K\\c\in\Gamma_j}}
\forall x\forall y:
\left[
\bigl(\mathrm{Min}(y)\Rightarrow\neg T_{e,j,0,R_c}(x,y)\bigr)
\wedge
\bigl(\mathrm{Max}(y)\Rightarrow\neg T_{e,j,b-1,L_c}(x,y)\bigr)
\right].
\label{eq:row-endpoints}
\end{equation}
The only adjacent marker pairs not forbidden by \eqref{eq:adjacent-labels} are $LL,LH,HR,RR$. Hence, once an $H$ marker occurs, only $R$ markers can follow, so a second $H$ is impossible. Equation~\eqref{eq:row-endpoints} forbids an $R$ marker at the first tape position $(0,\mathrm{Min})$ and an $L$ marker at the last tape position $(b-1,\mathrm{Max})$. In a row with no $H$, the only allowed adjacent pairs are $LL$ and $RR$, so the row would have to consist entirely of $L$ markers or entirely of $R$ markers; the right-end restriction excludes the all-$L$ row, while the left-end restriction excludes the all-$R$ row. Together with \eqref{eq:unique-tape-label}, the clauses therefore force every tape row to have exactly the form $L^*HR^*$.

We now write the local-update sentences directly. In the formulas below, an expression such as
$T_{e,j,r,\NewLabel{j}{\delta}(\alpha,\beta,\gamma)}$ denotes the predicate indexed by the fixed pair returned by the metalevel map $\NewLabel{j}{\delta}$; neither the map nor the pair is an object-language term.

First consider cells whose local tape window lies within a single position block. The two conjunctions below treat, respectively, a successor time within the current time block and across a time-block boundary. Include
\begin{align}
&\bigwedge_{\substack{e\in I_T,\ j\in K,\ r\in I_P,\ \delta\in\Delta\\
\alpha,\beta,\gamma\in\Lambda_j}}
\forall x\forall y:\,
\Bigl(
\neg\mathrm{Max}(x)\wedge\neg\mathrm{Max}(y)\wedge\neg\mathrm{Max}(f(y))
\notag\\[-1mm]
&\hspace{17mm}{}\wedge D_{e,\delta}(x)
\wedge T_{e,j,r,\alpha}(x,y)
\wedge T_{e,j,r,\beta}(x,f(y))
\wedge T_{e,j,r,\gamma}(x,f^2(y))
\Bigr)
\notag\\[-1mm]
&\hspace{49mm}{}\Rightarrow
T_{e,j,r,{\NewLabel{j}{\delta}(\alpha,\beta,\gamma)}}(f(x),f(y))
\notag\\
&{}\wedge
\bigwedge_{\substack{0\leq e<a-1,\ j\in K,\ r\in I_P,\ \delta\in\Delta\\
\alpha,\beta,\gamma\in\Lambda_j}}
\forall x\forall y:\,
\Bigl(
\mathrm{Max}(x)\wedge\neg\mathrm{Max}(y)\wedge\neg\mathrm{Max}(f(y))
\notag\\[-1mm]
&\hspace{17mm}{}\wedge D_{e,\delta}(x)
\wedge T_{e,j,r,\alpha}(x,y)
\wedge T_{e,j,r,\beta}(x,f(y))
\wedge T_{e,j,r,\gamma}(x,f^2(y))
\Bigr)
\notag\\[-1mm]
&\hspace{49mm}{}\Rightarrow
T_{e+1,j,r,{\NewLabel{j}{\delta}(\alpha,\beta,\gamma)}}(f(x),f(y)).
\label{eq:update-no-position-wrap}
\end{align}
If the right neighbor lies in the next position block, include
\begin{align}
&\bigwedge_{\substack{e\in I_T,\ j\in K,\ 0\leq r<b-1,\ \delta\in\Delta\\
\alpha,\beta,\gamma\in\Lambda_j}}
\forall x\forall y:\,
\Bigl(
\neg\mathrm{Max}(x)\wedge\neg\mathrm{Max}(y)\wedge\mathrm{Max}(f(y))
\notag\\[-1mm]
&\hspace{17mm}{}\wedge D_{e,\delta}(x)
\wedge T_{e,j,r,\alpha}(x,y)
\wedge T_{e,j,r,\beta}(x,f(y))
\wedge T_{e,j,r+1,\gamma}(x,f^2(y))
\Bigr)
\notag\\[-1mm]
&\hspace{49mm}{}\Rightarrow
T_{e,j,r,{\NewLabel{j}{\delta}(\alpha,\beta,\gamma)}}(f(x),f(y))
\notag\\
&{}\wedge
\bigwedge_{\substack{0\leq e<a-1,\ j\in K,\ 0\leq r<b-1,\ \delta\in\Delta\\
\alpha,\beta,\gamma\in\Lambda_j}}
\forall x\forall y:\,
\Bigl(
\mathrm{Max}(x)\wedge\neg\mathrm{Max}(y)\wedge\mathrm{Max}(f(y))
\notag\\[-1mm]
&\hspace{17mm}{}\wedge D_{e,\delta}(x)
\wedge T_{e,j,r,\alpha}(x,y)
\wedge T_{e,j,r,\beta}(x,f(y))
\wedge T_{e,j,r+1,\gamma}(x,f^2(y))
\Bigr)
\notag\\[-1mm]
&\hspace{49mm}{}\Rightarrow
T_{e+1,j,r,{\NewLabel{j}{\delta}(\alpha,\beta,\gamma)}}(f(x),f(y)).
\label{eq:update-wrap-after-cell}
\end{align}
If the left neighbor lies in the preceding position block, include
\begin{align}
&\bigwedge_{\substack{e\in I_T,\ j\in K,\ 0\leq r<b-1,\ \delta\in\Delta\\
\alpha,\beta,\gamma\in\Lambda_j}}
\forall x\forall y:\,
\Bigl(
\neg\mathrm{Max}(x)\wedge\mathrm{Max}(y)\wedge\neg\mathrm{Max}(f(y))
\notag\\[-1mm]
&\hspace{17mm}{}\wedge D_{e,\delta}(x)
\wedge T_{e,j,r,\alpha}(x,y)
\wedge T_{e,j,r+1,\beta}(x,f(y))
\wedge T_{e,j,r+1,\gamma}(x,f^2(y))
\Bigr)
\notag\\[-1mm]
&\hspace{49mm}{}\Rightarrow
T_{e,j,r+1,{\NewLabel{j}{\delta}(\alpha,\beta,\gamma)}}(f(x),f(y))
\notag\\
&{}\wedge
\bigwedge_{\substack{0\leq e<a-1,\ j\in K,\ 0\leq r<b-1,\ \delta\in\Delta\\
\alpha,\beta,\gamma\in\Lambda_j}}
\forall x\forall y:\,
\Bigl(
\mathrm{Max}(x)\wedge\mathrm{Max}(y)\wedge\neg\mathrm{Max}(f(y))
\notag\\[-1mm]
&\hspace{17mm}{}\wedge D_{e,\delta}(x)
\wedge T_{e,j,r,\alpha}(x,y)
\wedge T_{e,j,r+1,\beta}(x,f(y))
\wedge T_{e,j,r+1,\gamma}(x,f^2(y))
\Bigr)
\notag\\[-1mm]
&\hspace{49mm}{}\Rightarrow
T_{e+1,j,r+1,{\NewLabel{j}{\delta}(\alpha,\beta,\gamma)}}(f(x),f(y)).
\label{eq:update-wrap-before-cell}
\end{align}
When $n=1$, the two neighbors of an interior cell lie in the preceding and following position blocks. Include
\begin{align}
&\bigwedge_{\substack{e\in I_T,\ j\in K,\ 0\leq r<b-2,\ \delta\in\Delta\\
\alpha,\beta,\gamma\in\Lambda_j}}
\forall x\forall y:\,
\Bigl(
\neg\mathrm{Max}(x)\wedge\mathrm{Max}(y)\wedge\mathrm{Max}(f(y))
\notag\\[1mm]
&\hspace{17mm}{}\wedge D_{e,\delta}(x)
\wedge T_{e,j,r,\alpha}(x,y)
\wedge T_{e,j,r+1,\beta}(x,f(y))
\wedge T_{e,j,r+2,\gamma}(x,f^2(y))
\Bigr)
\notag\\[-1mm]
&\hspace{49mm}{}\Rightarrow
T_{e,j,r+1,{\NewLabel{j}{\delta}(\alpha,\beta,\gamma)}}(f(x),f(y))
\notag\\
&{}\wedge
\bigwedge_{\substack{0\leq e<a-1,\ j\in K,\ 0\leq r<b-2,\ \delta\in\Delta\\
\alpha,\beta,\gamma\in\Lambda_j}}
\forall x\forall y:\,
\Bigl(
\mathrm{Max}(x)\wedge\mathrm{Max}(y)\wedge\mathrm{Max}(f(y))
\notag\\[1mm]
&\hspace{17mm}{}\wedge D_{e,\delta}(x)
\wedge T_{e,j,r,\alpha}(x,y)
\wedge T_{e,j,r+1,\beta}(x,f(y))
\wedge T_{e,j,r+2,\gamma}(x,f^2(y))
\Bigr)
\notag\\[-1mm]
&\hspace{49mm}{}\Rightarrow
T_{e+1,j,r+1,{\NewLabel{j}{\delta}(\alpha,\beta,\gamma)}}(f(x),f(y)).
\label{eq:update-two-position-wraps}
\end{align}

At the left endpoint, when $n>1$, include
\begin{align}
&\bigwedge_{\substack{e\in I_T,\ j\in K,\ \delta\in\Delta\\
\beta,\gamma\in\Lambda_j}}
\forall x\forall y:\,
\Bigl(
\neg\mathrm{Max}(x)\wedge\mathrm{Min}(y)\wedge\neg\mathrm{Max}(y)
\notag\\[-1mm]
&\hspace{17mm}{}\wedge D_{e,\delta}(x)
\wedge T_{e,j,0,\beta}(x,y)
\wedge T_{e,j,0,\gamma}(x,f(y))
\Bigr)
\notag\\[-1mm]
&\hspace{49mm}{}\Rightarrow
T_{e,j,0,{\NewLabel{j}{\delta}(\triangleleft_j,\beta,\gamma)}}(f(x),y)
\notag\\
&{}\wedge
\bigwedge_{\substack{0\leq e<a-1,\ j\in K,\ \delta\in\Delta\\
\beta,\gamma\in\Lambda_j}}
\forall x\forall y:\,
\Bigl(
\mathrm{Max}(x)\wedge\mathrm{Min}(y)\wedge\neg\mathrm{Max}(y)
\notag\\[-1mm]
&\hspace{17mm}{}\wedge D_{e,\delta}(x)
\wedge T_{e,j,0,\beta}(x,y)
\wedge T_{e,j,0,\gamma}(x,f(y))
\Bigr)
\notag\\[-1mm]
&\hspace{49mm}{}\Rightarrow
T_{e+1,j,0,{\NewLabel{j}{\delta}(\triangleleft_j,\beta,\gamma)}}(f(x),y).
\label{eq:update-left-endpoint}
\end{align}
For the left endpoint when $n=1$, include
\begin{align}
&\bigwedge_{\substack{e\in I_T,\ j\in K,\ \delta\in\Delta\\
\beta,\gamma\in\Lambda_j}}
\forall x\forall y:\,
\Bigl(
\neg\mathrm{Max}(x)\wedge\mathrm{Min}(y)\wedge\mathrm{Max}(y)
\notag\\[-1mm]
&\hspace{17mm}{}\wedge D_{e,\delta}(x)
\wedge T_{e,j,0,\beta}(x,y)
\wedge T_{e,j,1,\gamma}(x,f(y))
\Bigr)
\notag\\[-1mm]
&\hspace{49mm}{}\Rightarrow
T_{e,j,0,{\NewLabel{j}{\delta}(\triangleleft_j,\beta,\gamma)}}(f(x),y)
\notag\\
&{}\wedge
\bigwedge_{\substack{0\leq e<a-1,\ j\in K,\ \delta\in\Delta\\
\beta,\gamma\in\Lambda_j}}
\forall x\forall y:\,
\Bigl(
\mathrm{Max}(x)\wedge\mathrm{Min}(y)\wedge\mathrm{Max}(y)
\notag\\[-1mm]
&\hspace{17mm}{}\wedge D_{e,\delta}(x)
\wedge T_{e,j,0,\beta}(x,y)
\wedge T_{e,j,1,\gamma}(x,f(y))
\Bigr)
\notag\\[-1mm]
&\hspace{49mm}{}\Rightarrow
T_{e+1,j,0,{\NewLabel{j}{\delta}(\triangleleft_j,\beta,\gamma)}}(f(x),y).
\label{eq:update-left-endpoint-singleton}
\end{align}
At the right endpoint, when $n>1$, include
\begin{align}
&\bigwedge_{\substack{e\in I_T,\ j\in K,\ \delta\in\Delta\\
\alpha,\beta\in\Lambda_j}}
\forall x\forall y:\,
\Bigl(
\neg\mathrm{Max}(x)\wedge\neg\mathrm{Max}(y)\wedge\mathrm{Max}(f(y))
\notag\\[-1mm]
&\hspace{17mm}{}\wedge D_{e,\delta}(x)
\wedge T_{e,j,b-1,\alpha}(x,y)
\wedge T_{e,j,b-1,\beta}(x,f(y))
\Bigr)
\notag\\[-1mm]
&\hspace{49mm}{}\Rightarrow
T_{e,j,b-1,{\NewLabel{j}{\delta}(\alpha,\beta,\triangleright_j)}}(f(x),f(y))
\notag\\
&{}\wedge
\bigwedge_{\substack{0\leq e<a-1,\ j\in K,\ \delta\in\Delta\\
\alpha,\beta\in\Lambda_j}}
\forall x\forall y:\,
\Bigl(
\mathrm{Max}(x)\wedge\neg\mathrm{Max}(y)\wedge\mathrm{Max}(f(y))
\notag\\[-1mm]
&\hspace{17mm}{}\wedge D_{e,\delta}(x)
\wedge T_{e,j,b-1,\alpha}(x,y)
\wedge T_{e,j,b-1,\beta}(x,f(y))
\Bigr)
\notag\\[-1mm]
&\hspace{49mm}{}\Rightarrow
T_{e+1,j,b-1,{\NewLabel{j}{\delta}(\alpha,\beta,\triangleright_j)}}(f(x),f(y)).
\label{eq:update-right-endpoint}
\end{align}
For the right endpoint when $n=1$, include
\begin{align}
&\bigwedge_{\substack{e\in I_T,\ j\in K,\ \delta\in\Delta\\
\alpha,\beta\in\Lambda_j}}
\forall x\forall y:\,
\Bigl(
\neg\mathrm{Max}(x)\wedge\mathrm{Min}(y)\wedge\mathrm{Max}(y)
\notag\\[-1mm]
&\hspace{17mm}{}\wedge D_{e,\delta}(x)
\wedge T_{e,j,b-2,\alpha}(x,y)
\wedge T_{e,j,b-1,\beta}(x,f(y))
\Bigr)
\notag\\[-1mm]
&\hspace{49mm}{}\Rightarrow
T_{e,j,b-1,{\NewLabel{j}{\delta}(\alpha,\beta,\triangleright_j)}}(f(x),f(y))
\notag\\
&{}\wedge
\bigwedge_{\substack{0\leq e<a-1,\ j\in K,\ \delta\in\Delta\\
\alpha,\beta\in\Lambda_j}}
\forall x\forall y:\,
\Bigl(
\mathrm{Max}(x)\wedge\mathrm{Min}(y)\wedge\mathrm{Max}(y)
\notag\\[-1mm]
&\hspace{17mm}{}\wedge D_{e,\delta}(x)
\wedge T_{e,j,b-2,\alpha}(x,y)
\wedge T_{e,j,b-1,\beta}(x,f(y))
\Bigr)
\notag\\[-1mm]
&\hspace{49mm}{}\Rightarrow
T_{e+1,j,b-1,{\NewLabel{j}{\delta}(\alpha,\beta,\triangleright_j)}}(f(x),f(y)).
\label{eq:update-right-endpoint-singleton}
\end{align}
The four displays \eqref{eq:update-no-position-wrap}--\eqref{eq:update-two-position-wraps} cover all interior tape cells, including every possible position-block wrap. The last four displays cover the two endpoints, separately for $n>1$ and $n=1$. The assumption $b\geq3$ makes every displayed block index valid. Thus the eight displays cover every tape cell exactly once at every nonfinal time, with no implicit additional cases.
Because \eqref{eq:unique-tape-label} assigns exactly one label to each cell at the successor time, requiring the pair returned by $\NewLabel{j}{\delta}$ is equivalent to excluding every incorrect successor label.

Finally, the following clauses require state $q_0$ at the first time coordinate $(0,\mathrm{Min})$ and state $q_{\mathrm{acc}}$ at the last time coordinate $(a-1,\mathrm{Max})$:
\begin{equation}
\forall x:\bigl(\mathrm{Min}(x)\Rightarrow S_{0,q_0}(x)\bigr),
\qquad
\forall x:\bigl(\mathrm{Max}(x)\Rightarrow S_{a-1,q_{\mathrm{acc}}}(x)\bigr).
\label{eq:time-endpoints}
\end{equation}
Likewise, \eqref{eq:row-endpoints} selects the first tape cell by block index $0$ and $\mathrm{Min}(y)$ and the last by block index $b-1$ and $\mathrm{Max}(y)$. Every displayed sentence uses only $x$ and $y$. In each of \eqref{eq:update-no-position-wrap}--\eqref{eq:update-right-endpoint-singleton}, the first finite conjunction applies exactly when the successor time remains in the current block, while the second applies exactly when it wraps to the next block. The restriction $e<a-1$ in the second conjunction excludes the final time coordinate $(a-1,\mathrm{Max})$. This remains true for $n=1$, when $f(x)=x$: the changing predicate indices still distinguish successive compiled coordinates.
\end{proof}

We make the initial-row clauses explicit as well. Renumber the tapes so that tape $1$ is the input tape. Let $\sqcup_j$ be the blank symbol on tape $j$, and put $\iota_1=1$ and $\iota_j=\sqcup_j$ for $j>1$. Let $\bar\iota_j$ denote the marked version of $\iota_j$ used at the left endpoint. In addition to the first formula in \eqref{eq:time-endpoints}, include
\begin{multline}
\bigwedge_{j\in K}\Bigg[
\forall x\forall y:
\bigl((\mathrm{Min}(x)\wedge\mathrm{Min}(y))
\Rightarrow T_{0,j,0,H_{\bar\iota_j}}(x,y)\bigr)\\
{}\wedge
\forall x\forall y:
\bigl((\mathrm{Min}(x)\wedge\neg\mathrm{Min}(y))
\Rightarrow T_{0,j,0,R_{\iota_j}}(x,y)\bigr)\\
{}\wedge
\bigwedge_{1\leq r<b}\ \forall x\forall y:
\bigl(\mathrm{Min}(x)\Rightarrow
T_{0,j,r,R_{\sqcup_j}}(x,y)\bigr)
\Bigg].
\label{eq:initial-row}
\end{multline}
Thus the input tape contains $1^n$ in its first position block with its first symbol marked and blanks thereafter, every work tape is blank with its leftmost blank marked, and every head is at its marked left endpoint. Formula~\eqref{eq:unique-rule} makes the rule predicate at the initial row record the first rule of the branch and permits no rule predicate on the final row.
}

\subsubsection{Local Transitions}

{
Recall that $(e,x)\in I_T\times[n]$ is a time coordinate and $(r,y)\in I_P\times[n]$ is a tape position. The predicates $S_{e,q}(x)$, $D_{e,\delta}(x)$, and $T_{e,j,r,H_c}(x,y)$ state, respectively, that the state at time $(e,x)$ is $q$, that rule $\delta$ is selected at that time, and that the head on tape $j$ scans symbol $c$ at position $(r,y)$. For $\delta\in\Delta$, recall that $q^-_\delta$ and $q^+_\delta$ are its source and target states, and $r_{\delta,j}$ is the symbol it reads on tape $j$. The successor of $(e,x)$ is $(e,f(x))$ when $x$ is not maximal and $(e+1,f(x))$ when $x$ is maximal and $e<a-1$.

Fix a transition rule $\delta\in\Delta$. Its source- and target-state clauses are
\begin{align}
&\bigwedge_{e\in I_T}\ \forall x:
\bigl(D_{e,\delta}(x)\Rightarrow S_{e,q^-_\delta}(x)\bigr),
\label{eq:control-source}\\
&\bigwedge_{e\in I_T}\ \forall x:
\bigl((\neg\mathrm{Max}(x)\wedge D_{e,\delta}(x))
\Rightarrow S_{e,q^+_\delta}(f(x))\bigr)
\notag\\
&\qquad{}\wedge
\bigwedge_{0\leq e<a-1}\ \forall x:
\bigl((\mathrm{Max}(x)\wedge D_{e,\delta}(x))
\Rightarrow S_{e+1,q^+_\delta}(f(x))\bigr).
\label{eq:control-target}
\end{align}
For every tape, a wrong scanned symbol is excluded by
\begin{equation}
\bigwedge_{\substack{e\in I_T,\ j\in K,\ r\in I_P\\
c\in\Gamma_j\setminus\{r_{\delta,j}\}}}
\forall x\forall y:
\neg\bigl(D_{e,\delta}(x)\wedge T_{e,j,r,H_c}(x,y)\bigr).
\label{eq:control-symbol}
\end{equation}
The unique $H$-label on each row then forces the scanned symbol to be
$r_{\delta,j}$. The same unary predicate $D_{e,\delta}(x)$ occurs in the control clauses for every tape and thereby synchronizes their updates.

The tape-update clauses are exactly the finite conjunctions in
\eqref{eq:update-no-position-wrap}--\eqref{eq:update-right-endpoint-singleton}. Hence every nonfinal time and every tape cell is covered by one displayed clause, including every combination of a time wrap, a position wrap, a tape boundary, and $n=1$. No transition clause is imposed after the final time.

Let $\Phi$ be the conjunction of $\CycSucc$, the state, transition-choice, unique-label, one-head, initial, accepting, control, and local-update clauses. The vocabulary and the sentence are fixed because $U,k,a,b$ and all alphabets are fixed, and every clause uses at most the variables $x$ and $y$.

\begin{lemma}[Local simulation]\label{lem:local-simulation}
Fix a model of $\CycSucc$ and a nonfinal time $t$, and consider the configurations represented at $t$ and at its successor, assuming the unique-label and one-head clauses. The control and local-update clauses relating these configurations hold if and only if the unique predicate $D_{e,\delta}$ at $t$ selects an enabled transition rule $\delta$ and the successor configuration is exactly the result of applying $\delta$ to the configuration at $t$; in particular, its global state is $q^+_\delta$.
\end{lemma}

\begin{proof}
The control clauses check the source state and, using the unique head on each tape, all $k$ scanned symbols. They also fix the successor state. Once $\delta$ is fixed, the rule determines the transition of each tape. For every tape cell, exactly one compiled local-update clause compares its label at the successor time with the value of $\NewLabel{j}{\delta}$ on its current label and the current labels of its immediate neighbors. Hence all labels in the successor configuration agree with the result of applying $\delta$. Conversely, applying an enabled rule $\delta$ satisfies every control and local-update clause by the definition of $\NewLabel{j}{\delta}$. Lemma~\ref{lem:coordinate-compilation} ensures that this argument includes time and position-block wraps, both tape endpoints, and the case $n=1$; none of these cases is left to an unspecified ``analogous'' clause.
\end{proof}
}

\subsubsection{Correctness and Counting}

\begin{lemma}\label{lem:tableau}
Fix a model of $\CycSucc$ on $[n]$. The satisfying interpretations of the tableau predicates are in bijection with the accepting computation branches of $U$ on the input $1^n$.
\end{lemma}

\begin{proof}
A satisfying interpretation supplies $an$ successive configurations, each consisting of a global state and $k$ tape rows of length $bn$. The initial clauses fix the first configuration. At every nonfinal time there is exactly one selected transition rule; Lemma~\ref{lem:local-simulation} makes the next configuration exactly its legal successor. The final-state clause forces acceptance. Thus every model determines an accepting branch.

Conversely, an accepting branch determines every state, selected transition rule, and tape label, and therefore yields one satisfying interpretation. No transition predicate is permitted at the final time, so there is no extra multiplicity there. Distinct branches select different transition rules at some time and hence differ in the corresponding predicate $D_{e,\delta}$. Once a branch is fixed, these predicates are fixed as well and introduce no additional multiplicity.
\end{proof}

\begin{proof}[Proof of Theorem~\ref{thm:hardness}]
By Lemma~\ref{lem:tableau}, each cyclic-order skeleton on $[n]$ admits exactly $\acc_U(n)$ satisfying tableau expansions. By Lemma~\ref{lem:cycsucc} there are exactly $n!$ skeletons, so
\[
\FOMC(\Phi, n) = n! \cdot \acc_U(n).
\]
An oracle for $\FOMC_\Phi$ therefore computes $\acc_U(n)$ after one query and exact division by $n!$, which proves $\PPone$-hardness. Membership in $\PPone$ holds for every fixed first-order sentence, as noted in Section~\ref{sec:background}, and hence $\FOMC_\Phi$ is $\PPone$-complete.
\end{proof}

\subsection{Proof of the Two-Function Hardness Theorem}\label{sec:prooftwofun}

We now prove Theorem~\ref{thm:twofun-hardness}, using the source machine fixed in Section~\ref{sec:common-machine}.

\subsubsection{Slowing the common source machine}\label{sec:twofun-slow}

Take the machine $U$ and constant $a$ from Lemma~\ref{lem:source-machine}. For the present proof it is convenient to have exactly $2an$ machine transitions. Replace every transition of $U$ by two transitions through a private intermediate state and add two deterministic no-op transitions, one before and one after the slowed simulation. The resulting fixed machine $V$ has
\begin{equation}\label{eq:Vlength}
2(an-1)+2=2an
\end{equation}
transitions on every branch and
\begin{equation}\label{eq:Vcount}
\acc_V(n)=\acc_U(n).
\end{equation}
Write $\Delta_V$ for its fixed transition set, $q_0$ for its initial state, and $q_{\mathrm{acc}}$ for its accepting state. These are meta-level names in the finite control, not constant symbols in the logical vocabulary. Distinct machine choices remain represented by distinct named transitions. By Lemma~\ref{lem:source-machine}, the first cell of every tape carries a marked version of its symbol, marked symbols are preserved when written, and no transition moves left while scanning a marked symbol.

\subsubsection{Rooted permutation components}\label{sec:twofun-root}

We introduce the structural vocabulary as it is needed.  Start with unary functions $S,P,\pi$.

\paragraph{Permutation cycles and component roots.}

Require
\begin{equation}\label{eq:SP}
  \forall x\,\bigl(P(S(x))=x\wedge S(P(x))=x\bigr).
\end{equation}
Thus $S$ is a permutation and $P=S^{-1}$.  Next require
\begin{equation}\label{eq:pi}
  \forall x\,\bigl(\pi(S(x))=\pi(x)\wedge \pi(\pi(x))=\pi(x)\bigr).
\end{equation}
The first equality makes $\pi$ constant on every $S$-cycle.  If this constant value is $c$, the second gives $\pi(c)=c$.

For each fixed point $c$ of $\pi$, let
\[
  B_c=\{x:\pi(x)=c\}.
\]
We call $B_c$ a \emph{$\pi$-component}.  The $\pi$-components partition the domain.  The $S$-cycle containing $c$ is the \emph{root cycle} of this component; every other $S$-cycle in $B_c$ is a \emph{satellite}.  Within $B_c$, the equality $x=\pi(x)$ holds at exactly one point, namely $c$.

\paragraph{Odd roots and even satellites.}

We use a unary predicate $E$ to distinguish root cycles from satellite cycles.  On every satellite, $E$ will alternate along each $S$-edge.  On a root cycle, the edge from the distinguished point $c$ to $S(c)$ is the unique edge on which this alternation is not required.  Impose
\begin{align}
  \forall x\,\bigl(x=\pi(x)&\Longrightarrow \neg E(x)\wedge\neg E(S(x))\bigr),
  \label{eq:Eroot}\\
  \forall x\,\bigl(x\ne\pi(x)&\Longrightarrow (E(S(x))\leftrightarrow\neg E(x))\bigr).
  \label{eq:Etoggle}
\end{align}

\begin{lemma}[Parity separation]\label{lem:parity}
In every model of \eqref{eq:SP}--\eqref{eq:Etoggle}, every root cycle has odd length and every satellite has even length.  On a fixed odd root cycle the predicate $E$ is uniquely determined, while a fixed even satellite has exactly two possible $E$-colorings.
\end{lemma}

\begin{proof}
Let $c$ be the unique point of a $\pi$-component with $c=\pi(c)$.  The edge $c\to S(c)$ is the unique root edge on which \eqref{eq:Etoggle} is not imposed; \eqref{eq:Eroot} fixes both endpoints to $E=0$.  Every other edge toggles $E$.  Returning from $S(c)$ to $c$ uses $r-1$ toggling edges when the root length is $r$, so consistency is equivalent to $r-1$ being even, i.e. to $r$ being odd.  The coloring is then forced.

On a satellite no point satisfies $x=\pi(x)$, so every edge toggles.  Such a coloring exists exactly for even cycle length and then has the two choices obtained by choosing the color of one point.
\end{proof}

\paragraph{Defining the root cycle without reachability.}

Introduce a unary predicate $A$, intended to mark exactly the root cycle, and unary functions $d,h$.  First require
\begin{equation}\label{eq:Aprop}
  \forall x\,\bigl((x=\pi(x)\Longrightarrow A(x))\wedge
  (A(S(x))\leftrightarrow A(x))\bigr).
\end{equation}
Thus the root cycle is active and every satellite is either wholly active or wholly inactive.

On active points require $d$ and $h$ to be mutually inverse bijections that remain in the same $\pi$-component:
\begin{align}
  \forall x\,\bigl(A(x)&\Longrightarrow A(d(x))\wedge A(h(x))\bigr),
  \label{eq:dhactive}\\
  \forall x\,\bigl(A(x)&\Longrightarrow h(d(x))=x\wedge d(h(x))=x\bigr),
  \label{eq:dhinv}\\
  \forall x\,\bigl(&\pi(d(x))=\pi(x)\wedge \pi(h(x))=\pi(x)\bigr).
  \label{eq:dhlocal}
\end{align}
We also impose the doubling law on active points,
\begin{equation}\label{eq:doubling}
  \forall x\,\bigl(A(x)\Longrightarrow d(S(x))=S(S(d(x)))\bigr),
\end{equation}
and anchor $d$ at the distinguished point of each component:
\begin{equation}\label{eq:danchor}
  \forall x\,\bigl(x=\pi(x)\Longrightarrow d(x)=x\bigr).
\end{equation}
On inactive points the otherwise unused values are fixed:
\begin{equation}\label{eq:dhdefault}
  \forall x\,\bigl(\neg A(x)\Longrightarrow
  (d(x)=\pi(x)\wedge h(x)=\pi(x))\bigr).
\end{equation}

\begin{lemma}[Root isolation by squaring]\label{lem:root-isolation}
In every finite model of \eqref{eq:SP}--\eqref{eq:dhdefault}, $A$ is exactly the root $S$-cycle of each $\pi$-component.
\end{lemma}

\begin{proof}
Fix one $\pi$-component and let $X$ be its active set.  By \eqref{eq:Aprop}, $X$ is a union of entire $S$-cycles and contains the root cycle.  Equations \eqref{eq:dhactive}--\eqref{eq:dhinv} make $d$ a bijection $X\to X$ with inverse $h$, while \eqref{eq:doubling} says
\[
  d\circ S=S^2\circ d
  \qquad\text{on }X.
\]
Since $d$ is a bijection, this is equivalently
\[
  S^2|_X=d\circ(S|_X)\circ d^{-1}.
\]
Thus $S|_X$ and $S^2|_X$ are \emph{conjugate permutations}: the bijection $d$ only relabels the elements, and in particular maps the cycles of $S|_X$ bijectively to those of $S^2|_X$.  Hence the two permutations have the same cycle structure.

By {Lemma~\ref{lem:parity}}, the root cycle is odd and every satellite is even.  Squaring an odd cycle leaves one cycle of the same length, whereas squaring an even cycle splits it into two cycles.  If $k$ satellites were active, $S|_X$ would therefore have $1+k$ cycles while $S^2|_X$ would have $1+2k$.  Hence $k=0$.  The root is active, and no satellite is.
\end{proof}

\paragraph{A canonical linear interval and two internal markers.}

Fix a $\pi$-component whose root has odd length $r$, and let $c$ be its unique point with $c=\pi(c)$.  By {Lemma~\ref{lem:root-isolation}} the whole root is active.  From \eqref{eq:danchor} and \eqref{eq:doubling},
\begin{equation}\label{eq:dexp}
  d(S^i(c))=S^{2i}(c)
  \qquad(0\le i<r).
\end{equation}
Since $r$ is odd, multiplication by $2$ modulo $r$ is invertible, so $d$ is uniquely determined and $h=d^{-1}$ is unique as well.

Define the unary terms
\begin{equation}\label{eq:markers}
  c(x):=\pi(x),
  \qquad
  b(x):=P(\pi(x)),
  \qquad
  q_2(x):=h(P(\pi(x))),
  \qquad
  q_1(x):=h(h(P(\pi(x)))).
\end{equation}
They have the same value for every $x$ in a fixed $\pi$-component.

We will be particularly interested in root lengths of the form
\begin{equation}\label{eq:mN}
  m=4N+1,
  \qquad\text{so that}\qquad
  N=\frac{m-1}{4}.
\end{equation}
The reason for singling out these lengths is that the two terms involving $h$ then identify canonical points one quarter and one half of the way around the root.  Indeed, writing $c,q_1,q_2,b$ for the values of the terms in \eqref{eq:markers},
\begin{equation}\label{eq:quartermarkers}
  q_1=S^N(c),
  \qquad
  q_2=S^{2N}(c),
  \qquad
  b=S^{4N}(c)=P(c).
\end{equation}
This follows because $h$ divides exponents by $2$ modulo the odd number $m=4N+1$.

Finally, cut the root cycle at the edge $b\to c$ and regard
\begin{equation}\label{eq:linear-time}
  c,S(c),\ldots,b
\end{equation}
as a linear interval.  For later use define the abbreviations
\[
  \First(x):\!\Longleftrightarrow A(x)\wedge x=\pi(x),
  \qquad
  \Final(x):\!\Longleftrightarrow A(x)\wedge S(x)=\pi(x),
\]
and $\NF(x):=A(x)\wedge\neg\Final(x)$.

\subsubsection{Preparing the unary input on the root}\label{sec:twofun-prepare}

We now connect the structural parameter $N$, where the root cycle has length $m=4N+1$ as in \eqref{eq:mN}, with the unary input length $n$.  We choose
\begin{equation}\label{eq:Nchoice}
  N=an,
  \qquad\text{and hence}\qquad
  m=4an+1.
\end{equation}
By \eqref{eq:quartermarkers}, the $4an$ successive $S$-edges of the linear root interval are divided by $q_1$ and $q_2$ into intervals of lengths
\[
  \underbrace{an}_{\text{build }1^n}
  \quad+
  \underbrace{an}_{\text{rewind input-tape head}}
  \quad+
  \underbrace{2an}_{\text{run }V}.
\]
The first two intervals deterministically prepare the input of $V$; the last interval has exactly the length required by \eqref{eq:Vlength}.  This is where the enlarged root is used as an actual time line: one root element represents one time instant, so no epoch index is needed.

Introduce unary phase predicates $\cB,\cR,\cM$ for the build, rewind, and machine intervals.  They partition the active root and are false outside it.  Concretely, impose
\begin{align}
  \forall x\,\bigl(A(x)&\Longrightarrow
  (\cB(x)\vee\cR(x)\vee\cM(x))\bigr),\label{eq:phase-cover}\\
  \forall x\,\bigl(&\neg(\cB(x)\wedge\cR(x))\wedge
  \neg(\cB(x)\wedge\cM(x))\wedge
  \neg(\cR(x)\wedge\cM(x))\bigr),\label{eq:phase-disjoint}\\
  \forall x\,\bigl(\neg A(x)&\Longrightarrow
  \neg\cB(x)\wedge\neg\cR(x)\wedge\neg\cM(x)\bigr),\label{eq:phase-offroot}\\
  \forall x\,\bigl(x=\pi(x)&\Longrightarrow\cB(x)\bigr).\label{eq:phase-first}
\end{align}
The phase changes are propagated along nonfinal root edges:
\begin{align}
 \forall x\,\bigl(\NF(x)\wedge\cB(x)\wedge S(x)\ne q_1(x)
 &\Longrightarrow \cB(S(x))\bigr),\label{eq:phase-Bstay}\\
 \forall x\,\bigl(\NF(x)\wedge\cB(x)\wedge S(x)=q_1(x)
 &\Longrightarrow \cR(S(x))\bigr),\label{eq:phase-BR}\\
 \forall x\,\bigl(\NF(x)\wedge\cR(x)\wedge S(x)\ne q_2(x)
 &\Longrightarrow \cR(S(x))\bigr),\label{eq:phase-Rstay}\\
 \forall x\,\bigl(\NF(x)\wedge\cR(x)\wedge S(x)=q_2(x)
 &\Longrightarrow \cM(S(x))\bigr),\label{eq:phase-RM}\\
 \forall x\,\bigl(\NF(x)\wedge\cM(x)
 &\Longrightarrow \cM(S(x))\bigr).\label{eq:phase-Mstay}
\end{align}
Together with the partition clauses, these formulas determine the phase predicates uniquely.  Under \eqref{eq:Nchoice}, the three phase intervals have respectively $an$, $an$, and $2an$ transitions.

The root cycle supplies a logical time point at every $S$-step, but during the build and rewind phases we want one actual tape-head move only once every $a$ logical steps.  Since $a$ is the fixed constant from the source-machine normalization, we can add the fixed family of unary predicates
\[
  Z_0,\ldots,Z_{a-1}
\]
as a modulo-$a$ clock.  Require exactly one $Z_i$ on every active point, no $Z_i$ off the root, and $Z_0$ at the first point.  In formulas,
\begin{align}
  \forall x\,\bigl(A(x)&\Longrightarrow \bigvee_{i=0}^{a-1}Z_i(x)\bigr),\label{eq:Z-cover}\\
  \forall x\,\bigl(Z_i(x)&\Longrightarrow\neg Z_j(x)\bigr)
  \qquad(0\le i<j<a),\label{eq:Z-disjoint}\\
  \forall x\,\bigl(\neg A(x)&\Longrightarrow\bigwedge_{i=0}^{a-1}\neg Z_i(x)\bigr),\label{eq:Z-offroot}\\
  \forall x\,\bigl(x=\pi(x)&\Longrightarrow Z_0(x)\bigr),\label{eq:Z-first}\\
  \forall x\,\bigl(\NF(x)\wedge Z_i(x)&\Longrightarrow Z_{i+1\bmod a}(S(x))\bigr)
  \qquad(0\le i<a).\label{eq:modcounter}
\end{align}
We propagate the predicates $Z_0,\ldots,Z_{a-1}$ only along the linear root interval from $c$ to $b$; no propagation condition is imposed across the closing edge $b\to c$.  Starting from $Z_0(c)$, the rules therefore give, for every $0\le t\le 4N$,
\begin{equation}\label{eq:modcounter-all}
  Z_i(S^t(c))\quad\Longleftrightarrow\quad t\equiv i\pmod a.
\end{equation}
In particular,
\begin{equation}\label{eq:modcounter-target}
  Z_0(S^t(c))\quad\Longleftrightarrow\quad a\mid t.
\end{equation}
Recall that $N=an$ and that
\[
  q_1=S^N(c),\qquad q_2=S^{2N}(c),\qquad b=S^{4N}(c).
\]
Since $N$ is divisible by $a$, the four boundary points $c,q_1,q_2,b$ all satisfy $Z_0$.  More importantly, the nonfinal points of the first $N$-step interval are $S^t(c)$ with $0\le t<N$, and among them $Z_0$ holds exactly at
\[
  t=0,a,2a,\ldots,(n-1)a.
\]
Thus this interval contains exactly $n$ nonfinal $Z_0$-points.  The same is true of the next $N$-step interval, where the corresponding times are
\[
  t=N,N+a,N+2a,\ldots,N+(n-1)a.
\]
The logical preparation below uses precisely these $Z_0$-points to trigger tape-head moves; at the intervening points the relevant head stays in place.  Hence each of the first two intervals performs exactly $n$ such moves.

Before encoding a run of $V$, the first two root intervals deterministically prepare the tape configuration that $V$ expects on input $1^n$.  This preparation is enforced directly by the logical sentence and is not part of the computation of $V$.  The interval from $c$ to $q_1$ constructs $1^n$ on the input tape of $V$, and the interval from $q_1$ to $q_2$ rewinds its head to the first cell.  From $q_2$ onward, the sentence encodes an ordinary run of $V$.

For each tape $\tau$, each symbol $\gamma$ in its fixed alphabet, and each displacement $e\in\{-1,0,+1\}$, use unary action predicates
\[
  \Read^\tau_\gamma,
  \qquad
  \Write^\tau_\gamma,
  \qquad
  \Move^\tau_e.
\]
Write $\Gamma_\tau$ for the fixed alphabet of tape $\tau$. Using the finite abbreviation $\mathsf{One}$ from Section~\ref{sec:proofhardness}, impose, for every tape $\tau$,
\begin{align}
  \forall x\,\bigl(\NF(x)&\Longrightarrow
  \mathsf{One}(\{\Read^\tau_\gamma(x):\gamma\in\Gamma_\tau\})\bigr),
  \label{eq:action-read-one}\\
  \forall x\,\bigl(\NF(x)&\Longrightarrow
  \mathsf{One}(\{\Write^\tau_\gamma(x):\gamma\in\Gamma_\tau\})\bigr),
  \label{eq:action-write-one}\\
  \forall x\,\bigl(\NF(x)&\Longrightarrow
  \mathsf{One}(\{\Move^\tau_e(x):e\in\{-1,0,+1\}\})\bigr).
  \label{eq:action-move-one}
\end{align}
At points off the active root and at the final root point, require all action predicates to be false:
\begin{multline}
  \forall x\,\Bigl((\neg A(x)\vee\Final(x))\Longrightarrow\\
  \bigwedge_{\gamma\in\Gamma_\tau}
  \bigl(\neg\Read^\tau_\gamma(x)\wedge\neg\Write^\tau_\gamma(x)\bigr)
  \wedge
  \bigwedge_{e\in\{-1,0,+1\}}\neg\Move^\tau_e(x)\Bigr).
  \label{eq:action-offroot}
\end{multline}
These are fixed finite universal one-variable clauses. In the deterministic prefix the rules below then fix the unique selected actions; during the machine phase the selected transition of $V$ fixes them through \eqref{eq:D-actions}.

Let $\bar\blank$ denote the marked blank at the left endpoint of a tape and $\bar 1$ the marked version of input symbol $1$ at the first cell of the input tape of $V$.  On every work tape, the logical preparation requires the head to remain at the marked left endpoint, reading and rewriting $\bar\blank$ throughout both preparation intervals.

{For completeness, these deterministic requirements can be written by a fixed finite family of one-variable clauses.  Write $\tau_{\mathrm{in}}$ for the input tape of $V$; this is only a meta-level name for one of the finitely many tapes.  For every work tape $\tau$, impose
\[
  \forall x\,\Bigl((\cB(x)\vee\cR(x))\wedge\NF(x)
  \Longrightarrow
  \Read^\tau_{\bar\blank}(x)\wedge
  \Write^\tau_{\bar\blank}(x)\wedge
  \Move^\tau_0(x)\Bigr).
\]
On $\tau_{\mathrm{in}}$, the build actions are expressed, for example, by
\begin{align*}
  \forall x\,\bigl(\First(x)&\Longrightarrow
  \Read^{\tau_{\mathrm{in}}}_{\bar\blank}(x)\wedge
  \Write^{\tau_{\mathrm{in}}}_{\bar 1}(x)\wedge
  \Move^{\tau_{\mathrm{in}}}_{+1}(x)\bigr),\\
  \forall x\,\bigl(\cB(x)\wedge\NF(x)\wedge Z_0(x)\wedge\neg\First(x)
  &\Longrightarrow
  \Read^{\tau_{\mathrm{in}}}_{\blank}(x)\wedge
  \Write^{\tau_{\mathrm{in}}}_{1}(x)\wedge
  \Move^{\tau_{\mathrm{in}}}_{+1}(x)\bigr),\\
  \forall x\,\bigl(\cB(x)\wedge\NF(x)\wedge\neg Z_0(x)
  &\Longrightarrow
  \Read^{\tau_{\mathrm{in}}}_{\blank}(x)\wedge
  \Write^{\tau_{\mathrm{in}}}_{\blank}(x)\wedge
  \Move^{\tau_{\mathrm{in}}}_{0}(x)\bigr).
\end{align*}
During rewind, the scanned symbol is copied unchanged.  Thus, for every symbol $\gamma$ in the fixed alphabet of $\tau_{\mathrm{in}}$, impose
\begin{align*}
  \forall x\,\bigl(\cR(x)\wedge\NF(x)\wedge Z_0(x)\wedge
  \Read^{\tau_{\mathrm{in}}}_\gamma(x)
  &\Longrightarrow
  \Write^{\tau_{\mathrm{in}}}_\gamma(x)\wedge
  \Move^{\tau_{\mathrm{in}}}_{-1}(x)\bigr),\\
  \forall x\,\bigl(\cR(x)\wedge\NF(x)\wedge\neg Z_0(x)\wedge
  \Read^{\tau_{\mathrm{in}}}_\gamma(x)
  &\Longrightarrow
  \Write^{\tau_{\mathrm{in}}}_\gamma(x)\wedge
  \Move^{\tau_{\mathrm{in}}}_{0}(x)\bigr).
\end{align*}
Together with the action-exclusivity clauses and the read-consistency clauses in Section~\ref{sec:twofun-history}, this finite schema fixes the preparation actions uniquely.}

On the input tape of $V$, the build interval is fixed as follows.  At the distinguished first point $c$, the sentence requires the head to read $\bar\blank$, write $\bar 1$, and move one cell to the right.  At each later build point carrying $Z_0$, it requires the head to read a blank cell, write $1$, and move right.  At all remaining build points, the scanned blank is rewritten unchanged and the head stays in place.  As shown above, the build interval contains exactly $n$ nonfinal $Z_0$-points, occurring at times $0,a,\ldots,(n-1)a$.  Consequently, at $q_1=S^{an}(c)$ the tape contains
\[
  \bar 1\,1^{n-1}\blank\blank\cdots,
\]
and the head scans the blank cell immediately following the last $1$.

During the rewind interval, the sentence requires the input-tape head to move one cell to the left precisely at the $Z_0$-points and to stay in place otherwise; the scanned symbol is always rewritten unchanged.  There are again exactly $n$ such left moves, so at $q_2=S^{2an}(c)$ the head has returned to the marked first cell.  Thus the configuration prepared at $q_2$ is exactly the initial configuration expected by $V$ on input $1^n$: the input tape of $V$ contains $1^n$ with its head on the marked first cell, while all work tapes are blank with their heads at their marked left endpoints.

The last $2an$ edges, from $q_2$ to $b$, form the machine phase.  For each named transition $\delta\in\Delta_V$ introduce a unary predicate $D_\delta$; the intended meaning of $D_\delta(x)$ is that $V$ takes transition $\delta$ at time $x$.  Impose
\begin{equation}\label{eq:D-cover}
  \forall x\,\bigl(\cM(x)\wedge\NF(x)\Longrightarrow
  \bigvee_{\delta\in\Delta_V}D_\delta(x)\bigr),
\end{equation}
\begin{equation}\label{eq:D-disjoint}
  \forall x\,\neg\bigl(D_\delta(x)\wedge D_{\delta'}(x)\bigr)
  \qquad(\delta\ne\delta'),
\end{equation}
and
\begin{equation}\label{eq:D-guard}
  \forall x\,\bigl(D_\delta(x)\Longrightarrow\cM(x)\wedge\NF(x)\bigr)
  \qquad(\delta\in\Delta_V).
\end{equation}

Each $\delta\in\Delta_V$ is a fixed machine rule.  Write $\operatorname{src}(\delta)$ and $\operatorname{tgt}(\delta)$ for its source and target states.  For every tape $\tau$, the same rule also specifies a read symbol $\operatorname{read}_\tau(\delta)$, a written symbol $\operatorname{write}_\tau(\delta)$, and a displacement $\operatorname{move}_\tau(\delta)\in\{-1,0,+1\}$.  We tie the action predicates to the selected rule by the clauses
\begin{equation}\label{eq:D-actions}
  \forall x\,\Bigl(D_\delta(x)\Longrightarrow
  \Read^\tau_{\operatorname{read}_\tau(\delta)}(x)\wedge
  \Write^\tau_{\operatorname{write}_\tau(\delta)}(x)\wedge
  \Move^\tau_{\operatorname{move}_\tau(\delta)}(x)\Bigr)
  \qquad(\delta\in\Delta_V,\ \tau\text{ a tape}).
\end{equation}
Since exactly one read, write, and move predicate is allowed for each tape at every nonfinal machine-phase point, the selected $D_\delta$ uniquely determines all of these actions.

The control state is enforced directly from the source and target states of the named rules.  At the first machine point the selected transition must leave $q_0$:
\begin{equation}\label{eq:D-initial-state}
  \forall x\,\Bigl(x=q_2(x)\Longrightarrow
  \bigvee_{\substack{\delta\in\Delta_V\\\operatorname{src}(\delta)=q_0}}D_\delta(x)\Bigr).
\end{equation}
For every $\delta\in\Delta_V$, if its successor is not the final point $b$, the rule selected at the next time must have the matching source state:
\begin{equation}\label{eq:D-state-step}
  \forall x\,\Bigl(D_\delta(x)\wedge S(x)\ne b(x)\Longrightarrow
  \bigvee_{\substack{\delta'\in\Delta_V\\
  \operatorname{src}(\delta')=\operatorname{tgt}(\delta)}}D_{\delta'}(S(x))\Bigr).
\end{equation}
Finally, the transition immediately preceding $b$ must enter the accepting state:
\begin{equation}\label{eq:D-final-state}
  \forall x\,\Bigl(\cM(x)\wedge S(x)=b(x)\Longrightarrow
  \bigvee_{\substack{\delta\in\Delta_V\\\operatorname{tgt}(\delta)=q_{\mathrm{acc}}}}D_\delta(x)\Bigr).
\end{equation}
All disjunctions above range over the fixed finite set $\Delta_V$, so these are fixed one-variable clauses.

At this stage, the $D_\delta$-clauses record the action prescribed by the selected transition; in particular, they do not yet show that $\operatorname{read}_\tau(\delta)$ is the symbol actually present under the tape head.  This is verified in Section~\ref{sec:twofun-history}.  The previous-visit functions introduced there identify the earlier time at which the cell being entered was last visited, and hence last written; the read-consistency clauses then compare that symbol with the read symbol prescribed by $D_\delta$.

The preceding paragraphs describe the intended tape contents.  We next show how one variable verifies them without a tape-position coordinate.

\subsubsection{Nearest-neighbour history in one variable}\label{sec:twofun-history}

The semantic idea behind the following construction appears in \citet{durandmore2002}: their arithmetic encoding defines the most recent earlier time at which a Turing-machine head occupied its current tape cell and uses the symbol written on that visit to determine the symbol currently read. Our setting does not provide arithmetic tape positions or a way to compare two independently quantified time points. We therefore store the needed history as unary functions on the time points themselves.

The purpose of the following construction is to verify tape contents without representing tape positions as elements of the structure.  Instead, we record enough information about the past motion of each tape head.  For a fixed tape, at each time we remember the most recent earlier visit to the cell immediately to the left of the current head position and, when there was one, the most recent earlier visit to the cell immediately to the right.  We call the resulting unary maps the \emph{previous-visit functions}, and refer to the mechanism based on them as the \emph{history construction}.

These two previous visits are enough because a tape head moves only between neighbouring cells.  When the head moves into one of its neighbouring cells, the relevant previous-visit function tells us when that destination cell was last visited, and therefore when its current symbol was last written.  The previous-visit information stored at that earlier time also determines the previous-visit information for the new head position.  Thus both tape motion and tape contents can be checked using unary functions on time points alone.

Fix a tape $\tau$ and follow its head through the \emph{entire} linear root interval, including the build, rewind, and machine phases.  Write $p_t$ for the tape position scanned at time $t$.  Every head starts at the leftmost cell, which we number $0$, and no legal action moves to a negative position.  Thus
\[
  p_0=0,
  \qquad
  p_t\ge 0,
  \qquad
  p_{t+1}-p_t\in\{-1,0,+1\}.
\]
The previous-visit information is not reset at $q_1$ or $q_2$: a lookup during the machine phase may point back into the deterministic preparation.  This is essential on the input tape of $V$, whose input symbols were written during the build phase.  At time $0$ all tape cells are blank, with the blank in the leftmost cell marked; hence a cell that has never been visited still contains blank.

At time $t$ define the partial previous-visit pointers
\begin{align}
  \lambda(t)&=\max\{u<t:p_u=p_t-1\},\label{eq:lambda}\\
  \rho(t)&=\max\{u<t:p_u=p_t+1\},\label{eq:rho}
\end{align}
with value $\bot$ if the corresponding set is empty.  Thus $\lambda$ is the \emph{left} pointer and $\rho$ the \emph{right} pointer; their arguments are current times and their defined values are earlier times.

There is an important asymmetry between them.  Since the walk starts at $0$ and uses only nearest-neighbour moves, every position $p_t>0$ can be reached only after visiting $p_t-1$.  Consequently,
\begin{equation}\label{eq:left-undefined-endpoint}
  \lambda(t)=\bot
  \quad\Longleftrightarrow\quad
  p_t=0.
\end{equation}
The right pointer has no analogous property: $\rho(t)$ may be undefined at any tape position because the cell to the right may never have been visited.

The recursive update has a simple nearest-neighbour interpretation.  Suppose the head moves right, from $p_t$ to $p_t+1$.  At time $t+1$ the new left neighbour is the cell $p_t$ just left behind, so $\lambda(t+1)=t$.  For the new right neighbour $p_t+2$, let $u=\rho(t)$ when this is defined.  At time $u$ the head was at $p_t+1$, and hence the previous visit to $p_t+2$, if any, is recorded by $\rho(u)=\rho(\rho(t))$.  For a left move, $\lambda(t)$ is necessarily defined by \eqref{eq:left-undefined-endpoint}; the cell entered is the old left neighbour, and its own left history is therefore obtained by one further $\lambda$-lookup.  A stay move changes neither neighbouring cell.

\begin{lemma}[Nearest-neighbour recurrence]\label{lem:history-recurrence}
The pointers in \eqref{eq:lambda}--\eqref{eq:rho} satisfy the following recurrences.
\begin{enumerate}[label=(\alph*),leftmargin=2.2em]
\item If $p_{t+1}=p_t+1$, then
\[
  \lambda(t+1)=t,
  \qquad
  \rho(t+1)=
  \begin{cases}
    \bot,&\rho(t)=\bot,\\
    \rho(\rho(t)),&\rho(t)\ne\bot.
  \end{cases}
\]
\item If $p_{t+1}=p_t-1$, then $\lambda(t)$ is defined and
\[
  \rho(t+1)=t,
  \qquad
  \lambda(t+1)=\lambda(\lambda(t)).
\]
Here the nested value may be $\bot$, precisely when the new position $p_t-1$ is the left endpoint.
\item If $p_{t+1}=p_t$, both pointers are unchanged.
\end{enumerate}
\end{lemma}

\begin{proof}
For the right move, the identity $\lambda(t+1)=t$ is immediate.  Let $v=\rho(t)$.  If $v$ is undefined, then $p_t+2$ cannot have been visited earlier: returning from $p_t+2$ to the current position $p_t$ would force another visit to $p_t+1$.  If $v$ is defined, the latest earlier visit to $p_t+2$ is exactly the latest visit to the right neighbour of the head at time $v$, namely $\rho(v)$.

For a left move, legality gives $p_t>0$, so \eqref{eq:left-undefined-endpoint} makes $u=\lambda(t)$ defined.  At time $u$ the head was at the new position $p_t-1$.  Its left neighbour is $p_t-2$, so the latest earlier visit to that cell is $\lambda(u)$, which may be undefined only when $p_t-1=0$.  Finally, a stay move changes neither adjacent cell.
\end{proof}

The same pointers reconstruct tape contents.  After a right move, the destination cell is fresh when $\rho(t)=\bot$ and therefore contains blank; otherwise its contents are the symbol written at time $\rho(t)$.  After a left move the destination cell has necessarily been visited before, and its contents are the symbol written at time $\lambda(t)$.  After a stay, the next symbol read is the symbol just written at time $t$.

\paragraph{Unary encoding of the pointers.}

For each tape $\tau$, we now represent the two semantic previous-visit pointers by unary function symbols
\[
  \lambda_\tau,\rho_\tau.
\]
The semantic pointers defined above are partial, whereas a first-order function symbol must have a value at every element.  We therefore use the current time itself as the dummy value when a previous-visit pointer is undefined.  This convention is unambiguous on the active root: all updates take place along the linear interval from $c$ to $b$, so every genuine previous visit is a strictly earlier, and hence different, root element.  Consequently
\[
  \lambda_\tau(x)=x
\]
means exactly that the head is at the left endpoint, while
\[
  \rho_\tau(x)=x
\]
means that the cell immediately to the right of the head has not yet been visited.

The latter case occurs repeatedly in the read clauses, so for readability we introduce a unary predicate $\Rvis_\tau$ and define it by the fixed-point test
\begin{equation}\label{eq:R-definition}
  \forall x\,\bigl(\Rvis_\tau(x)\leftrightarrow \rho_\tau(x)\ne x\bigr).
\end{equation}
Thus $\Rvis_\tau(x)$ simply says that the right neighbour has already been visited.  It carries no information beyond the value of $\rho_\tau(x)$; in particular, the off-root identity convention below automatically makes $\Rvis_\tau$ false off the active root.

At the first root point both pointers are undefined, so
\begin{equation}\label{eq:history-first}
  \forall x\,\bigl(\First(x)\Longrightarrow
  (\lambda_\tau(x)=x\wedge\rho_\tau(x)=x)\bigr).
\end{equation}
Off the active root the otherwise unused function values are fixed by the same identity convention:
\begin{equation}\label{eq:history-offroot}
  \forall x\,\bigl(\neg A(x)\Longrightarrow
  (\lambda_\tau(x)=x\wedge\rho_\tau(x)=x)\bigr).
\end{equation}
The previous-visit functions preserve the $\pi$-component, and values reached from the active root stay on the root cycle:
\begin{align}
  \forall x\,\bigl(&\pi(\lambda_\tau(x))=\pi(x)\wedge
  \pi(\rho_\tau(x))=\pi(x)\bigr),\label{eq:history-component-local}\\
  \forall x\,\bigl(A(x)&\Longrightarrow
  A(\lambda_\tau(x))\wedge A(\rho_\tau(x))\bigr).\label{eq:history-root-local}
\end{align}
Thus a previous-visit function may not jump to another $\pi$-component; when its argument is on the active root, it cannot jump to a satellite either.

At every nonfinal active point, the selected displacement determines the next previous-visit values.  For a right move impose
\begin{equation}\label{eq:rightL}
  \forall x\,\bigl(\NF(x)\wedge\Move^\tau_{+1}(x)
  \Longrightarrow \lambda_\tau(S(x))=x\bigr),
\end{equation}
and, for the new right neighbour,
\begin{align}
  \forall x\,\bigl(\NF(x)\wedge\Move^\tau_{+1}(x)\wedge
  \Rvis_\tau(\rho_\tau(x))
  &\Longrightarrow
  \rho_\tau(S(x))=\rho_\tau(\rho_\tau(x))\bigr),\label{eq:rightRval}\\
  \forall x\,\bigl(\NF(x)\wedge\Move^\tau_{+1}(x)\wedge
  \neg\Rvis_\tau(\rho_\tau(x))
  &\Longrightarrow
  \rho_\tau(S(x))=S(x)\bigr).\label{eq:rightRfresh}
\end{align}
The first condition uses only $\Rvis_\tau(\rho_\tau(x))$.\footnote{The fixed-point convention is important here.  If $\Rvis_\tau(x)$ is false, then $\rho_\tau(x)=x$, and therefore $\Rvis_\tau(\rho_\tau(x))=\Rvis_\tau(x)$ is false as well.  Thus $\Rvis_\tau(\rho_\tau(x))$ already guarantees that the old right pointer was defined.  This depends on the chosen dummy value: with a different totalization convention, the recurrence would need an additional definedness test.}

For a left move, the fixed-point characterization of the left endpoint first gives the legality condition
\begin{equation}\label{eq:left-legal}
  \forall x\,\bigl(\NF(x)\wedge\Move^\tau_{-1}(x)
  \Longrightarrow \lambda_\tau(x)\ne x\bigr).
\end{equation}
The new right neighbour is the cell just left behind:
\begin{equation}\label{eq:leftR}
  \forall x\,\bigl(\NF(x)\wedge\Move^\tau_{-1}(x)
  \Longrightarrow \rho_\tau(S(x))=x\bigr).
\end{equation}
For the new left pointer, the nested value $\lambda_\tau(\lambda_\tau(x))$ is a genuine earlier visit unless $\lambda_\tau(x)$ was itself a time at the left endpoint.  Accordingly impose
\begin{align}
  \forall x\,\bigl(\NF(x)\wedge\Move^\tau_{-1}(x)\wedge
  \lambda_\tau(\lambda_\tau(x))=\lambda_\tau(x)
  &\Longrightarrow \lambda_\tau(S(x))=S(x)\bigr),\label{eq:leftLendpoint}\\
  \forall x\,\bigl(\NF(x)\wedge\Move^\tau_{-1}(x)\wedge
  \lambda_\tau(\lambda_\tau(x))\ne\lambda_\tau(x)
  &\Longrightarrow
  \lambda_\tau(S(x))=\lambda_\tau(\lambda_\tau(x))\bigr).\label{eq:leftLval}
\end{align}

Finally, a stay move carries a genuine pointer forward, while an undefined pointer must use the \emph{new} current time as its fixed-point dummy value:
\begin{align}
  \forall x\,\bigl(\NF(x)\wedge\Move^\tau_0(x)\wedge
  \lambda_\tau(x)=x
  &\Longrightarrow \lambda_\tau(S(x))=S(x)\bigr),\label{eq:stay-L-endpoint}\\
  \forall x\,\bigl(\NF(x)\wedge\Move^\tau_0(x)\wedge
  \lambda_\tau(x)\ne x
  &\Longrightarrow \lambda_\tau(S(x))=\lambda_\tau(x)\bigr),\label{eq:stay-L}\\
  \forall x\,\bigl(\NF(x)\wedge\Move^\tau_0(x)\wedge
  \Rvis_\tau(x)
  &\Longrightarrow \rho_\tau(S(x))=\rho_\tau(x)\bigr),\label{eq:stay-R}\\
  \forall x\,\bigl(\NF(x)\wedge\Move^\tau_0(x)\wedge
  \neg\Rvis_\tau(x)
  &\Longrightarrow \rho_\tau(S(x))=S(x)\bigr).\label{eq:stay-R-fresh}
\end{align}
All nested lookups are unary terms in the single variable $x$.  Starting from \eqref{eq:history-first}, induction along the linear root interval \eqref{eq:linear-time} shows that the two previous-visit-function values are uniquely determined by the sequence of head movements; $\Rvis_\tau$ is then uniquely determined by \eqref{eq:R-definition}.

\paragraph{Verifying reads.}

The read predicates at $c$ are fixed by the deterministic preparation: the input tape of $V$ reads $\bar\blank$ and every work tape reads its marked blank.  For every later time point, including the build, rewind, and machine phases, the same previous-visit functions verify the scanned symbol.  In particular, the history is not restarted at $q_2$; a machine-phase read may be justified by a write made during preparation.

For every tape $\tau$ and tape symbol $\gamma$, a right move is checked by
\begin{align}
  \forall x\,\bigl(\NF(x)\wedge\neg\Final(S(x))\wedge
  \Move^\tau_{+1}(x)\wedge\neg\Rvis_\tau(x)
  &\Longrightarrow \Read^\tau_{\blank}(S(x))\bigr),\label{eq:read-right-blank}\\
  \forall x\,\bigl(\NF(x)\wedge\neg\Final(S(x))\wedge
  \Move^\tau_{+1}(x)\wedge\Rvis_\tau(x)
  &\Longrightarrow
  (\Read^\tau_\gamma(S(x))\leftrightarrow
   \Write^\tau_\gamma(\rho_\tau(x)))\bigr).\label{eq:read-right}
\end{align}
After a left move the destination has necessarily been visited, so no blank/unvisited case is needed:
\begin{equation}\label{eq:read-left}
  \forall x\,\bigl(\NF(x)\wedge\neg\Final(S(x))\wedge
  \Move^\tau_{-1}(x)
  \Longrightarrow
  (\Read^\tau_\gamma(S(x))\leftrightarrow
   \Write^\tau_\gamma(\lambda_\tau(x)))\bigr).
\end{equation}
After a stay impose
\begin{equation}\label{eq:read-stay}
  \forall x\,\bigl(\NF(x)\wedge\neg\Final(S(x))\wedge\Move^\tau_0(x)
  \Longrightarrow
  (\Read^\tau_\gamma(S(x))\leftrightarrow\Write^\tau_\gamma(x))\bigr).
\end{equation}
No read is required at the final point $b$, because no transition is selected there.

Equation~\eqref{eq:left-legal} prevents a satisfying assignment from moving left of the first tape cell.  This agrees with the marked-left-end normalization of $V$: while scanning a marked symbol, no machine transition prescribes a left move.  Thus every satisfying assignment describes a legal semi-infinite tape history.

Let $\Theta$ denote the conjunction of the structural, preparation, transition, previous-visit, and read clauses introduced above, including the explicit action clauses \eqref{eq:action-read-one}--\eqref{eq:action-offroot} and the finite transition-consistency clauses described in the text.  Every clause is constant-free and has the form $\forall x\,\theta(x)$ with quantifier-free $\theta$, so $\Theta$ is a fixed constant-free universal $\mathrm{FO}^1_{=}$ sentence over a finite unary vocabulary.

\begin{lemma}[Branch--root correspondence]\label{lem:branch-root}
Let a labeled structural model have size $m=4an+1$, with its whole domain equal to the active root cycle, and let $c$ be its distinguished root point.  Its computational expansions satisfying the preparation, transition, previous-visit, and read clauses are in bijection with the accepting branches of $U$ on input $1^n$.
\end{lemma}

\begin{proof}
Fix an accepting branch of $U$.  By \eqref{eq:Vcount} and the deterministic slowing construction, it gives {a unique accepting branch} of the slowed machine $V$.  The phase and modulo-$a$ predicates are already uniquely fixed by the structural root.  Hence the build and rewind actions are deterministic.  At $q_2$ they leave the input tape of $V$ containing exactly $1^n$, with its head back at the marked first cell, and leave all work tapes blank at their marked left endpoints.

The accepting branch of $V$ fixes one named transition predicate $D_\delta$ at each of the final $2an$ nonfinal time points, and therefore fixes every read, write, and move action there.  For each tape, Lemma~\ref{lem:history-recurrence} and the totalization clauses above determine $\lambda_\tau$ and $\rho_\tau$ uniquely from their initial fixed points at $c$ through the build, rewind, and machine phases.  Equation~\eqref{eq:R-definition} then fixes $\Rvis_\tau$, while \eqref{eq:history-offroot} fixes every otherwise unused previous-visit-function value.  The read clauses hold because they reproduce exactly the standard most-recent-write semantics of a Turing tape.  Thus the branch has exactly one computational expansion.

Conversely, a satisfying expansion selects one named transition of $V$ at every machine-phase step.  The control clauses make these transitions state-consistent.  By induction along the complete root interval and Lemma~\ref{lem:history-recurrence}, the stored functions are the genuine previous visits to the neighbouring cells: $\lambda_\tau(x)=x$ exactly at the left endpoint, while $\rho_\tau(x)=x$ exactly when the right neighbour has not previously been visited.  Equation~\eqref{eq:R-definition} records the latter condition as $\Rvis_\tau$.  The read clauses \eqref{eq:read-right-blank}--\eqref{eq:read-stay} therefore certify the actual scanned symbols, including reads at the beginning of the machine phase whose last write occurred during preparation.  The deterministic first two phases establish the initial configuration of $V$ on $1^n$, and the final transition has accepting target state.  Hence the selected transition sequence is an accepting branch.  The two constructions are inverse.
\end{proof}

\subsubsection{Exact root count}\label{sec:twofun-rootcount}

Let $R_r$ denote the number of labeled models of $\Theta$ on an $r$-element domain that consist of a single $\pi$-component with no satellites; equivalently, the whole domain is the active root cycle.

\begin{proposition}[Root multiplicity]\label{prop:rootcount}
For every $n\ge1$, with $m=4an+1$,
\begin{equation}\label{eq:rootcount}
  R_m=m!\,\acc_U(n).
\end{equation}
\end{proposition}

\begin{proof}
On a fixed labeled $m$-element set there are $(m-1)!$ directed $m$-cycles for $S$.  Once $S$ is fixed, the component root $c$ can be chosen in $m$ ways; then $\pi$ is uniquely the constant map to $c$ and $P=S^{-1}$ is forced.  Since $m$ is odd, {Lemma~\ref{lem:parity}} fixes $E$ uniquely.  Lemma~\ref{lem:root-isolation} fixes $A$ as the whole domain, and \eqref{eq:dexp} fixes $d$ and $h$ uniquely.  At $m=4an+1$ the phase predicates and the modulo-$a$ counter are uniquely determined.  Finally {Lemma~\ref{lem:branch-root}} contributes exactly $\acc_U(n)$ computational expansions.  Multiplying the only structural choices gives
\[
  (m-1)!\cdot m\cdot \acc_U(n)=m!\,\acc_U(n).
\]
\end{proof}

The next section shows how to recover $R_m$ even though the constant-free sentence $\Theta$ necessarily allows arbitrary SETs of components and arbitrary even satellites inside each component.

\subsubsection{Two exact inversions}\label{sec:twofun-inversions}

For an arbitrary domain size $\ell$, let
\[
  M_\ell:=\FOMC(\Theta,\ell)
\]
be the number of labeled models of the fixed sentence $\Theta$ on domain $[\ell]$, and let $K_\ell$ be the number of such models having exactly one $\pi$-component.  Put $M_0=1$.

\paragraph{SET inversion across components.}

{By construction, every unary function used in $\Theta$ preserves the $\pi$-components.  For $S$ this is \eqref{eq:pi}; for $P$ it follows from
\[
  \pi(P(x))=\pi(S(P(x)))=\pi(x);
\]
$\pi,d,h$ preserve components by the structural equations, and the previous-visit functions do so by \eqref{eq:history-component-local}.  Hence each $\pi$-component is a substructure.  Moreover, the SET decomposition here uses the specific syntactic form of $\Theta$: it is a conjunction of universal one-variable clauses with quantifier-free bodies.  Restricting a model to one $\pi$-component therefore preserves every clause, while models with one $\pi$-component on disjoint labeled sets combine uniquely by disjoint union. Thus, for this particular sentence $\Theta$, a model is exactly an unordered collection of one-component models on disjoint label sets. We need no general species formalism: the recurrence below follows directly by singling out the unique component containing label $1$.}

To derive the recurrence, consider a labeled model of $\Theta$ on domain $[\ell]$, and focus on the unique $\pi$-component containing the distinguished label $1$.  Suppose this component has size $k$.  We may choose its remaining $k-1$ labels in $\binom{\ell-1}{k-1}$ ways; the induced structure on those $k$ labels can be chosen in $K_k$ ways, and the remaining $\ell-k$ labels carry an arbitrary model of $\Theta$, counted by $M_{\ell-k}$.  Summing over $k$ gives
\begin{equation}\label{eq:setrec}
  M_\ell=
  \sum_{k=1}^{\ell}
  \binom{\ell-1}{k-1}K_kM_{\ell-k}.
\end{equation}
Hence
\begin{equation}\label{eq:setinv}
  K_\ell=
  M_\ell-
  \sum_{k=1}^{\ell-1}
  \binom{\ell-1}{k-1}K_kM_{\ell-k}.
\end{equation}
Thus $K_1,\ldots,K_\ell$ are recovered sequentially from $M_1,\ldots,M_\ell$ in polynomial time.

\paragraph{Removing the satellites from a one-component count.}

We now relate $K_\ell$, which counts models having one $\pi$-component, to $R_r$, which counts models whose whole domain is the active root cycle.  The only extra structure present in a model counted by $K_\ell$ consists of the satellite cycles attached to the root.

Once the root and its distinguished point $c$ are fixed, let $X$ be a prescribed set of $2s$ labels used by satellites.  On $X$, the satellite structure is determined by the restriction of $S$ and by $E$: every $S$-cycle has even length by {Lemma~\ref{lem:parity}}, and $E$ alternates around each such cycle.  Conversely, any permutation of $X$ whose cycles all have even length, together with one of the two alternating $E$-assignments on each cycle, determines the satellite structure uniquely.  Indeed, $P=S^{-1}$, $\pi$ maps every point of $X$ to $c$, and the remaining predicates and function values are fixed by the inactive clauses.

Let $T_{2s}$ be the number of satellite structures on $X$.  We use the standard enumeration of permutations by cycle type: if a permutation of an $n$-element set has $c_i$ cycles of length $i$, then the number of such permutations is
\begin{equation}\label{eq:cycle-type}
  \frac{n!}{\prod_i i^{c_i}c_i!};
\end{equation}
see \citet[\S1.3, Proposition~1.3.2]{stanley2012}.  Suppose that $S$ has $m_j$ cycles of length $2j$, so that $\sum_{j\ge1}j m_j=s$.  For this fixed cycle type, \eqref{eq:cycle-type} counts the choices of $S$, while each cycle has two alternating $E$-assignments.  Hence the number of satellite structures of this type is
\begin{equation}\label{eq:satellite-type}
  \frac{(2s)!\,2^{\sum_jm_j}}{\prod_j(2j)^{m_j}m_j!}
  =
  (2s)!\prod_j\frac{1}{j^{m_j}m_j!}.
\end{equation}
By the same cycle-type formula,
\[
  \frac{s!}{\prod_j j^{m_j}m_j!}
\]
is the number of permutations of an $s$-element set having $m_j$ cycles of length $j$.  Thus \eqref{eq:satellite-type} is exactly $(2s)!/s!$ times the number of such permutations.  Summing over all cycle types, which together account for all $s!$ permutations of an $s$-element set, gives
\begin{equation}\label{eq:satellite-assembly}
  T_{2s}=\frac{(2s)!}{s!}\,s!=(2s)!.
\end{equation}
There are no satellite structures on an odd number of labels.

We can now express $K_\ell$ in terms of the counts $R_r$.  Suppose $\ell$ is odd, and consider a model counted by $K_\ell$.  Its active root cycle has some odd size $r\le \ell$.  Choose the $r$ labels belonging to the root, choose one of the $R_r$ structures on those labels, and use the remaining $\ell-r$ labels for satellites.  By \eqref{eq:satellite-assembly}, the last step has exactly $(\ell-r)!$ possibilities.  Consequently
\begin{align}
  K_\ell
  &=
  \sum_{\substack{r\le \ell\\ r\text{ odd}}}
  \binom{\ell}{r}R_r(\ell-r)! \notag\\
  &=
  \ell!\sum_{\substack{r\le \ell\\ r\text{ odd}}}\frac{R_r}{r!}.
  \label{eq:KfromR}
\end{align}
Dividing by $\ell!$ gives
\[
  \frac{K_\ell}{\ell!}
  =
  \frac{R_1}{1!}+\frac{R_3}{3!}+\cdots+\frac{R_\ell}{\ell!}.
\]
The corresponding identity for $\ell-2$ has exactly the same terms except the last one.  Hence
\begin{equation}\label{eq:rootinv}
  \boxed{
  R_\ell=K_\ell-\ell(\ell-1)K_{\ell-2}
  }
\end{equation}
for odd $\ell\ge3$, with $R_1=K_1$.

Together with \eqref{eq:setinv}, equation~\eqref{eq:rootinv} shows that $R_\ell$ can be recovered from $M_1,\ldots,M_\ell$.

\subsubsection{Exact compression to two unary functions}\label{sec:twofun-compression}

The source sentence $\Theta$ is constant-free and universal, but it still uses a fixed finite family of unary functions.  We now compress them to two while preserving model counts up to an explicit multiplier.

\begin{lemma}[Exact two-function compression]\label{lem:compression}
Let $\Psi$ be a constant-free $\mathrm{FO}^1_{=}$ sentence whose vocabulary consists of unary functions
\[
  h_0,\ldots,h_{r-1}
\]
and unary predicates $P_1,\ldots,P_s$, where $r\ge2$ is fixed.  One can effectively construct a constant-free sentence $\Psi^\flat\in\mathrm{FO}^1_{=}[f,g]$, using only $f,g$, the source predicates, and fresh unary predicates $L_0,\ldots,L_{r-1}$, such that
\begin{equation}\label{eq:compression-count}
  \FOMC(\Psi^\flat,rq)
  =\frac{(rq)!}{q!}\FOMC(\Psi,q)
\end{equation}
for every $q\ge1$, and $\FOMC(\Psi^\flat,s)=0$ when $r\nmid s$.  If $\Psi$ is universal, so is $\Psi^\flat$.
\end{lemma}

\begin{proof}
Require the predicates $L_0,\ldots,L_{r-1}$ to partition the target domain, using the usual finite cover and pairwise-exclusion clauses, and impose
\[
  \forall x\, f^r(x)=x,
  \qquad
  \forall x\,\bigl(L_i(x)\Longrightarrow L_{i+1\bmod r}(f(x))\bigr)
  \quad(0\le i<r).
\]
Thus $f$ is a permutation mapping each layer bijectively to the next, so all layers have equal size.  Require $g$ to preserve every layer and source predicates to live only on $L_0$:
\[
  \forall x\,\bigl(L_i(x)\Longrightarrow L_i(g(x))\bigr)
  \quad(0\le i<r),
  \qquad
  \forall x\,\bigl(\neg L_0(x)\Longrightarrow\neg P_j(x)\bigr)
  \quad(1\le j\le s).
\]

For $0\le i<r$, write $\bar i=(-i)\bmod r$.  For readability, we use
\[
  H_i(t)
\]
as an abbreviation for the $f,g$-term
\[
  f^{\bar i}\bigl(g(f^i(t))\bigr).
\]
Thus $H_i$ is not an additional function symbol.

We translate the original formula recursively.  For a term $t$, let $\widehat t$ denote its translation, and set
\[
  \widehat x=x,
  \qquad
  \widehat{h_i(t)}=H_i(\widehat t).
\]
In atomic formulas, replace every term by its translation, and leave Boolean connectives unchanged.  {Finally, restrict every quantifier to $L_0$ by replacing}
\[
  \exists x\,\varphi
  \quad\text{by}\quad
  \exists x\,\bigl(L_0(x)\wedge\widehat\varphi\bigr),
\]
and replace
\[
  \forall x\,\varphi
  \quad\text{by}\quad
  \forall x\,\bigl(L_0(x)\rightarrow\widehat\varphi\bigr).
\]
No additional first-order variables are introduced, so the translated sentence remains in $\mathrm{FO}^1_{=}$.

{The purpose of the $r$ layers is to let the single function $g$ encode all $r$ source functions.  Once $f$ is fixed, the map $z\mapsto f^i(z)$ is a bijection from $L_0$ to $L_i$.  We use this bijection to transport the source function $h_i$ from $L_0$ to the $i$th layer: for $z\in L_0$, require
\[
  g(f^i(z))=f^i(h_i(z)).
\]
Thus the restriction of $g$ to $L_i$ is just $h_i$ written in the coordinates of $L_i$.  Since every element of $L_i$ has a unique form $f^i(z)$ with $z\in L_0$, this equation determines $g$ uniquely on $L_i$, and hence on the whole target domain.  Conversely, if $g$ preserves the layers, its restriction to $L_i$ can be transported back to $L_0$, giving
\[
  h_i(z)=f^{\bar i}(g(f^i(z)))=H_i(z).
\]
Consequently, once the predicates $L_i$ and the function $f$ are fixed, choosing the functions $h_0,\ldots,h_{r-1}$ on $L_0$ is equivalent to choosing a layer-preserving function $g$.

The number of ordered partitions of $[rq]$ into $r$ labeled layers of size $q$ is
\[
  \frac{(rq)!}{(q!)^r}.
\]
For a fixed partition, choose arbitrary bijections $f:L_i\to L_{i+1}$ for $0\le i<r-1$.  There are $(q!)^{r-1}$ choices, and the last bijection is uniquely determined by $f^r=\mathrm{id}$.  After the layer partition and this $f$-skeleton have been fixed, there are $\FOMC(\Psi,q)$ possible interpretations on $L_0$ of the source functions and predicates that satisfy $\Psi$.  For each such interpretation, the preceding correspondence determines $g$ uniquely, while the source predicates are false outside $L_0$.  Hence
\[
  \FOMC(\Psi^\flat,rq)
  =\frac{(rq)!}{(q!)^r}(q!)^{r-1}\FOMC(\Psi,q)
  =\frac{(rq)!}{q!}\FOMC(\Psi,q),
\]
which proves \eqref{eq:compression-count}.  Equal layer sizes are impossible when $r\nmid s$.}
\end{proof}

The construction is a count-preserving counterpart of the constant-factor two-function spectrum encodings studied by Durand, Fagin, and Loescher \citep{durand1998}; the exact multiplicity above is what is needed for model counting.

\subsubsection{Completing the reduction}\label{sec:twofun-finish}

Let $r$ be the fixed number of unary function symbols in the source sentence $\Theta$ and let
\[
  \Xi:=\Theta^\flat
\]
be supplied by {Lemma~\ref{lem:compression}}.  On input $1^n$, set
\[
  m=4an+1.
\]
For every $1\le k\le m$, query the single fixed oracle function at domain size $rk$ and recover
\begin{equation}\label{eq:recoverM}
  M_k
  =
  \frac{k!}{(rk)!}\FOMC(\Xi,rk)
\end{equation}
by exact division.  Compute $K_1,\ldots,K_m$ from \eqref{eq:setinv}.  Since $m$ is odd, compute
\[
  R_m=K_m-m(m-1)K_{m-2}.
\]
Finally {Proposition~\ref{prop:rootcount}} gives
\begin{equation}\label{eq:final-reduction}
  \boxed{
  \acc_U(n)=\frac{R_m}{m!}
  }.
\end{equation}

There are $m=O(n)$ oracle queries, and every queried domain size $rk$ is $O(n)$.  All counts involved have $O(n\log n)$ bits because the source and target vocabularies are fixed and unary, so the factorial multiplications, exact divisions, binomial coefficients, and recurrence computations take polynomial time.  Since $n\mapsto\acc_U(n)$ is $\#\mathrm P_1$-hard, \eqref{eq:final-reduction} proves $\#\mathrm P_1$-hardness of $N\mapsto\FOMC(\Xi,N)$.  Together with membership, this proves Theorem~\ref{thm:twofun-hardness}.

\appendix

\section{Bijectivity and the Permutation Axiom}\label{sec:permutation-axiom}

The one-variable fragment is too weak to enforce bijectivity, even with unary predicates serving as certificates. Nevertheless, the model-counting algorithm of Section~\ref{sec:proofonevar} can be restricted to bijective interpretations of $f$. We make both statements precise here.

\begin{proposition}[Bijectivity has no unary certificate]\label{prop:no-permutation-certificate}
There is no finite list of unary predicate symbols $P_1,\ldots,P_s$ together with a fixed sentence $\psi\in\Cke{1}[f,P_1,\ldots,P_s]$ such that, for every nonempty finite domain $D$ and every endofunction $f:D\to D$,
\[
f\text{ is bijective}
\quad\Longleftrightarrow\quad
\text{some interpretations }P_1,\ldots,P_s\subseteq D
\text{ make }(D,f,P_1,\ldots,P_s)\models\psi.
\]
Thus unary sets cannot encode bijectivity even when they are allowed to serve as existential certificates.
\end{proposition}

\begin{proof}[Proof of Proposition~\ref{prop:no-permutation-certificate}]
Suppose, toward a contradiction, that $P_1,\ldots,P_s$ and $\psi$ have the stated property. Let $d$ be the greatest nesting depth of $f$ in a term of $\psi$, taking $d=0$ if $f$ does not occur. For fixed interpretations of the predicates, give each element its unary color
\[
\tau(a)=\bigl(\ind{a\in P_1},\ldots,\ind{a\in P_s}\bigr)\in\{0,1\}^s
\]
and its length-$d$ color block
\[
w_f(a)=
\bigl(\tau(a),\tau(f(a)),\ldots,\tau(f^{d-1}(a))\bigr),
\]
where the block is empty when $d=0$. There are at most
\[
B=2^{sd}
\]
such blocks.

Choose $n>B(d+1)$, put $D=[n]$, and let $f:D\to D$ be a directed $n$-cycle. By the assumed equivalence, some interpretations of $P_1,\ldots,P_s$ make the resulting structure satisfy $\psi$. At least one color block occurs at least $d+2$ times. Choose an occurrence $a$ of this block. Among the other occurrences choose $b$ outside
\[
\{a,f^{-1}(a),\ldots,f^{-d}(a)\};
\]
this is possible because the displayed set has only $d+1$ elements. Put $p=f^{-1}(a)$ and change only the edge leaving $p$:
\[
g(p)=b,
\qquad
g(x)=f(x)\quad(x\neq p).
\]
Now $a$ has no $g$-predecessor, while $b$ has two, so $g$ is not bijective. The functional digraph of $g$ is a directed cycle with a nonempty transient path feeding into it. More precisely, if $b=f^{-k}(a)$ on the original cycle, then $k>d$ by the choice of $b$, and the remaining directed cycle has length $k$.

Retain the same interpretations of the unary predicates. We claim that every labeled element has the same $d$-profile under $f$ and under $g$. A forward segment that does not traverse the modified edge is unchanged. If it does traverse that edge, the suffix after $p$ begins at $a$ under $f$ and at $b$ under $g$. The equality
\[
w_f(a)=w_f(b)
\]
therefore preserves the unary color at every position of the segment. Moreover, the first $d+1$ elements of every relevant forward segment are pairwise distinct in both structures: the original cycle has length $n>d$, while under $g$ a segment follows a simple transient path and then a cycle of length $k>d$. Hence every equality atom $f^i(x)=f^j(x)$, with $0\leq i<j\leq d$, also has the same truth value before and after the rewiring. This proves the claim.

Every element has the same profile in the two structures. Hence every set $S_j$ from Lemma~\ref{lem:capped-evaluation} has the same cardinality in both, so the two structures have the same capped-cardinality state. Therefore
\[
(D,f,P_1,\ldots,P_s)\models\psi
\quad\Longleftrightarrow\quad
(D,g,P_1,\ldots,P_s)\models\psi.
\]
The right-hand structure therefore satisfies $\psi$, although $g$ is not bijective, contradicting the assumed equivalence.
\end{proof}

Following the standard FOMC convention, we may nevertheless impose bijectivity as an external background axiom. Write
\[
\PERM(f)
:=
\forall x\forall y :
\bigl(f(x)=f(y)\longrightarrow x=y\bigr).
\]
On finite domains this says exactly that $f$ is a permutation. The axiom uses two variables, so $\PERM(f)\notin\Cke{1}[f]$; the expression $\FOMC(\PERM(f)\wedge\varphi,n)$ means that the models of the one-variable sentence $\varphi$ are counted subject to this external axiom.

\begin{theorem}[Tractability with the permutation axiom]\label{thm:permutation-axiom}
For every fixed sentence $\varphi\in\Cke{1}[f]$ over an arbitrary finite relational vocabulary and one unary function symbol $f$, the function
\[
n\longmapsto
\FOMC\bigl(\PERM(f)\wedge\varphi,n\bigr)
\]
is computable in time polynomial in $n$.
\end{theorem}

\begin{proof}[Proof of Theorem~\ref{thm:permutation-axiom}]
First suppose that every relation symbol in $\varphi$ is unary, and use the profiles and coefficient algebra from Section~\ref{sec:proofonevar}. A finite endofunction is a permutation exactly when its functional digraph has no transient vertices. In the generating-function construction, we therefore delete the rooted in-tree decorations and replace \eqref{eq:cycle-decoration} by
\[
U_q^{\mathrm{perm}}(x)=x\omega_q,
\qquad q\in\Qc.
\]
For each $\ell$, let $B_\ell^{\mathrm{perm}}(x)$ be obtained from \eqref{eq:cycle-matrix} by replacing $U_q$ with $U_q^{\mathrm{perm}}$, and put
\[
C_\ell^{\mathrm{perm}}(x)
=
\frac1\ell\operatorname{tr}\bigl((B_\ell^{\mathrm{perm}}(x))^\ell\bigr),
\qquad
F_\varphi^{\mathrm{perm}}(x)
=
\exp\!\left(\sum_{\ell\geq1}C_\ell^{\mathrm{perm}}(x)\right).
\]
The correctness argument of Lemma~\ref{lem:egf-formula} now simplifies because every component is a directed cycle; the restriction to $\Qc_\ell$ still enforces the correct visible cycle length. Hence
\begin{equation}\label{eq:permutation-decomposition}
\FOMC\bigl(\PERM(f)\wedge\varphi,n\bigr)
=
n!
\sum_{\substack{\mathbf u\in M\\\operatorname{Acc}_\varphi(\mathbf u)=1}}
[x^nY^{\mathbf u}]F_\varphi^{\mathrm{perm}}(x).
\end{equation}
The polynomial-time coefficient calculation from Lemma~\ref{lem:monadic-core} therefore gives the monadic case.

For arbitrary relation arities, the case $n=0$ is immediate, so assume $n\geq1$ and let $\theta_\varphi$ be the monadic trace compression from Lemma~\ref{lem:exact-arity-reduction}. The proof of that lemma compares multiplicities after the interpretation of $f$ has been fixed, and the correction factor is independent of that interpretation. We may therefore restrict both sides to bijective interpretations of $f$ and obtain
\[
2^{K_\varphi(n)}
\FOMC\bigl(\PERM(f)\wedge\varphi,n\bigr)
=
2^{N_\varphi(n)}
\FOMC\bigl(\PERM(f)\wedge\theta_\varphi,n\bigr).
\]
The right-hand side is computable in polynomial time by the monadic case, and the same exact power-of-two correction as in the proof of Theorem~\ref{thm:onevar} recovers the left-hand count. Nullary relation symbols, if present, are handled by enumerating their finitely many truth assignments.
\end{proof}

Proposition~\ref{prop:no-permutation-certificate} explains why the background axiom in Theorem~\ref{thm:permutation-axiom} cannot simply be absorbed into the one-variable fragment, even with auxiliary unary predicates. This does not prevent $\Cke{1}[f]$ from defining restricted classes of permutations: for every fixed $q\geq1$, the sentence
\[
\forall x : f^q(x)=x
\]
forces $f$ to be a permutation all of whose cycle lengths divide $q$.

\section*{Acknowledgments}

This work was funded by the Czech Science Foundation project 24-11820S (``Automatic Combinatorialist'').

\section*{Acknowledgment of AI Use}

The author used AI assistants (GPT-5.4, GPT-5.5, GPT-5.6, and Claude Fable) extensively as interactive tools while exploring possible proof strategies, drafting and revising proofs and exposition, and checking arguments and calculations. The interaction consisted of repeated exchanges between the author and individual assistants, directed by the author and often conducted paragraph by paragraph; the assistants did not develop the manuscript autonomously. The author set the mathematical agenda, chose the directions and proof approaches to pursue, contributed ideas to the arguments, critically evaluated the assistants' suggestions and drafts, decided what to retain and how to revise it, and independently verified the final statements and proofs. The author takes sole responsibility for the manuscript and any remaining errors.

\bibliographystyle{plainnat}
\bibliography{refs}

\begin{thebibliography}{34}
\providecommand{\natexlab}[1]{#1}
\providecommand{\url}[1]{\texttt{#1}}
\expandafter\ifx\csname urlstyle\endcsname\relax
  \providecommand{\doi}[1]{doi: #1}\else
  \providecommand{\doi}{doi: \begingroup \urlstyle{rm}\Url}\fi

\bibitem[Beame et~al.(2015)Beame, {Van den Broeck}, Gribkoff, and
  Suciu]{beame2015}
Paul Beame, Guy {Van den Broeck}, Eric Gribkoff, and Dan Suciu.
\newblock Symmetric weighted first-order model counting.
\newblock In \emph{Proceedings of the 34th {ACM} Symposium on Principles of
  Database Systems ({PODS})}, pages 313--328, 2015.
\newblock \doi{10.1145/2745754.2745760}.

\bibitem[B{\"o}rger et~al.(1997)B{\"o}rger, Gr{\"a}del, and
  Gurevich]{borger1997}
Egon B{\"o}rger, Erich Gr{\"a}del, and Yuri Gurevich.
\newblock \emph{The Classical Decision Problem}.
\newblock Perspectives in Mathematical Logic. Springer, 1997.
\newblock \doi{10.1007/978-3-642-59207-2}.

\bibitem[Chv{\'a}tal and Lov{\'a}sz(1974)]{chvatal1974semikernel}
V{\'a}clav Chv{\'a}tal and L{\'a}szl{\'o} Lov{\'a}sz.
\newblock Every directed graph has a semi-kernel.
\newblock In \emph{Hypergraph Seminar}, volume 411 of \emph{Lecture Notes in
  Mathematics}, pages 175--175. Springer, 1974.
\newblock \doi{10.1007/BFb0066192}.

\bibitem[Durand and More(2002)]{durandmore2002}
Arnaud Durand and Malika More.
\newblock {Nonerasing, Counting, and Majority over the Linear Time Hierarchy}.
\newblock \emph{Information and Computation}, 174\penalty0 (2):\penalty0
  132--142, 2002.
\newblock \doi{10.1006/inco.2001.3084}.

\bibitem[Durand and Olive(2006)]{durandolive2006}
Arnaud Durand and Fr{\'e}d{\'e}ric Olive.
\newblock {First-Order Queries over One Unary Function}.
\newblock In \emph{Computer Science Logic, {CSL} 2006}, volume 4207 of
  \emph{Lecture Notes in Computer Science}, pages 334--348. Springer, 2006.
\newblock \doi{10.1007/11874683_22}.

\bibitem[Durand et~al.(1998)Durand, Fagin, and Loescher]{durand1998}
Arnaud Durand, Ronald Fagin, and Bernd Loescher.
\newblock Spectra with only unary function symbols.
\newblock In \emph{Computer Science Logic, {CSL} '97}, volume 1414 of
  \emph{Lecture Notes in Computer Science}, pages 189--202. Springer, 1998.
\newblock \doi{10.1007/BFb0028015}.

\bibitem[Flajolet and Sedgewick(2009)]{flajolet2009}
Philippe Flajolet and Robert Sedgewick.
\newblock \emph{{Analytic Combinatorics}}.
\newblock Cambridge University Press, 2009.
\newblock \doi{10.1017/CBO9780511801655}.

\bibitem[Gessel(1990)]{gessel1990}
Ira~M. Gessel.
\newblock {Symmetric functions and P-recursiveness}.
\newblock \emph{Journal of Combinatorial Theory, Series A}, 53\penalty0
  (2):\penalty0 257--285, 1990.
\newblock \doi{10.1016/0097-3165(90)90060-A}.

\bibitem[Goulden and Jackson(1986)]{goulden1986}
Ian~P. Goulden and David~M. Jackson.
\newblock {Labelled graphs with small vertex degrees and P-recursiveness}.
\newblock \emph{SIAM Journal on Algebraic and Discrete Methods}, 7\penalty0
  (1):\penalty0 60--66, 1986.

\bibitem[Gr{\"a}del et~al.(1997{\natexlab{a}})Gr{\"a}del, Kolaitis, and
  Vardi]{gradel1997fo2}
Erich Gr{\"a}del, Phokion~G. Kolaitis, and Moshe~Y. Vardi.
\newblock On the decision problem for two-variable first-order logic.
\newblock \emph{Bulletin of Symbolic Logic}, 3\penalty0 (1):\penalty0 53--69,
  1997{\natexlab{a}}.
\newblock \doi{10.2307/421196}.

\bibitem[Gr{\"a}del et~al.(1997{\natexlab{b}})Gr{\"a}del, Otto, and
  Rosen]{gradel1997c2}
Erich Gr{\"a}del, Martin Otto, and Eric Rosen.
\newblock Two-variable logic with counting is decidable.
\newblock In \emph{Proceedings of the 12th Annual {IEEE} Symposium on Logic in
  Computer Science ({LICS})}, pages 306--317, 1997{\natexlab{b}}.
\newblock \doi{10.1109/LICS.1997.614957}.

\bibitem[Grandjean(1990)]{grandjean1990}
Etienne Grandjean.
\newblock First-order spectra with one variable.
\newblock \emph{Journal of Computer and System Sciences}, 40\penalty0
  (2):\penalty0 136--153, 1990.
\newblock \doi{10.1016/0022-0000(90)90009-A}.

\bibitem[Klazar(2018)]{klazar2018}
Martin Klazar.
\newblock {What is an answer? Remarks, results and problems on PIO formulas in
  combinatorial enumeration, part I}.
\newblock arXiv:1808.08449, 2018.

\bibitem[Kuang et~al.(2026)Kuang, Ku{\v{z}}elka, Wang, and Wang]{kuang2026}
Qipeng Kuang, Ond{\v{r}}ej Ku{\v{z}}elka, Yuanhong Wang, and Yuyi Wang.
\newblock Bridging weighted first order model counting and graph polynomials.
\newblock In \emph{34th {EACSL} Annual Conference on Computer Science Logic
  ({CSL} 2026)}, volume 363 of \emph{Leibniz International Proceedings in
  Informatics (LIPIcs)}, pages 7:1--7:22, 2026.
\newblock \doi{10.4230/LIPIcs.CSL.2026.7}.

\bibitem[Kuusisto and Lutz(2018)]{kuusisto2018}
Antti Kuusisto and Carsten Lutz.
\newblock Weighted model counting beyond two-variable logic.
\newblock In \emph{Proceedings of the 33rd Annual {ACM}/{IEEE} Symposium on
  Logic in Computer Science ({LICS})}, pages 619--628, 2018.
\newblock \doi{10.1145/3209108.3209168}.

\bibitem[Ku{\v{z}}elka(2021)]{kuzelka2021}
Ond{\v{r}}ej Ku{\v{z}}elka.
\newblock Weighted first-order model counting in the two-variable fragment with
  counting quantifiers.
\newblock \emph{Journal of Artificial Intelligence Research}, 70:\penalty0
  1281--1307, 2021.
\newblock \doi{10.1613/jair.1.12320}.

\bibitem[Libkin(2004)]{libkin2004}
Leonid Libkin.
\newblock \emph{{Elements of Finite Model Theory}}.
\newblock Texts in Theoretical Computer Science. An EATCS Series. Springer,
  2004.
\newblock \doi{10.1007/978-3-662-07003-1}.

\bibitem[Malhotra et~al.(2025)Malhotra, Bizzaro, and Serafini]{malhotra2025}
Sagar Malhotra, Davide Bizzaro, and Luciano Serafini.
\newblock Lifted inference beyond first-order logic.
\newblock \emph{Artificial Intelligence}, 342:\penalty0 104310, 2025.
\newblock \doi{10.1016/j.artint.2025.104310}.

\bibitem[Mortimer(1975)]{mortimer1975}
Michael Mortimer.
\newblock On languages with two variables.
\newblock \emph{Zeitschrift f{\"u}r Mathematische Logik und Grundlagen der
  Mathematik}, 21\penalty0 (1):\penalty0 135--140, 1975.
\newblock \doi{10.1002/malq.19750210118}.

\bibitem[Pacholski et~al.(1997)Pacholski, Szwast, and Tendera]{pacholski1997}
Leszek Pacholski, Wies{\l}aw Szwast, and Lidia Tendera.
\newblock Complexity of two-variable logic with counting.
\newblock In \emph{Proceedings of the 12th Annual {IEEE} Symposium on Logic in
  Computer Science ({LICS})}, pages 318--327, 1997.
\newblock \doi{10.1109/LICS.1997.614958}.

\bibitem[Pak(2018)]{pak2018}
Igor Pak.
\newblock Complexity problems in enumerative combinatorics.
\newblock In \emph{Proceedings of the International Congress of Mathematicians
  2018}, volume~4, pages 3299--3322. World Scientific, 2018.

\bibitem[Pak(2024)]{pak2024}
Igor Pak.
\newblock What is a combinatorial interpretation?
\newblock In \emph{Open Problems in Algebraic Combinatorics}, volume 110 of
  \emph{Proceedings of Symposia in Pure Mathematics}, pages 191--260. American
  Mathematical Society, 2024.

\bibitem[P{\'o}lya and Read(1987)]{polya1987}
George P{\'o}lya and Ronald~C. Read.
\newblock \emph{Combinatorial Enumeration of Groups, Graphs, and Chemical
  Compounds}.
\newblock Springer, 1987.

\bibitem[Pratt-Hartmann(2005)]{pratt2005}
Ian Pratt-Hartmann.
\newblock Complexity of the two-variable fragment with counting quantifiers.
\newblock \emph{Journal of Logic, Language and Information}, 14:\penalty0
  369--395, 2005.
\newblock \doi{10.1007/s10849-005-5791-1}.

\bibitem[Rabin(1969)]{rabin1969}
Michael~O. Rabin.
\newblock Decidability of second-order theories and automata on infinite trees.
\newblock \emph{Transactions of the American Mathematical Society},
  141:\penalty0 1--35, 1969.
\newblock \doi{10.2307/1995086}.

\bibitem[Stanley(2012)]{stanley2012}
Richard~P. Stanley.
\newblock \emph{Enumerative Combinatorics}, volume~1.
\newblock Cambridge University Press, {2nd} edition, 2012.

\bibitem[T{\'o}th and Ku{\v{z}}elka(2023)]{toth2023}
Jan T{\'o}th and Ond{\v{r}}ej Ku{\v{z}}elka.
\newblock Lifted inference with linear order axiom.
\newblock In \emph{Proceedings of the 37th {AAAI} Conference on Artificial
  Intelligence ({AAAI})}, 2023.

\bibitem[Valiant(1979)]{valiant1979}
Leslie~G. Valiant.
\newblock The complexity of enumeration and reliability problems.
\newblock \emph{SIAM Journal on Computing}, 8\penalty0 (3):\penalty0 410--421,
  1979.
\newblock \doi{10.1137/0208032}.

\bibitem[{van Bremen} and Ku{\v{z}}elka(2021)]{vanbremen2021}
Timothy {van Bremen} and Ond{\v{r}}ej Ku{\v{z}}elka.
\newblock Lifted inference with tree axioms.
\newblock In \emph{Proceedings of the 18th International Conference on
  Principles of Knowledge Representation and Reasoning ({KR})}, pages 599--608,
  2021.
\newblock \doi{10.24963/kr.2021/57}.

\bibitem[{Van den Broeck}(2011)]{vandenbroeck2011}
Guy {Van den Broeck}.
\newblock On the completeness of first-order knowledge compilation for lifted
  probabilistic inference.
\newblock In \emph{Advances in Neural Information Processing Systems 24
  ({NIPS})}, pages 1386--1394, 2011.

\bibitem[{Van den Broeck} et~al.(2014){Van den Broeck}, Meert, and
  Darwiche]{vandenbroeck2014}
Guy {Van den Broeck}, Wannes Meert, and Adnan Darwiche.
\newblock Skolemization for weighted first-order model counting.
\newblock In \emph{Proceedings of the 14th International Conference on
  Principles of Knowledge Representation and Reasoning ({KR})}, 2014.

\bibitem[Wilf(1982)]{wilf1982}
Herbert~S. Wilf.
\newblock What is an answer?
\newblock \emph{The American Mathematical Monthly}, 89\penalty0 (5):\penalty0
  289--292, 1982.

\bibitem[Woods(1997)]{woods1997}
Alan~R. Woods.
\newblock Counting finite models.
\newblock \emph{The Journal of Symbolic Logic}, 62\penalty0 (3):\penalty0
  925--949, 1997.
\newblock \doi{10.2307/2275580}.

\bibitem[Zou et~al.(2025)Zou, Mai, Zhang, Wang, Ku{\v{z}}elka, Wang, and
  Chang]{zou2025}
Kuncheng Zou, Jiahao Mai, Yonggang Zhang, Yuyi Wang, Ond{\v{r}}ej
  Ku{\v{z}}elka, Yuanhong Wang, and Yi~Chang.
\newblock Faster lifting for ordered domains with predecessor relations.
\newblock In \emph{Proceedings of the 28th European Conference on Artificial
  Intelligence ({ECAI})}, 2025.

\end{thebibliography}

\end{document}